\documentclass[11pt]{article}
\usepackage{amsmath}
\usepackage{amssymb}
\usepackage{amsthm}
\usepackage{graphicx} 
\usepackage{fullpage}
\usepackage{hyperref}
\usepackage{xcolor}
\usepackage{multicol}
\usepackage{url}
\usepackage[utf8]{inputenc}
\usepackage{enumitem}

\title{\textbf{Redundancy Is All You Need \\ (for CSP Sparsification)}\thanks{An extended abstract of this work appeared in STOC 2025~\cite{brakensiek2024Redundancy}.}}
\author{Joshua Brakensiek\thanks{Department of Electrical Engineering and Computer Sciences, University of California, Berkeley. Contact: \href{mailto:josh.brakensiek@berkeley.edu}{josh.brakensiek@berkeley.edu}} \and Venkatesan Guruswami\thanks{Simons Institute for the Theory of Computing, and the Departments of EECS and Mathematics, University of California, Berkeley. Contact: \href{mailto:venkatg@berkeley.edu}{venkatg@berkeley.edu}}}

\date{}

\newcommand{\card}[1]{\lvert #1\rvert}

\newcommand{\Ham}{\operatorname{Ham}}
\newcommand{\Hom}{\operatorname{Hom}}
\newcommand{\supp}{\operatorname{supp}}

\newcommand{\CSP}{\operatorname{CSP}}

\newcommand{\NAE}{\operatorname{NAE}}
\newcommand{\Span}{\operatorname{span}}

\newcommand{\tLIN}{\operatorname{3LIN}}

\newcommand{\VC}{\operatorname{VC}}
\newcommand{\sat}{\operatorname{sat}}

\newcommand{\ex}{\operatorname{ex}}
\newcommand{\poly}{\operatorname{poly}}

\newcommand{\cP}{\mathcal P}
\newcommand{\cQ}{\mathcal Q}

\newcommand{\cF}{\mathcal F}

\newcommand{\cR}{\mathcal R}

\newcommand{\N}{\mathbb{N}}
\newcommand{\F}{\mathbb{F}}
\newcommand{\B}{\mathbb{B}}
\newcommand{\cC}{\mathcal C}
\newcommand{\cD}{\mathcal D}
\newcommand{\AND}{\operatorname{AND}}

\newcommand{\OR}{\operatorname{OR}}
\newcommand{\CUT}{\operatorname{CUT}}

\newcommand{\eps}{\varepsilon}

\newcommand{\val}{\operatorname{val}}

\newcommand{\NRD}{\operatorname{NRD}}
\newcommand{\SPR}{\operatorname{SPR}}
\newcommand{\CL}{\operatorname{CL}}

\newcommand{\E}{\mathbb E}
\newcommand{\one}{{\bf 1}}
\newcommand{\wt}{\operatorname{wt}}
\newcommand{\argmin}{\operatorname{argmin}}
\newcommand{\argmax}{\operatorname{argmax}}
\newcommand{\spor}{\Span_{\OR}}
\newcommand{\LIN}{\operatorname{LIN}}
\newcommand{\POLY}{\operatorname{POLY}}
\newcommand{\BCK}{\operatorname{BCK}}
\newcommand{\BCKp}{\operatorname{BCK}^{+}}

\newcommand{\R}{\mathbb{R}}
\newcommand{\Z}{\mathbb{Z}}

\newcommand{\wSPR}{\operatorname{wSPR}}
\newcommand{\type}{\operatorname{type}}

\newtheorem{theorem}{Theorem}
\numberwithin{theorem}{section}
\newtheorem{lemma}[theorem]{Lemma}
\newtheorem{claim}[theorem]{Claim}
\newtheorem{proposition}[theorem]{Proposition}
\newtheorem{corollary}[theorem]{Corollary}

\newtheorem{question}[theorem]{Question}

\theoremstyle{definition}
\newtheorem{definition}[theorem]{Definition}
\newtheorem{remark}[theorem]{Remark}
\newtheorem{example}[theorem]{Example}
\makeatletter
\renewcommand{\paragraph}{%
  \@startsection{paragraph}{4}%
  {\z@}{6pt \@plus 1pt \@minus 1pt}{-5pt}%
  {\normalfont\normalsize\bfseries}%
}
\makeatother

\begin{document}

\maketitle
\thispagestyle{empty}

\begin{abstract}
The seminal work of Bencz{\' u}r and Karger demonstrated cut sparsifiers of near-linear size, with several applications throughout theoretical computer science. Subsequent extensions have yielded sparsifiers for hypergraph cuts and more recently linear codes over Abelian groups. A decade ago, Kogan and Krauthgamer asked about the sparsifiability of arbitrary constraint satisfaction problems (CSPs). For this question, a trivial lower bound is the size of a \emph{non-redundant} CSP instance, which admits, for each constraint, an assignment satisfying only that constraint (so that no constraint can be dropped by the sparsifier). For instance, for graph cuts, spanning trees are non-redundant instances.

\smallskip 
Our main result is that redundant clauses are sufficient for sparsification: for any CSP predicate $R$, every unweighted instance of $\CSP(R)$ has a sparsifier of size at most its non-redundancy (up to polylog and $1/\epsilon$ factors). For weighted instances, we similarly pin down the sparsifiability to the so-called chain length of the predicate. These results precisely determine the extent to which any CSP can be sparsified.

\smallskip
Our result is established in the general setting of non-linear codes, or equivalently set families, yielding a VC-type theorem for multiplicative error approximation. This unifies and extends known sparsification results for graph cuts, linear codes, and broad CSP classes. 
A key technical ingredient in our work is a novel application of the entropy method from Gilmer's recent breakthrough on the union-closed sets conjecture. 

\smallskip
As an immediate consequence of our main theorem, a number of results in the non-redundancy literature immediately extend to CSP sparsification.
We also contribute new techniques for understanding the non-redundancy of CSP predicates. 
By adapting methods from the matching vector codes literature in coding theory, we are able to construct an explicit predicate whose non-redundancy lies between $\Omega(n^{1.5})$ and $\widetilde{O}(n^{1.6})$, the first example with a provably non-integral exponent. The tools we introduce in this paper, including a conditional version of non-redundancy, have spurred several follow-ups on non-redundancy as well as new connections to extremal combinatorics, sparsification, and streaming.
\end{abstract}

\clearpage

\setcounter{tocdepth}{2} 
\tableofcontents
\thispagestyle{empty}
\pagebreak
\setcounter{page}{1}

\section{Introduction} \label{sec:intro}

The broad goal in sparsification is to replace an object by a more compact surrogate, typically a carefully chosen subsample, that preserves the behavior of the object under some metric of interest. For instance, for preserving cuts in undirected graphs, the influential works of Karger~\cite{DBLP:conf/soda/Karger93} and Bencz\'{u}r and Karger~\cite{DBLP:conf/stoc/BenczurK96} showed that every graph has an edge-weighted subgraph with near-linear number of edges that preserves the value of all (edge) cuts up to a $(1\pm \epsilon)$ multiplicative factor. These papers have had a substantial impact in shaping the last thirty years of work in areas such as spectral sparsifiers~\cite{spielman2011Spectral,DBLP:journals/siamcomp/BatsonSS12,lee2018constructing}, clustering~\cite{kannan2004clusterings,DBLP:conf/sigmod/SatuluriPR11}, hypergraph sparsifiers~\cite{kogan2015, chen2020linear,KKTY2021,kenneth2023cut,khanna2024optimal}, linear solvers~\cite{SpielmanT04,vishnoi2013Lx,koutis2014approaching}, convex optimization \cite{lee2014path,DBLP:journals/jacm/AroraK16,DBLP:books/daglib/0040632}, sketching/streaming algorithms~\cite{ahn2009Graph,ahn2012graph,DBLP:conf/soda/AhnGM12,andoni2016Sketching,McGregor14,kapralov2017single,braverman2021adversarial}, max-flow/min-cut algorithms~\cite{leighton1999multicommodity,christiano2011Electrical,kelner2014almost,DBLP:conf/focs/ChenKLPGS22}, machine learning~\cite{luo2021Learning,chen2022Graph,zhang2023Survey,guo2024Datacentric}, submodular functions~\cite{kenneth2023cut,schwartzman2024mini,rafiey2024decomposable,quanrud2024Quotient}, differential privacy~\cite{blocki2012johnson,arora2019Differentially}, PageRank~\cite{Chung14}, and even theoretical physics~\cite{hadjighasem2016Spectralclustering,vanchurin2018Information,taira2022Networkbased}, among many other works. 

Among the multiple exciting dimensions in which cut sparsification has been generalized, we now highlight two which form the backdrop for our work. Note that the graph cut problem can be modeled by the arity-two Boolean constraint $x + y =1 \pmod 2$. One can thus generalize cut sparsification by allowing for arbitrary constraints (of any arity over some finite domain) as considered in the field of constraint satisfaction problems (CSPs), leading to \emph{CSP sparsification.} This direction was proposed by Kogan and Krauthgamer~\cite{kogan2015} in their work on hypergraph cut sparsifiers, where the not-all-equal constraint captures hypergraph cut. As as special case, arbitrary binary CSPs (where each constraint has two variables) were studied in \cite{filtser2017Sparsification} for the Boolean domain and in \cite{butti2020Sparsificationa} for general domains, leading to a dichotomy: either near-linear sized sparsifiers exist, or no improvement over quadratic is possible.

In another direction, one can instead look toward more general structures to sparsify. For instance, a recent line of work by Khanna, Putterman, and Sudan turned toward sparsifying linear codes~\cite{khanna2024Code}, or more generally subgroups of powers of Abelian groups \cite{khanna2024Characterizations}. Beyond being algorithmically efficient~\cite{khanna2024Characterizations}, these structural results have led to exciting new results in CSP sparsification by constructing optimal sparsifiers when the constraints can be embedded into linear/Abelian equations.

In this work, we obtain sparsifiers encompassing both these generalizations via a unified approach to sparsification of \emph{non-linear} codes. The resulting sparsifiers for CSPs have \emph{optimal asymptotic size} up to polylogarithmic factors, for \emph{every choice of predicate} defining the CSP. In other words, we \emph{pinpoint the optimal extent to which an arbitrary CSP can be sparsified.}\footnote{In this work we focus on the \emph{existence} of sparsifiers, which is already highly non-trivial (e.g., \cite{khanna2024Code,butti2020Sparsificationa} are also non-algorithmic). Future directions (and barriers) for algorithmic aspects are briefly discussed in Sections~\ref{subsec:open-questions} and~\ref{sec:conclusion}.}

\subsection{Non-linear Code Sparsification}

We first state our result for codes as it is very general and crisply stated, and then turn to the consequences and further new results for CSPs. For a non-linear code\footnote{Such a problem can be equivalently phrased as the sparsification of \emph{set families}. However, we find the code framework notationally more convenient for our applications.} $C \subseteq \{0,1\}^m$, an $\eps$-sparsifier (for a parameter $\eps \in (0,1)$) is a weight function $w : [m] \to \R_{\ge 0}$ such that for every codeword $c$, adding up the weights of its nonzero positions, i.e., $\sum_i w(i) c_i$, is an accurate estimate of the Hamming weight of $c$ (i.e., $\Ham(c) :=\sum_i c_i$) to within a $(1 \pm \eps)$ multiplicative factor (Definition~\ref{def:SPR-code}). The goal is to minimize the support of $w$ (i.e., the number of nonzero entries $w(i)$), and the minimum value is called $\eps$-sparsifiability of $C$ and is denoted $\SPR(C,\eps)$. One of our main results is an upper bound on the sparsifiability in terms of a natural combinatorial parameter of the code called its \emph{non-redundancy} $\NRD(C)$, defined as follows: 
\begin{quote}
    $\NRD(C)$ is the size of the largest subset of indices $I \subseteq [m]$ such that for each $i \in I$, there is a codeword $c \in C$ with $c_i=1$ and $c_{i'}=0$ for $i' \in I \setminus \{i\}$.
\end{quote}
In other words, if we imagine the code as a matrix whose rows are codewords, its non-redundancy is largest square submatrix which is a permutation matrix.

Our main result can then be stated compactly as follows.
\begin{theorem}[Main]\label{thm:main-code}
For all $C \subseteq \{0,1\}^m$ and $\eps \in (0,1)$,
\[
  \SPR(C, \eps) = O(\NRD(C) (\log m)^6 / \eps^2).
\]
\end{theorem}

This theorem is highly versatile and can cover the many motivating examples for studying such sparsficiation. For the sparsification of constraint satisfaction problems, we can let the codewords in our code $C$ correspond to indicators vectors of assignments to clauses; we discuss this in more detail in Section~\ref{subsec:csp-spar}. For sparsification of linear codes and more broadly group codes, Theorem~\ref{thm:main-code} can be used to give essentially tight bounds, matching the results of \cite{khanna2024Code,khanna2024Characterizations} (modulo the efficiency of \cite{khanna2024Characterizations}); see Section~\ref{sec:spr-app} for further discussion and connections.

\medskip\noindent\textbf{Interpreting Theorem~\ref{thm:main-code} as a multiplicative-error VC theorem.}
If the goal is to estimate the Hamming weight of every $c \in C$ within an  \emph{additive} error of $\pm \eps m$, then it well-known that the sample complexity is characterized by the \emph{VC-dimension} of $C$. Recall that the VC-dimension of $C \subseteq \{0,1\}^m$ is largest subset of indices $I \subseteq [m]$ such that every pattern of $0$'s and $1$'s occurs amongst the codewords of $C$ in the positions $I$, i.e., $C_{|I} = \{0,1\}^{|I|}$.\footnote{In the set family view, for $\mathcal{F} \subseteq 2^{U}$, this is the largest subset of the universe $U$ that is \emph{shattered} by the family $\mathcal{F}$.}
A random subset $I \subset [m]$ of size $O(\text{VC-dim}(C)/\eps^2)$~\cite{VC71,AnthonyBartlett1999} is known to  
satisfy $\frac{1}{|I|} (\sum_{i \in I} c_i ) \in \tfrac{\Ham(c)}{m} \pm \eps $ 
for every $c \in C$. Conversely, one needs to sample at least $\Omega_\eps(\text{VC-dim}(C))$ coordinates to obtain $\pm \eps$ additive error.

With this view, Theorem~\ref{thm:main-code} in effect says that the non-redundancy $\NRD(C)$ plays the role of the VC dimension of $C$ when the goal is to estimate the Hamming weight of the set $c \in C$ within a \emph{multiplicative} $(1 \pm \eps)$ error. This is a harder goal, since one has to be pay special attention to low-weight codewords (uniform sampling will clearly fail in general).
Rather surprisingly, when one performs the sampling carefully and re-weights the sampled coordinates appropriately, Theorem~\ref{thm:main-code} says that multiplicative error sparsification \emph{is} in fact possible for any code, albeit with sparsity roughly the non-redundancy of the code.

The NRD of a code can in general be much larger than its VC-dimension. In fact, it is easily seen that $\NRD(C)$ is precisely the VC-dimension of the \emph{union-closure} of $C$ (see Proposition~\ref{prop:nrd-vc}). This connection to union-closed families plays a crucial role in the proof of Theorem~\ref{thm:main-code}. See the technical overview (Section~\ref{subsec:technical-overview}) for more details, including discussion of a significantly simpler $\widetilde{O}_{\eps}(\NRD(C)\log|C|)$-sized sparsifier.

\subsection{CSP Sparsification}\label{subsec:csp-spar}
We now turn to (unweighted\footnote{The weighted case is discussed in Section~\ref{subsec:weighted}.}) CSP sparsification. For a relation $R \subseteq D^r$ of arity $r$ over a finite domain $D$, an instance $\Psi$ of the $\CSP(R)$ problem consists a variable set $X$ and a constraint set $Y \subseteq X^r$. An assignment $\sigma : X \to D$ satisfies a constraint $y=(x_1,x_2,\dots,x_r) \in Y$ if $(\sigma(x_1),\sigma(x_2),\dots,\sigma(x_r)) \in R$. The value $\val(\Psi,\sigma)$ of an assignment $\sigma$ is the number of constraints $y \in Y$ that it satisfies. Similarly, for a weight function $w : Y \to \R_{\ge 0}$, the weighted value $\val(\Psi,w,\sigma)$ is the sum of weights $w(y)$ of all constraints $y \in Y$ that $\sigma$ satisfies.  The goal in CSP sparsification is to output a weight function $w : Y \to \R_{\ge 0}$ of small support, such that for \emph{every} assignment $\sigma : X \to D$, 
\[ (1-\eps) \val(\Psi,\sigma) \le \val(\Psi,w,\sigma) \le (1+\eps) \val(\Psi,\sigma) \ , \]
and minimum such support size is denoted $\SPR(\Psi,\eps)$. 

The $\eps$-sparsifiability of the relation $R \subseteq D^r$, as a function of number of variables, is defined to the maximum (i.e., worst-case) value of $\SPR(\Psi,\eps)$ over all $n$-variables instances $\Psi$ of $\CSP(R)$. We denote it by $\SPR(R,n,\eps)$ and it is the chief object of our study. Note that this is for the unweighted case, see Section~\ref{subsec:weighted} how this result can be (tightly) applied to the weighted case.

Let us note an obvious obstruction to sparsification. Suppose we have an instance $\Psi=(X,Y)$ of $\CSP(R)$ such that for each of its constraints $y \in Y$, there is an assignment $\sigma_y : X \to D$ that satisfies only $y$ and no other constraint. Then clearly $\Psi$ cannot be sparsified at all---dropping any constraint $y$ would make the value of $\sigma_y$ drop from $1$ to $0$. We call such an instance a \emph{non-redundant} instance of $\CSP(\overline{R})$, where $\overline{R} = D^r \setminus R$ (cf., \cite{bessiere2013constraint,bessiere2020Chain}).\footnote{We use $\overline{R}$ rather than $R$ due to the conventions of each community. See Remark~\ref{rem:backwards} for deeper technical reasons.}  As introduced by Bessiere, Carbonnel, and Katsirelos~\cite{bessiere2020Chain}, we denote the size of the largest such non-redundant instance of $\CSP(\overline{R})$ on $n$-variables by $\NRD(\overline{R},n)$ and call it the non-redundancy of $\overline{R}$.
Thus a trivial lower bound on sparsifiability of $\CSP(R)$, regardless of the choice of $\eps \in (0,1)$, is given by 
\begin{align}
\label{eq:nrd-lb}
\SPR(R,n,\eps) \ge \NRD(\overline{R},n) \ , 
\end{align}
and this holds even if the goal is merely to preserve which assignments have nonzero value.

Rather remarkably,  this simplistic lower bound can be met and one can sparsify all the way down to $\NRD(\overline{R}, n)$ times polylogarithmic factors!  In fact, this turns out to be an easy corollary of Theorem~\ref{thm:main-code}. One can associate a canonical code $C_\Psi \subseteq \{0,1\}^Y$ with any $\CSP(R)$ instance $\Psi = (X,Y)$ whose codewords $c_\sigma$  correspond to the assignments $\sigma : X \to D$, and $c_{\sigma,y}$ is $1$ precisely when $\sigma$ satisfies $y$. It is easy to check that CSP sparsification of $\Psi$ reduces to code sparsification of $C_\Psi$, and the non-redundancy of $C$ equals the size of the largest non-redundant sub-instance of $\Psi$ (viewed as an instance of $\CSP(\overline{R})$). Combining Theorem~\ref{thm:main-code} and \eqref{eq:nrd-lb}, we therefore have our main result pinning down the sparsifiability of every CSP up to polylogarithmic factors.

\begin{theorem}\label{thm:main}
For every nonempty $R \subsetneq D^r$ and $\eps \in (0,1)$, we have that
\[
  \NRD(\overline{R}, n) \le \SPR(R, n, \eps) \le O(\NRD(\overline{R}, n)(r\log n)^6/\eps^2).
\]
\end{theorem}

\subsection{Weighted CSP Sparsification}\label{subsec:weighted}

The discussion so far has focused on unweighted CSP instances, and we now shift our focus to the weighted case, where each constraint of $Y$ comes with a weight. We also get a tight characterization of weighted CSP sparsifiablity, in terms of a parameter called the \emph{chain length}, which was defined by Lagerkvist and Wahlstr{\"o}m~\cite{lagerkvist2017Kernelization,lagerkvist2020Sparsification} in the context of \emph{CSP kernelization} and later utilized by Bessiere, Carbonnel, and Katsirelos~\cite{bessiere2020Chain} in the context of learning CSPs in a certain query model (see Section~\ref{subsec:related-work} for more details on these connections). As before, the result is obtained in the setting of weighted non-linear codes, with the consequence for weighted CSPs being an easy corollary. We just state the result for codes here (see Section~\ref{sec:spr-weighted} for the full treatment of weighted CSPs).

For weighted sparsification of a code $C \subseteq \{0,1\}^m$, we might have an arbitrary input weighting $\zeta : [m] \to \R_{\ge 0}$ of its coordinates, and we must find a sparsifier $\widetilde{w}: [n] \to \R_{\ge 0}$ of low support that sparsifies $C$ with respect to the weighting $\zeta$, i.e., $\langle \widetilde{w},c\rangle \in (1 \pm \eps) \langle \zeta,c\rangle$. The minimum possible support of sparsifiers over all weightings $\zeta$ is called the weighted $\eps$-sparsity $\wSPR(C,\eps)$.

Now we define chain length. If we line up the codewords of $C$ as rows of an $|C| \times m$ matrix and allow arbitrary column permutations, the chain length of $C$, denoted $\CL(C)$, is the dimension of the largest upper triangular square submatrix with $1$'s on the diagonal.\footnote{In this view $\NRD(C)$ is the dimension of the largest identity submatrix, so clearly $\NRD(C) \le \CL(C)$. The quantity $\CL(C)$ was called \emph{visible rank} in \cite{alrabiahG21} and served as a field independent lower bound on the rank of $C$.}

In our main result for the weighted setting, we pin the sparsifiability of a weighted code to its chain length. 
Note that in the weighted case $\CL(C)$ is also a lower bound.

\begin{theorem}\label{thm:wspr-code-intro}
For all $C \subseteq \{0,1\}^m$ and $\eps \in (0,1)$, we have 
\[
    \CL(C) \le \wSPR(C, \eps) = O(\CL(C)(\log m)^6/\eps^2).
\]
\end{theorem}

The upper bound proceeds by using Theorem~\ref{thm:main-code} as a black-box together with a geometric weight bucketing technique from \cite{khanna2024Characterizations}. The lower bound proceeds by applying an exponential sequence of weights to the indices $i_1, \hdots, i_{\CL{(C)}} \in [m]$ forming a maximal chain. Of note, if for a particular set of weights, the ratio between maximum and minimal weights is $\lambda \ll \exp(\CL(C) / \NRD(C))$, we get a sharper upper bound of $\widetilde{O}_{\eps}(\NRD(C)\log \lambda)$ (see Corollary~\ref{cor:wspr-improved}). 

\subsection{Non-redundancy and Chain Length of Specific Relations}

As shown in our main results, to understand the sparsifiability of various CSPs, we must understand their non-redundancy and chain length. These quantities are readily computed in some simple cases. For example, for the relation $\OR_r := D^r \setminus \{0^r\}$, we have that $\NRD(\OR_r,n) = \CL(\OR_r, n) = \Theta(n^r)$. Indeed $Y = \binom{X}{r}$ is a non-redundant instance because setting all but $r$ variables to $1$ fails to satisfy exactly that $r$-tuple (see \cite{filtser2017Sparsification,carbonnel2022Redundancy,khanna2024Characterizations}).
When $R$ is affine, $\NRD(R,n) = \CL(R, n) = \Theta(n)$, and when $R$ is defined as the zero set of a degree $k$ polynomial, $\NRD(R,n) = \CL(R, n) = O(n^k)$; these follow from simple rank arguments (e.g., \cite{lagerkvist2020Sparsification}). Via Theorem~\ref{thm:main}, these special cases (plus simple gadget reductions) already capture all the previously known upper and lower bounds for CSP sparsification (see Section~\ref{subsec:related-work} for more details on the CSP sparsification literature). 

Furthermore, there are also some non-trivial upper bounds known on NRD and CL in the literature, which we can now import to sparsifiability for free courtesy Theorem~\ref{thm:main}. 
For instance, the so-called Mal'tsev relations, which generalize affine predicates (i.e., cosets) over Abelian groups, have been shown to have $O_D(n)$ non-redundancy and chain length~\cite{lagerkvist2020Sparsification,bessiere2020Chain}, and therefore by Theorem~\ref{thm:main}  their complements have near-linear sparsifiability. Carbonnel~\cite{carbonnel2022Redundancy} showed that if $R$ is an arity $r$ relation that does not contain\footnote{See Theorem~\ref{thm:carbonnel} for a precise definition.} any copy of $\OR_r$, then $\NRD(R, n) \le O(n^{r-\delta_r})$ for $\delta_r = 2^{1-r}$ (the specific bound arises from a classic hypergraph Tur{\'a}n result~\cite{erdos1964extremal}). By Theorem~\ref{thm:main} this immediately implies $\SPR(\overline{R},n,\eps) \le \widetilde{O}_\eps(n^{r-\delta_r})$,  where $\widetilde{O}(\cdot)$ hides polylogarithmic factors in $n$,  yielding an $\Omega(n^r)$ vs $\widetilde{O}(n^{r-\delta_r})$ dichotomy for sparsification of arity $r$ CSPs. (This was known for the Boolean case~\cite{khanna2024Characterizations}; see the related work subsection.)

The non-redundancy (and chain length) of a relation can in general be difficult to estimate.
Thus while in principle Theorem~\ref{thm:main} pins down the sparisifiability of every CSP, for specific relations, it can still be non-trivial to actually determine the asymptotic behavior of its sparsifiability. Our next set of results makes  progress in this direction via novel methods to bound non-redundancy.

Given that the non-redundancy of linear predicates is easy to pin down, we consider a natural family of relations which are very close to being linear. Specifically, let 
$\tLIN_G = \{(x,y,z) \mid x+y+z =0\}$ over an Abelian group $G$, and consider
$\tLIN^*_G = \tLIN_G \setminus \{(0,0,0)\}$. (We pick arity $3$ since the arity $2$ case is already fully resolved~\cite{filtser2017Sparsification,butti2020Sparsificationa}.)
Being defined by a linear equation over an Abelian group, we already know that $\NRD(\tLIN_G,n) = \Theta_G(n)$.
However the non-redundancy of $\tLIN^*_G$ seems challenging to understand. Existing methods in the literature only yield 
 $\NRD(\tLIN^{*}_G, n) \in [\Omega_G(n), O_G(n^2)]$. 

 We introduce a new method for bounding the non-redundancy of predicates like $\tLIN^{*}_G$ by connecting them to the theory of \emph{matching vector (MV) families} \cite{yekhanin20083query, dvir2011Matching} that have been used in the construction of locally decodable codes. Exploiting this connection, we construct a non-redundant instance to establish that $\NRD(\tLIN_G,n) \ge \Omega(n^{1.5})$ for all Abelian groups of order $\ge 3$.  Adapting ideas from the analysis of MV families together with some combinatorial ideas, we also prove an upper bound $\NRD(\tLIN_{\Z/p\Z},n) = \widetilde{O}_p(n^{2-\eps_p})$ for $\eps_p = \tfrac{2}{2p-1}$ and $p$ prime.  Specializing for $p=3$, we have the following result, which also gives the first examples of relations whose non-redundancy and sparsifiability have a non-integral exponent. 

 \begin{theorem}
     \label{thm:3linG-intro}
 We have 
\begin{align*}
\NRD(\tLIN^*_{\Z/3\Z}, n) &\in [\Omega(n^{1.5}), \widetilde{O}(n^{1.6})], \ \ \ \text{ and }& 
\SPR(\overline{\tLIN^*_{\Z/3\Z}}, n, \eps) &\in [\Omega(n^{1.5}), \widetilde{O}(n^{1.6}/\eps^2)].
\end{align*}
\end{theorem}

We now transition to discussing the broader context of our work in the literature.

\subsection{Technical Overview}\label{subsec:technical-overview}

We next describe the primary techniques we use to prove Theorem~\ref{thm:main-code} and Theorem~\ref{thm:3linG-intro}. 

\paragraph{A Simple Sparsifier.} To begin, we discuss a warm-up version of Theorem~\ref{thm:main-code} which proves a weaker upper bound of $\SPR(C, \eps) \le \widetilde{O}_{\eps}(\NRD(C) \cdot \log |C|)$ (see Theorem~\ref{thm:nrd-spr-warmup}), which for CSPs corresponds to an extra factor of the number of variables $n$. The key technical insight (Lemma~\ref{lem:NRD-supp-bound}) is that for all $d \in [m]$, the set of codewords of $C$ with Hamming weight at most $d$ (denoted by $C_{\le d}$) has total support size at most $d \cdot \NRD(C)$. This can proved inductively by noticing that dropping a suitable non-redundant set of coordinates decreases the Hamming weight of every codeword of $C$ by at least one.

With this lemma, we can recursively construct a sparsifier as follows, similar to the divide-and-conquer framework in \cite{khanna2024Code,khanna2024Characterizations} for linear codes. Pick $d \approx \widetilde{\Theta}_{\eps}(\log |C|)$ and let $I \subseteq [m]$ be the support of $C_{\le d}$. Every $i \in I$ is given weight $1$ in our sparsifier. For the rest of $[m]$, let $J \subseteq [m] \setminus I$ be a subsample where each $i \in [m] \setminus I$ is kept independently with probability $1/3$. Using a standard Chernoff bound, we can show that with positive\footnote{We only need positive probability since we are focused on existence. This can easily be amplified to $1-1/m^{\Omega(1)}$ probability by making $d$ a factor of $\log m$ bigger. In applications to CSPs, the main algorithmic bottleneck is (approximately) finding $I$, which appears to be similar in difficulty to an open problem in CSP kernelization (see Section~\ref{subsec:open-questions}).} probability the following holds for all $c \in C$:
\[
    3\Ham(c|_{J}) + \Ham(c|_{I}) \in \left[1-\frac{\eps}{2\log_2 m}, 1+\frac{\eps}{2\log_2 m}\right] \cdot \Ham(c).
\]
By induction, we can find a $\widetilde{O}_{\eps'}(\NRD(C') \cdot \log |C'|)$ $\eps'$-sparsifier for $C' := C|_{J}$ with $\eps' := (1 - 1/\log_2 m) \eps$. Scaling this sparsifier by $3$ and adding weights for $I$ gives us an $\eps$-sparsifier of $C$.

\paragraph{Entropy-based Sparsification.} The key inefficiency of the $\widetilde{O}_{\eps}(\NRD(C) \cdot \log |C|)$ bound is that the use of Lemma~\ref{lem:NRD-supp-bound} is too conservative. For the purposes of this overview, assume that all codewords of $C$ have the same Hamming weight $d \approx \NRD(C)$ as that is is the most representative case. Naively, Lemma~\ref{lem:NRD-supp-bound} says we should set aside $d^2$ coordinates of $[m]$ to ``sparsify'' all codewords of weight $d$. However, we can give a heuristic argument that far fewer than $d^2$ of these potential coordinates contain useful information for our sparsifier.

Assume without loss of generality that the support of $C$ lies in $[d^2]$. For each $i \in [d^2]$, let $p_i$ be the probability that a codeword $c \in C$ selected uniformly at random has $c_i = 1$. Since each codeword of $C$ has Hamming weight $d$, we have that $p_1 + \cdots + p_{d^2} = d$. Thus, the average value of $p_i$ is $1/d$. Consider the case in which each $p_i = O(1/d)$. In particular, no coordinate is distinguishing itself as a ``must'' to add to the sparsifier. A priori, the size of $C$ may be $\exp(\widetilde{\Omega}(d))$, so we cannot immediately use Chernoff bounds to analyze a random subsampling of the coordinates.

To get around this issue, we need to prove a much stronger upper bound on the size of $C$, similar to Bencz\'{u}r and Karger's cut-counting bound~\cite{DBLP:conf/stoc/BenczurK96} and its adaptation to linear codes~\cite{khanna2024Code,khanna2024Characterizations}. However, we use an entirely new method for proving such bounds based on the entropy method Gilmer~\cite{gilmer2022constant} developed to prove the union-closed sets conjecture up to a constant factor. In our context, pick $t = \widetilde{\Theta}(d)$ and sample uniformly and independently $t$ codewords $c_1, \hdots, c_t \in C$. Let $c$ be the bitwise OR of these $t$ codewords, and let $\cD$ be the distribution of $c$ over $\{0,1\}^{d^2}$ (recall that the weight $d$ codewords are supported on $d^2$ coordinates). Since each $p_i = O(1/d)$, by adapting Gilmer's method (or more precisely, a refinement due to Sawin~\cite{sawin2023improved}), we can show the entropy of $\cD$ is at least $\widetilde{\Theta}(t) = \widetilde{\Theta}(d)$ times the entropy of the uniform distribution over $C$ (i.e., $\log_2|C|$)--a similar inequality appears in \cite{wakhare2024two}.

To apply this fact, observe that each sample of $\cD$ lies in the ``$\OR$-closure'' of $C$ (denoted by $\spor(C)$). As such, the entropy of $\cD$ is at most $\log \card{\spor(C)}$, which by the Sauer-Shelah-Peres lemma is at most (up to log factors) the VC dimension of $\spor(C)$. It is easily seen that the VC dimension of $\spor(C)$ equals the non-redundancy of $C$~\cite{bessiere2020Chain}. Therefore, we have proved that $\widetilde{\Theta}(t) \cdot \log_2(C) \le \widetilde{O}(\NRD(C))$. Since $t \approx d \approx \NRD(C)$, $C$ is actually at most quasipolynomial in size! Thus we can now use a Chernoff bound to prove that $C$ can be subsampled to $\widetilde{O}_{\eps}(d)$ coordinates while approximately preserving all Hamming weights.

Recall this discussion was purely about the ``uniform'' case $p_i = O(1/d)$. In general, we apply minimax theorem to prove the following ``skewed'' versus ``sparse'' dichotomy (see Proposition~\ref{prop:cover-sparse}): for every code $C$ and parameter choice $\theta \ge 1$ there is either a probability distribution $\cP$ over $C$ for which each coordinate equals $1$ with probability at most $1/\theta$ (i.e., $\cP$ is ``$\theta$-sparse''); or, there is a probability distribution $\cQ$ over the coordinates of $C$ such that for every (nonzero) $c \in C$, we have that $\cQ$'s measure of $\supp(c)$ is at least $1/\theta$ (i.e., $\cQ$ is a ``$\theta$-cover.'') For a suitable choice of $\theta$, we repeatedly apply Proposition~\ref{prop:cover-sparse} to recursively build the sparsifier: in the $\theta$-sparse case, we use the entropy method to prove that a ``small'' number of codewords of $C$ can be removed to put us in the $\theta$-cover case (see Lemma~\ref{lem:sparse-removal}); and in the $\theta$-cover case, we sample from the $\theta$-cover to get a coordinate to add to our sparsifier. This procedure culminates in showing that we can set aside $\widetilde{O}_{\eps}(\NRD(C))$ coordinates to have weight $1$ in our sparsifier with the remainder of the code being sufficiently sparse that subsampling can be used (Theorem~\ref{thm:NRD-decomp}). Note that the statement of Theorem~\ref{thm:NRD-decomp} resembles the analogous decompositions for linear codes \cite{khanna2024Code,khanna2024Characterizations}. However, their method found all the coordinates to set aside in ``one pass,'' whereas we iteratively understand the dense and sparse structure of our non-linear code. With Theorem~\ref{thm:NRD-decomp} in hand, we construct the sparsifier with a recursive argument similar to that of Theorem~\ref{thm:nrd-spr-warmup}.

As mentioned earlier, extended these ideas to weighted sparsification (Theorem~\ref{thm:wspr-code-intro}) is relatively straightforward. We adapt a weight-binning argument of \cite{khanna2024Characterizations} by essentially computing an (unweighted) sparsifier for each group of coordinates that is similar in weight (within $\poly(m)$).  We then analyze the aggregated size of these sparsifiers by comparing the sum of the non-redundancies of the groups of coordinates to the chain length of the code.

\paragraph{Connections to Matching Vector Families.} We now switch gears to briefly discussing the key ideas behind Theorem~\ref{thm:3linG-intro}. Let $G := \Z/3\Z$ and recall that $\tLIN_G = \{(x,y,z) \mid x+y+z =0\}$ and $\tLIN^*_G = \tLIN_G \setminus \{(0,0,0)\}$. It is well-known that since $\tLIN_G$ is an affine predicate, we have that $\NRD(\tLIN_G, n) = \Theta(n)$, which is much smaller than our bound on $\NRD(\tLIN^*_G, n)$. As such, we prove that to understand the asymptotics of $\NRD(\tLIN^*_G, n)$ it suffices to look at specially-structured non-redundant instances.

Recall that an instance $\Psi := (X,Y)$ of $\CSP(\tLIN^*_G)$ is non redundant if for every clause $y \in Y$ there is an assignment $\sigma_y$ which satisfies every clause of $\Psi$ except $y$. We show that with at most an additive $\Theta(n)$ change in size, we can assume that $\sigma_y$ maps $y$ to $(0,0,0)$. In other words, each $\sigma_y$ is a satisfying assignment to $\Psi$ when viewed as an instance of $\CSP(\tLIN_G)$ (see Proposition~\ref{prop:nrd-tri}). This idea of ``conditional'' non-redundancy abstracts and generalizes an approach from \cite{bessiere2020Chain}. 

Since the set of solutions to an instance of $\CSP(\tLIN_G)$ form a vector space (of some dimension, say $d$) over $\F_3$, we can think of each variable $x \in X$ of $\Psi$ as a vector $v_x \in \F_3^d$ and the assignments as linear maps on the vectors. Because we are studying satisfying assignment to $\CSP(\tLIN_G)$, these vectors are highly structured: for each $y := (x_1,x_2,x_3) \in Y$, we have that $v_{x_1} + v_{x_2} + v_{x_3} = 0$. Further, $\sigma_y$ can be viewed as a linear map taking each of $v_{x_1}, v_{x_2}, v_{x_3}$ to $0$, while mapping at least one vector in every other triple in $Y$ to a nonzero value. We call this family of vectors together with these assignments a \emph{$G$-ensemble} (Definition~\ref{def:G-ensemble}), and note that it bears a strong resemblance to matching vector families.

In particular, we adapt techniques used by Dvir, Gopalan, and Yekhanin~\cite{dvir2011Matching} for constraining the size of matching vector families to give nontrivial upper and lower bounds on the size of $G$-ensembles. For the lower bound (Theorem~\ref{thm:nrd-3-lin-lb}), we directly construct a non-redundant instance with $\Omega(n^{1.5})$ clauses. The proof is self-contained and elementary.

The upper bound (Theorem~\ref{thm:ub-1.6}) is slightly more technical. We break the proof into cases based on whether the embedding dimension $d$ of the vectors is small ($d = \widetilde{O}(n^{0.4})$) or large ($d = \widetilde{\Omega}(n^{0.4})$). For small $d$, we adapt the polynomial method used in \cite{dvir2011Matching} to prove there can be at most $O(d^4) = \widetilde{O}(n^{1.6})$ non-redundant clauses. On the other hand, when $d$ is large, we ignore the assignments $\sigma_y$ and use a careful induction (Lemma~\ref{lem:ind-ub}) to show that the geometry of the vectors imply that some $x \in X$ is a member of at most $\widetilde{O}(n/d) = \widetilde{O}(n^{0.6})$ clauses, thereby leading to a bound of at most $\widetilde{O}(n^{1.6})$ clauses total. Closing the gap between $\Omega(n^{1.5})$ and $\widetilde{O}(n^{1.6})$ for $\NRD(\tLIN_{\Z/3\Z}^*, n)$ is a tantalizing open question.

\subsection{Related Work}\label{subsec:related-work}

Our results and techniques have connections to many areas including computational complexity theory, extremal combinatorics, coding theory, and learning theory. We now give a general overview of these connections.

\paragraph{CSP Sparsification.} Since we already discussed the history of CSP sparsification, we give a comprehensive list of known results about CSP sparsification (up to polylog factors). 
\begin{itemize}
\item The case of binary CSPs ($r=2$) is fully classified. In particular, for every finite domain $D$ and $R \subseteq D^2$, we either have that $\SPR(R, n, \eps) = O(n/\eps^2)$ or $\SPR(R, n, \eps) = \Omega(n^2)$~\cite{butti2020Sparsificationa}. However, the sparsification routine is only efficient in the Boolean case~\cite{filtser2017Sparsification}. Of note, $\SPR(R, n, \eps) = \Omega(n^2)$ if and only if there exist $D_1, D_2 \subseteq D$ of size exactly $2$ such that $|R \cap (D_1 \times D_2)| = 1$ (informally $R$ has an ``induced copy'' of $\AND_2$).
\item For $r \ge 3$, much less is known. Kogan and Krauthgamer~\cite{kogan2015} contributed near-linear hypergraph cut sparsifiers (i.e., the predicate is $\NAE_r := \{0,1\}^r \setminus \{0^r, 1^r\}$). Since then, there have been multiple improvements in efficiently constructing hypergraph sparsifiers/sketches (e.g., \cite{chen2020linear,KKTY2021,khanna2024optimal}).
\item The breakthroughs of Khanna, Putterman, and Sudan~\cite{khanna2024Code,khanna2024Characterizations} construct near-linear sparsifiers for any predicate which can defined by a system of linear (in)equations (possibly over a higher domain). For example $\NAE_r = \{x \in \{0,1\}^r : x_1 + \cdots + x_r \not\equiv 0 \mod r\}$. Of note, their first paper~\cite{khanna2024Code} only proved the result over finite fields (and was nonalgorithmic), whereas their second paper~\cite{khanna2024Characterizations} extended the result to all Abelian groups and was computationally efficient.
\item The framework of Khanna, Putterman, and Sudan~\cite{khanna2024Characterizations} produced numerous corollaries. In particular, if a predicate can be expressed as the nonzero set of a degree $k$ polynomial, then it has a sparsifier of size $\widetilde{O}_{\eps}(n^k)$. Furthermore, they show if a predicate $R$ can express\footnote{More specifically, we say that $R \subseteq \{0,1\}^r$ can express
$\AND_k$ if there exits a map $z : [r] \to \{0,1,x_1, \hdots, x_k, \overline{x_1}, \hdots, \overline{x_k}\}$ such that $R(z(1), \hdots, z(r)) = \AND_k(x_1, \hdots, x_k)$. We discuss a more general framework of gadget reductions in Section~\ref{subsec:gadget}.} 
$\AND_k := \{1^k\}$, then $\SPR(R, n, \eps) = \Omega(n^k)$. As a consequence, they also classify all ternary Boolean predicates ($r=3$) as well as which Boolean predicates of arity $r$ cannot be sparsified below $\Omega(n^r)$ (just $\AND_r$ and its bit flips), while also constructing a sparsifier of size $\widetilde{O}_{\eps}(n^{r-1})$ in the other cases.
\item It appears that lower bounds with a nontrivial dependence on $\eps$ are only known for cut sparsifiers (and thus hypergraph cut sparsifiers via a simple gadget reduction). See \cite{andoni2016Sketching,MR3909627} as well as Section~\ref{sec:conclusion} for further discussion. 
\end{itemize}

\paragraph{CSP Kernelization.} Another question similar in spirit to CSP sparsification is that of \emph{CSP kernelization.}\footnote{More commonly, CSP kernelization is referred to as CSP sparsification (e.g., \cite{dell2014Satisfiability,lagerkvist2020Sparsification}). However, we refer to this line of work by the former name to reduce ambiguity. This similarity in name has been noted before in the literature (e.g., \cite{butti2020Sparsificationa}), but we appear to be the first work to notice both variants of ``CSP sparsification'' can be analyzed with similar techniques.} The basic question is to, given an instance $\Psi$ of $\CSP(R)$, \emph{efficiently} find as small of an instance $\Psi'$ of $\CSP(R)$ as possible (not necessarily a subinstance) such that $\Psi$ and $\Psi'$ are either both satisfiable or both unsatisfiable. This particular question can be attributed to Dell and van Melkebeek~\cite{dell2014Satisfiability}, who were particularly inspired by Impagliazzo, Paturi, and Zane's sparsification lemma \cite{IPZ2001} and Harnik and Naor's compression framework~\cite{HarnikNaor2010}. See the literature review in \cite{dell2014Satisfiability} for further motivations.  

At first, the problem seems rather unrelated to CSP sparsification. For example, if $\CSP(R)$ is polynomial-time tractable, then there trivially exists a kernel of size $O(1)$. When $\CSP(R)$ is NP-hard, however, the size of the smallest possible kernelization seems to much more closely track with the non-redundancy of $R$. In particular, Dell and van Melkebeek~\cite{dell2014Satisfiability}, proved that assuming $\mathsf{coNP} \nsubseteq \mathsf{NP/poly}$, the problem $k$-SAT cannot be kernelized below $\Omega(n^{k-\eps})$ for any constant $\eps > 0$, which is close to $k$-SAT's non-redundancy of $\Theta(n^k)$.

Furthermore, most upper bounds on the kernelization of NP-hard predicates follow from upper bounds on non-redundancy (see \cite{carbonnel2022Redundancy}). For example the works of Chen, Jansen, and Pieterse~\cite{chen2020BestCase} as well as Lagerkvist and Wahlstr{\"o}m~\cite{lagerkvist2017Kernelization,lagerkvist2020Sparsification} develop various kernelization methods that happen to just be ``efficient'' non-redundancy upper bounds. For example, these works show that if the predicate $R$ can be expressed as the zero set of a polynomial of degree $k$, then there exist a kernel of size $O(n^k)$. This kernel happens to preserve every solution to $R$, so it is also a non-redundancy upper bound. Using techniques like these, they are able to prove a number of results similar to the state-of-the-art in CSP sparsification, such as a complete classification of ternary Boolean predicates and a $O(n^{r-1})$ vs $\Omega(n^{r-\eps})$ Boolean dichotomy~\cite{chen2020BestCase}. See \cite{jansen2019Optimal,jansen2020Optimal,jansen2020Crossing,takhanov2023algebraic,beukersExtending} and citations therein for related work.

We seek to emphasize that any \emph{efficient} CSP sparsification algorithm for $\CSP(R)$ is by design a kernelization algorithm for $\CSP(\overline{R})$ (since all codewords with weight $0$ are preserved). As such, making Theorem~\ref{thm:main} efficient would require explicitly proving that every CSP can be kernelized to (approximately) its non-redundancy, which is a significant open question in the CSP kernelization community (see~\cite{carbonnel2022Redundancy}). See Section~\ref{subsec:open-questions} and Section~\ref{sec:conclusion} for further discussion.

\paragraph{The Union-closed Sets Conjecture.} A family $\cF$ of subsets of $[n]$ is \emph{union-closed} if $A, B \in \cF$ imply that $A \cup B \in \cF$. In 1979, Frankl~\cite{frankl1995extremal} conjectured that there always exists $i \in [n]$ which appears in at least half of the sets of $\cF$. For decades, progress on the conjecture was minimal, with the best general result being that some $i \in [n]$ appears in $\Omega(1/\log_2|\cF|)$ of the sets \cite{knill1994graph,wojcik1999Unionclosed,gilmer2022constant}. However, in 2022, Gilmer~\cite{gilmer2022constant} shocked the combinatorics community by using an entropy-based approach to prove that some $i \in [n]$ appears in $1/100$ of the sets. This immediately led to a large number of follow-up works refining Gilmer's entropy method~\cite{alweiss2022improved,chase2022approximate,pebody2022Extensiona,sawin2023improved,cambie2022better}. In particular, we can now replace `$1/100$' with `$0.382\hdots$', leaving Frankl's conjecture (technically) still open.

For our application to CSP sparsification, the entropy method used by Gilmer (and its subsequent refinements by many other reseachers) is the key idea needed to show that non-redundancy is essentially the optimal size for a CSP sparsifier. In particular, the improvement from $1/\log_2|\cF|$ to $\Omega(1)$ is precisely the same ``gain'' we utilize to go from a very simple $\widetilde{O}_{\eps}(\NRD(C)\cdot \log_2|C|)$ sparsifier (see Section~\ref{sec:simple}) to our $\widetilde{O}_{\eps}(\NRD(C))$ sparsifier. See the technical overview (Section~\ref{subsec:technical-overview}) for more details. To the best of our knowledge, our work is the first application of Gilmer's entropy method to sparsification.\footnote{Gilmer's breakthough is cited in the literature review of \cite{DBLP:conf/innovations/0001D0NS24}, but the property-testing question they study on union-closed families has no technical connection to Gilmer's entropy method. See also \cite{wakhare2024two} for applications of the entropy method to learning theory and statistical physics.}

\paragraph{Matching Vector Families and Locally Decodable Codes.} In coding theory, locally decodable codes (LDCs) are a class of codes which allow for jthe reliable recovery of any message symbol based on a small sample of codeword symbols, even in the presence of a constant fraction of errors. A particularly interesting familiy of constructions of LDCs has arisen out of a theory of \emph{matching vector codes}~\cite{yekhanin20083query} and follow-ups ~\cite{raghavendra2007note,gopalan2009note,efremenko20093query,dvir2011Matching}.
See \cite{dvir2011Matching} for a literature survey. Simply stated, a matching vector (MV) family over a (finite) ring $\cR$ is a pair of lists of vectors $u_1, \hdots, u_k, v_1, \hdots, v_k \in \cR^d$ such that the inner products $\langle u_i, v_j\rangle$ are nonzero\footnote{Or, more generally the inner products lie in some restricted subset of $\cR$.} if and only $i \neq j$. Informally, the $u_i$'s play a role in the encoding of the $i$'th message symbol, with the matching vector $v_i$ helping with its local decoding. Given a choice of $\cR$ and $d$, the primary question of interest is to find the maximal possible value of $k$. This ``spin off'' question about LDCs has become a topic of interest in its own right~\cite{dvir2011Matching,yekhanin2012Locally,gopalan2012Locality,bhowmick2013New}.

In this work, we demonstrate a novel application of matching vector families to the study of non-redundancy and thus (by Theorem~\ref{thm:main}) sparsification. In particular, we construct an explicit family of predicates such that their non-redundant instances can be viewed as a generalized MV family. We then use techniques developed for MV families to given nontrivial bounds on the non-redundancy of the predicates. See Section~\ref{sec:nrd-mv} and the technical overview (Section~\ref{subsec:technical-overview}) for more details.

\paragraph{Extremal Combinatorics.} Computing the non-redundancy of a predicate can be viewed as a problem in extremal combinatorics known as a \emph{hypergraph Tur{\'a}n problem}. In particular, for an instance of a CSP to be non-redundant, every instance induced by a subset of the variables must also be non-redundant. In particular, if $\cF$ is a family of hypergraphs which can never appear in non-redundant instances of $\CSP(R)$, then $\NRD(R, n) \le \ex_r(n, \cF)$, where the hypergraph Tur{\'a}n number $\ex_r(n, \cF)$ is the size of the largest $r$-uniform hypergraph on $n$ vertices without any $F \in \cF$ as a subgraph. This observation was first made explicit by Carbonnel~\cite{carbonnel2022Redundancy} although the technique was also used in earlier work \cite{bessiere2020Chain}. As far as we are aware, ours is the first work to observe that these insights can also benefit the study of CSP sparsification. 

The literature on hypergraph Tur{\'a}n numbers is quite rich. For instance, Keevash~\cite{keevash2011Hypergraph} surveys the vast body of work on the ``non-degenerate'' case in which $\ex_r(n, \cF) = \Omega_r(n^r)$. However, for our applications, we are mostly interested in the ``denegerate'' case in which $\ex_r(n, \cF) = O(n^{c})$ for some $c \in [1, r)$. The works \cite{bessiere2020Chain,carbonnel2022Redundancy} apply some of the most well-known works in this setting~\cite{erdos1964extremal,sos1973existence,MR519318} to get some nontrivial results such as classifying precisely which predicates $R$ have $\NRD(R, n) = \Theta(n^r)$, extending Chen, Jansen, and Pieterse's result for the Boolean case~\cite{chen2020BestCase}. See Sections~\ref{subsec:size}, \ref{subsec:simplest}, and \ref{subsec:msd-cl} for more details on specific applications.

\paragraph{Query Complexity and Learning Theory.} Rather surprisingly, the definition of non-redundancy appears to have come out of the artificial intelligence community~\cite{bessiere2020Chain}. In particular, a rather broad and well-studied question (e.g., \cite{freuder2002Suggestion,paulin2008Automatic,lallouet2010Learning,bessiere2012non,bessiere2013constraint,bessiere2020Chain}) is that of \emph{constraint acquisition}: how can an agent learn the constraints defining an instance of a constraint satisfaction problem?

A model specifically relevant to our work is the \emph{partial membership queries} model studied by Bessiere,  Carbonnel, and Katsirelos~\cite{bessiere2020Chain}. In this model, the domain $D$, the constraint type $R$ (or types), and the set of variables $X$ are known but the constraints are hidden. For each query, the agent picks some subset of variables $X' \subseteq X$ as well as a partial assignment $\sigma : X' \to D$. The response to the query is `YES' if $\sigma$ satisfies every constraint induced by $X'$, and `NO' otherwise. The goal is to construct an instance of $\CSP(R)$ with the same solution set as the hidden CSP. For every CSP predicate $R$, they prove that the query complexity of an instance of $\CSP(R)$ on $n$ variables is bounded between $\Omega(\NRD(R,n))$ and $O(\CL(R,n)\cdot \log n)$. Notably, the lower bound is proved by showing that the VC dimension of the query complexity problem equals $\NRD(R,n)$.\footnote{This observation is directly used in proving our main result, see Section~\ref{subsec:vc-prelim}.}

\subsection{Subsequent Work}

Since the initial version of our paper was posted, numerous follow-up works have emerged which expand on the many directions covered in this paper.

\subsubsection{Sparsifier Improvements}

\paragraph{Improvements to the Unweighted Sparsifier.} Very recently, \cite{SRY26} gave an improved bound for Theorem~\ref{thm:main-code}. They coin the term ``moonflowers'' to refer to non-redundant subsets of a code $C \subseteq \{0,1\}^m$ and adapt ideas from the study of sunflowers (e.g., \cite{ALWZ21}) to study the maximum density of a code which avoids particular moonflowers. In particular, their main sparsifier has a bound of the form.
\[
    \SPR(C, \eps) \le O\left(\frac{\NRD(C) \cdot \log m}{\eps^2}\right) \cdot (\log (\NRD(C) / \eps))^{O(1)}(\log \log m)^{O(1)}
\]
The overall method of constructing the sparsifier is rather similar to our work in the sense that it crucially uses the entropy method of Gilmer as a key step in the sparsification process, although also have a number of technical refinements of our techniques.  By suitably adapting our Example~\ref{example:chain} they also show a lower bound of $\Omega(\NRD(C) \cdot \log m / \eps)$ for an adversarial choice of $C$. Despite the improvements, the impact on Theorem~\ref{thm:main} is negligible as the parameters $\log m$ and $\log \NRD(C)$ are related by a factor of $r$ (the arity of the underlying CSP) which we consider to be a constant. Furthermore, the lower bound example does not arise from a CSP, so no limitation beyond $\NRD(C)$ itself is known for Theorem~\ref{thm:main}.

\paragraph{Improvements to the Weighted Sparsifier.} Another recent work \cite{BGP26b} improved Theorem~\ref{thm:wspr-code-intro} on optimal weighted code sparsifiers. By adapting the ``contraction''-based techniques for linear code sparsifiers~\cite{khanna2024Code,khanna2024Characterizations} in combination with some of the techniques developed in Section~\ref{sec:spr-weighted}, the authors of \cite{BGP26b} showed that Theorem~\ref{thm:wspr-code-intro} can be improved to
\[
    \wSPR(C, \eps) \le O\left(\frac{\CL(C) \cdot \log^2(\CL(S)/\eps) \cdot (\log \log (\CL(S)/\eps))^2}{\eps^2}\right).
\]
Note that the dependence on $m$ is removed entirely. That said, since we do not understand the relationship between the non-redundancy and chain length of CSPs, these techniques are unable to give any improvements to Theorem~\ref{thm:main}. Nor can they recover the interpolation between Theorem~\ref{thm:main-code} and Theorem~\ref{thm:wspr-code-intro} in Corollary~\ref{cor:wspr-improved}.

\subsubsection{Further Study of Non-redundancy}

Building off of the contributions this paper makes to the study of CSP non-redundancy, multiple follow-up works \cite{brakensiek2025richness,BGP26a,SV26b} have further expanded our knowledge of non-redundancy.

\paragraph{New Techniques.} In particular, the lower-bound technique of Theorem~\ref{thm:3linG-intro} was abstracted in \cite{brakensiek2025richness} in a logical framework known as generalized ``functionally guarded primitive positive'' (fgpp) definitions. They also introduce using the Kruskal-Katona and Shearer's inequalities to give novel upper bounds on CSP non-redundancy. As a result, we now know for any rational $\beta \ge 1$, there exists an explicit CSP predicate $R$ such that $\NRD(R, n) = \Theta_{R}(n^{\beta})$. Our notion of conditional non-redundancy (Definition~\ref{def:cond-nrd}) was critical in proving this fact. 

In addition, \cite{brakensiek2025richness} made a number of other contributions to the study of CSP non-redundancy. They show that many extremal hypergraph problems (such as understanding the maximum density of graphs with a given girth) can be encoded into suitable CSP non-redundancy problems. They also introduce a novel tool known as \emph{Catalan polymorphisms} to study the existence of Mal'tsev embeddings (see~\cite{lagerkvist2020Sparsification,bessiere2020Chain}).

\paragraph{Further Classifications.} We mentioned earlier that the non-redundancy of every Boolean predicate of arity at most $3$ is known~\cite{chen2020BestCase,khanna2024Characterizations}. The recent work \cite{BGP26a} extends this result to arity $4$ by classifying the non-redundancy of every arity-4 Boolean predicate except (up to isomorphism) one. Of interest, they identify a predicate $R \subseteq \{0,1\}^4$ for which
\[
    \NRD(R, n) \in \left[\frac{n^3}{2^{O(\sqrt{\log n})}}, \frac{n^3}{2^{\Omega(\log^*(n))}}\right],
\]
the first example of non-redundancy growth which provably does not follow a simple power law (e.g., $\Theta(n^{\beta})$).

The work of \cite{SV26b} investigates the non-redundancy of symmetric Boolean predicate of arity 4 and 5. Notably, they identify two interesting symmetric predicates of arity 5 for which state-of-the-art techniques leave a polynomial gap between the upper and lower bounds.

\subsubsection{Broader Connections}

\paragraph{Spectral Sparsifiers.} Although Theorem~\ref{thm:main} resolves the optimal size of CSP sparsifiers, there is still much ongoing work to understand \emph{spectral} variants of CSP sparsification, analogous to the theory of spectral graph and hypergraph cut sparsifiers~\cite{spielman2011Spectral,DBLP:journals/siamcomp/BatsonSS12,lee2018constructing,chen2020linear,KKTY2021,kenneth2023cut,khanna2024optimal,yoshida2026}.

In one such direction, \cite{KPS25spectral} define the notion of a spectral sparsifier for any CSP predicate. Their work is mostly focused on the study of spectral sparsifiers for linear equations (i.e., ``spectral code sparsifiers'') where they show such sparsifiers exist of size $\widetilde{O}_{\eps}(n^2)$. They leave reducing the sparsifier size to near-linear as an open question.

In another direction, \cite{BKLM26} motivated by the theory of Cayley graph sparsifiers (see also \cite{khanna2024Code,HLMPZ26}) developed a theory of sparsification of a family $\mathcal A$ of positive semi-definite (PSD) matrices. In particular, they introduce a quantity $N^*(\alpha, \mathcal A)$ which measures how tightly any matrix $A \in \mathcal A$ can be dominated by sums of other matrices (up to a factor of $\alpha$). Their main result is that up to a logarithmic factor in the dimension of the matrices, $N^*(\alpha, \mathcal A)$ controls the sparsifiability of the family.

In the setting of our code sparsification problem, their result can recover our warm-up bound Theorem~\ref{thm:nrd-spr-warmup} by representing a code $C \subseteq \{0,1\}^m$ as a family of $m$ diagonal matrices which are each $|C|$-dimensional. However, their union bound is too coarse to recover Theorem~\ref{thm:main}.

\paragraph{Streaming Algorithms.} Very recently, \cite{SV26a} showed that (up to logarithmic factors), non-redundancy of a CSP also governs the \emph{streaming} decidability of CSPs. That is, fix a CSP predicate $R$ and consider a streaming of clauses corresponding to an instance of $\CSP(R)$ on $n$ variables. The goal is to minimize the amount of space needed to determine if the stream corresponds to a satisfiable instance. The reason that non-redundancy is connected is that if the streaming algorithm just keeps track of a non-redundant set of clauses, any other (dropped) clause can be logically deduced from these clauses, so decidability is preserved. Using a communication complexity lower bound, they show this analysis is essentially tight. Such questions are also closely related to the study of approximating min-CSPs. As such, any future improvements to the study of streaming decidability will also impact the study of CSP non-redundancy and CSP sparsifiability.

\paragraph{Average-case Sparsification.} From Theorem~\ref{thm:main}, we know that every instance $\CSP(R)$ has a sparsifier of size approximately its own non-redundancy. A previous version of the paper asked whether the sparsifier size can be substantially improved in the \emph{average case} setting. That is, if the CSP predicate $R$ is fixed but the instance is sampled randomly, can the $\NRD(\overline{R}, n)$ barrier be broken? This was generally resolved in the affirmative by \cite{BGP25} where they compute optimal sparsifier sizes for random instances of all CSPs and \emph{valued} CSPs. See Section~\ref{sec:conclusion} for further discussion.

\subsection{Open Questions}\label{subsec:open-questions}

We conclude the introduction with a few directions of further study.  See Section~\ref{sec:open-nrd} and Section~\ref{sec:conclusion} for a more thorough discussion of directions for future exploration.

\begin{itemize}
\itemsep=0ex
\vspace{-1ex}
\item \textbf{Making Theorem~\ref{thm:main} efficient.} Note that the underlying construction for Theorem~\ref{thm:main-code}, if made algorithmic, runs in polynomial time with respect to the size of the code, yielding an $\exp(O(n))$-time algorithm\footnote{This is already nontrivial, as a naive guess-and-check algorithm would require $\exp(\widetilde{O}(\NRD(\overline{R},n)))$ time. } for Theorem~\ref{thm:main}. The primary barrier in constructing our sparsifier in $\poly(n)$ time is the fact that an efficient sparsifier is also a kernelization algorithm, but kernelizing every CSP instance to its non-redundancy is a significant open question in the kernelization community \cite{carbonnel2022Redundancy}.
\item \textbf{Computing $\NRD(R,n)$.} For a general predicate $R \subseteq D^r$, there is no simple (even conjectured) expression for $\NRD(R,n)$. In fact, even determining when $\NRD(R,n) = \Theta(n)$ is a major open question (e.g., \cite{bessiere2020Chain,carbonnel2022Redundancy}). In Section~\ref{sec:open-nrd}, we explore a number of predicates from the various parts of the literature whose status is unresolved, including a predicate we categorize as the ``simplest unresolved predicate.'' 
\item \textbf{Non-redundancy versus Chain Length.} Recall we show that unweighted sparsification is closely tied to non-redundancy while weighted sparsification is closely tied to chain length. For non-linear codes, $\NRD$ and $\CL$ can be very different (e.g., Example~\ref{example:chain}), but the relationship for CSPs is unknown~\cite{bessiere2020Chain,carbonnel2022Redundancy}. In particular, it seems quite possible that there exists a CSP predicate $R$ for which $\wSPR(R, n, \eps) / \SPR(R, n, \eps) = n^{\Omega(1)}.$ 
\item \textbf{Spectral CSP Sparsification.} As mentioned, \cite{KPS25spectral} recently defined a notion of spectral CSP sparsifiers. Proving an analogue of Theoerm~\ref{thm:main} in this setting would be a rather interesting result.
\end{itemize}

\subsection*{Organization}

In Section~\ref{sec:prelim}, we prove some basic facts about non-redundancy, sparsification and their relationship. In Section~\ref{sec:simple}, we give a straightforward proof that $\SPR(C, \eps) = \widetilde{O}_{\eps}(\NRD(C) \log_2|C|)$. In Section~\ref{sec:entropy}, we prove Theorem~\ref{thm:main} by connecting CSP sparsification to non-redundancy via Gilmer's entropy method.  In Section~\ref{sec:spr-weighted}, we extend Theorem~\ref{thm:main} to weighted instances. In Section~\ref{sec:spr-app}, we discuss the immediate applications of Theorem~\ref{thm:main} (and its weighted variant) based on what is known about non-redundancy and chain length in the literature. In Section~\ref{sec:nrd-mv}, we bound the non-redundancy of a family of predicates via methods related to matching vector families. In Section~\ref{sec:open-nrd}, we give examples of CSP predicates in the literature whose non-redundancy is unresolved.  In Section~\ref{sec:conclusion}, we wrap up with other directions of exploration. We emphasize that the material in Sections~\ref{sec:simple}, \ref{sec:entropy}, and \ref{sec:spr-weighted} are largely independent of the material in Sections~\ref{sec:spr-app}, \ref{sec:nrd-mv}, \ref{sec:open-nrd}.

\subsection*{Acknowledgments}
We thank Libor Barto, Dmitry Zhuk, Madhu Sudan, and Louie Putterman for valuable conversations. We thank Arpon Basu, Pravesh Kothari, Cosmas Kravaris, and Louie Putterman for helping to correct some errors in the original manuscript. We also thank anonymous STOC reviewers for suggestions which improved the quality of this manuscript. This research was supported in part by a Simons Investigator award, NSF grant CCF-2211972, and NSF grant DMS-2503280.

\section{Preliminaries} \label{sec:prelim}

In this section, we give precise definitions of the non-redundancy and sparsification of CSPs and non-linear codes and discuss how they related to each other. In particular, we show that the main sparsification result for codes (Theorem~\ref{thm:main-code}) implies the main sparsification result for CSPs (Theorem~\ref{thm:main}). We conclude with a few tail inequalities useful for sparsification.

\subsection{CSPs}

Given a finite set $D$ called the \emph{domain} and integer $r \in \N := \{1,2,\hdots\}$ called the \emph{arity}, we call any subset $R \subseteq D^r$ a \emph{predicate} (interchangeably called a \emph{relation}). We say that $R$ is \emph{non-trivial} if $R \neq \emptyset$ and $R \neq D^r$. We define an \emph{instance} $\Psi$ of $\CSP(R)$ to be a pair $(X, Y)$, where $X$ is a (finite) set of variables and $Y \subseteq X^r$ is the set of clauses. Typically, we let $n := |X|$ parameterize the number of variables and $m := |Y|$ parameterize the number of clauses. Note that $m \le n^r$, so $\log m \le r\log n$. We now clarify a couple points about our model of CSPs.
\begin{itemize}
\itemsep=0ex
\item In general, a CSP may consist of a family of multiple predicates. In the context of finding solutions to a CSPs, this detail may significantly change the complexity, but for sparsification and related concepts (like non-redundancy), each predicate can be sparsified separately (e.g.,~\cite{khanna2024Characterizations}) resulting in a multiplicative error of at most the number of predicates (e.g., \cite{bessiere2020Chain,carbonnel2022Redundancy}). Since we think of predicates as being of constant size in this paper, this has no effect on our results.
\item The choice $Y \subseteq X^r$ explicitly allows for repeated variables in clauses.\footnote{Various papers in the literature make different (and sometimes ambiguous) choices regarding allowing repeated variables. As we argue, as long as $r$ is a constant, the answer can vary by at most a constant factor. Thus, we apply results from the literature (e.g., in Section~\ref{sec:spr-app}) without concern to their specific convention.} All our results also apply to the restricted family of instances without repeated variables. To see why, let $X' := X \times [r]$ and map each $(y_1, \hdots, y_r) \in Y$ to $((y_1,1), \hdots, (y_r,r)) \in Y' \subseteq X'^r$. One can verify that any sparsifier of $(X', Y')$ is also a sparsifier of $(X, Y)$, and the number of variables only differs by a constant factor.
\end{itemize}

We say that an \emph{assignment} is a map $\sigma : X \to D$. A clause $y := (y_1, \hdots, y_r) \in Y$ is \emph{$R$-satisfied} by $\sigma$ if and only if $\sigma(y) := (\sigma(y_1), \hdots, \sigma(x_r)) \in R$. We say that $\sigma$ is an $R$-satisfying assignment to $\Psi$ (or just $Y$) if every clause is satisfied. We let $\sat(R,\Psi)$ (or $\sat(R,Y)$ if $X$ is fixed) denote the set of satisfying assignments to $\Psi$. We now formally define \emph{non-redundancy}.

\begin{definition}[Non-redundancy of an instance, adapted from~\cite{bessiere2020Chain}]
We say that an instance $\Psi := (X,Y)$ of $\CSP(R)$ is \emph{non-redundant} if for all $Y' \subseteq Y$ of size $|Y|-1$, we have that $\sat(R,Y') \neq \sat(R,Y)$. In other words, for all $y \in Y$, there exists an assignment $\sigma_y$ which $R$-satisfies $Y \setminus \{y\}$ but not $y$. We define\footnote{We purposely deviate from the definition of the non-redundancy of an instance in  \cite{bessiere2020Chain}. In their work, they define the non-redundancy of an instance to be the number of $Y' \subseteq Y$ of size $|Y|-1$ for which $\sat(R,Y') \neq \sat(R,Y)$. For non-redundant instances our definitions coincide (so, Definition~\ref{def:NRD-CSP} always gives the same value), but for ``highly redundant'' instances our non-redundancy is much greater. We do this to ensure that our definition of non-redundancy is monotone with respect to adding clauses.} the non-redundancy of an arbitrary instance $(X,Y)$ to be the size of the largest $Y' \subseteq Y$ such that $(X,Y')$ is non-redundant.
\end{definition}

\begin{definition}[Non-redundancy of a predicate~\cite{bessiere2020Chain}]\label{def:NRD-CSP}
Given a relation $R \subseteq D^r$ and $n \in \N$, we define $\NRD(R, n)$ to be the maximum number of clauses of any non-redundant instance of $\CSP(R)$ on $n$ variables. 
\end{definition}

We further define the notion of a CSP sparsifier. Given an instance $\Psi := (X,Y)$ of $\CSP(R)$, we define the \emph{$R$-weight} of an assignment $\sigma : X \to D$, denoted by $\wt(R,\Psi,\sigma)$, to be the number of clauses of $Y$ $R$-satisfied by $\sigma$. Given a weight function $w : Y \to \R_{\ge 0}$, we define the \emph{$(w,R)$-weight} of $\sigma$ to be 
\[
    \wt(R, \Psi, w, \sigma) := \sum_{y \in Y} w(y) \cdot \one[\sigma(y) \in R].
\]
We can now define a CSP sparsifier.

\begin{definition}[CSP sparsifiers~(e.g., Definition~3.11~\cite{khanna2024Code})]\label{def:SPR-CSP}
Given $R \subseteq D^r$, we say that $w : Y \to \R_{\ge 0}$ is an \emph{$\eps$-sparsifier} of an instance $\Psi$ of $\CSP(R)$ if for all assignments $\sigma : X \to D$, we have that
\[
    (1-\eps)\wt(R, \Psi,\sigma) \le \wt(R, \Psi, w, \sigma) \le (1+\eps)\wt(R,\Psi,\sigma).
\]
We define the \emph{$\eps$-sparsity} of $\Psi$ to be the minimum support (i.e., $|w^{-1}(\R_{>0})|$) of any $\eps$-sparsifier. We define $\SPR(R, n, \eps)$ to be the maximum $\eps$-sparsity of any instance of $\CSP(R)$ on $n$ elements.
\end{definition}
\begin{remark}
One may also define sparsifiers for \emph{weighted instances} in essentially the same manner. See Section~\ref{sec:spr-weighted} for a discussion of this variant. The main takeaway is that we also prove a tight characterization in terms of another combinatorial quantity called \emph{chain length}.
\end{remark}

We now prove a lower bound for sparsity in terms of non-redundancy. Observe that in an instance $(X,Y)$ of $\CSP(R)$ if some $y \in Y$ has a corresponding assignment $\sigma_y$ which only satisfies $y$, then $y$ must have nonzero weight in any sparsifier (or else the weight of $\sigma_y$ will fall from $1$ to $0$). In particular, this implies that non-redundant instances of the \emph{complement} of $R$ cannot be sparsified at all. 

\begin{proposition}\label{prop:spr-lb}
For any $R \subseteq D^r$, let $\overline{R} := D^r \setminus R$. For all $n \in \N$ and $\eps \in (0,1)$, we have that
\begin{align}
    \SPR(R, n, \eps) \ge \NRD(\overline{R}, n).\label{eq:SPR-ge-NRD}
\end{align}
\end{proposition}
\begin{proof}
Consider any non-redundant instance $\Psi := (X,Y)$ of $\CSP(\overline{R})$ with $n = |X|$. For each $y \in Y$, let $\sigma_y$ be a solution in $\sat(\overline{R}, Y \setminus \{y\}) \setminus \sat(\overline{R}, Y)$.

View $\Psi$ also as an instance of $\CSP(R)$ and observe that for all $y \in Y$, $\sigma_y$ has a $R$-weight of $1$. Further, for any $w : Y \to \R_{\ge 0}$, $\sigma_y$ has a $(R,w)$-weight of $w(y)$. Thus, any $\eps$-sparsifier $w : Y \to \R_{\ge 0}$ of $\Psi$ must have that $w(y) \in (1-\eps,1+\eps)$ for all $y \in Y$. Therefore, $\supp(w) = Y$.

This proves that $\SPR(R, n, \eps) \ge |Y|$ for every non-redundant instance $(X,Y)$ of $\CSP(\overline{R})$ on $n$ varaibles. Taking the maximum over all choices of $Y$ proves (\ref{eq:SPR-ge-NRD}).
\end{proof}

Theorem~\ref{thm:main} establishes that this simple lower bound is essentially optimal.

\begin{remark}
\label{rem:backwards}
Based on the complement in (\ref{eq:SPR-ge-NRD}), one might think that $\NRD$ is defined ``backwards.'' However, as we shall see in Section~\ref{subsec:gadget}, the definition of $\NRD$ allows for a rich (universal) algebraic framework, which should greatly aid with the classification of $\NRD$ for all predicates.
\end{remark}

\subsection{Non-linear Codes}

An important abstraction for studying CSP sparsification is what we call \emph{non-linear code sparsification}. We define a (non-linear) Boolean \emph{code} to be an arbitrary $C \subseteq \{0,1\}^m$. We say that $C$ is \emph{non-trivial} if $C \neq \emptyset, \{0^m\}$. For any $c \in C$, we define its \emph{Hamming weight}, to be $\Ham(c) := c_1 + \cdots + c_m$. We define the support $c \in \{0,1\}^m$, denoted by $\supp(c) \subseteq [m]$ to be the set of nonzero coordinates. We further define $\supp(C) = \bigcup_{c\in C} \supp(c)$.

Given $S \subseteq [m]$ and $c \in \{0,1\}^m$, we define $c|_{S} \in \{0,1\}^S$ to be the \emph{list} $(c_i : i \in S)$. Likewise, we define \emph{punctured} code $C|_S := \{c|_S : c \in C\} \subseteq \{0,1\}^S$. Next, we formally define the non-redundancy of a code.

\begin{definition}\label{def:NRD-code}
A subset $I \subseteq [m]$ is \emph{non-redundant} for a code $C \subseteq \{0,1\}^m$ if for each $i \in I$, there exists $c \in C$ such that for all $i' \in I$, $c_{i'} = 1$ if and only if $i = i'$. We define the \emph{non-redundancy} of $C$, denoted by $\NRD(C)$, to be the size of the largest non-redundant set that is non-redundant for $C$.
\end{definition}

In other words, if we line up the codewords of $C$ as rows of an $|C| \times m$ matrix, $\NRD(C)$ is the dimension of the largest identity submatrix within that matrix. When $C$ is a linear code, $\NRD(C)$ equals the dimension of $C$. For non-trivial $C$, we have that $1 \le \NRD(C) \le m$.

We now formally define a Hamming-weight sparsifier of a code. Given a \emph{weight function}, $w : [m] \to \R_{\ge 0}$, we define the $w$-weight of a codeword $c \in C \subseteq \{0,1\}^m$ to be
\[
    \langle w, c\rangle := \sum_{i=1}^m w(i) c_i.
\]

\begin{definition}\label{def:SPR-code}
For $\eps \in (0,1)$, we say that $w : [m] \to \R_{\ge 0}$ is a $\eps$-sparsifier of $C \subseteq \{0,1\}^m$ if for all $c \in C$, we have that
\[
    (1-\eps)\Ham(c) \le \langle w,c\rangle \le (1+\eps)\Ham(c).
\]
We define the $\eps$-sparsity of $C$, denoted by $\SPR(C,\eps)$, to be the minimum support size (i.e., number of nonzero coordinates) of any $\eps$-sparsifier of $C$.
\end{definition}

\subsection{Theorem~\ref{thm:main-code} Implies Theorem~\ref{thm:main}}\label{subsec:implies}

We now connect the non-redundancy notions of CSPs and codes together. Given a relation $R \subseteq D^r$ and any instance $\Psi := (X,Y)$ of $\CSP(R)$, we can define a \emph{satisfiability code}, denoted by $C_{R,\Psi} \subseteq \{0,1\}^Y$ as follows. For each assignment $\sigma : X \to D$, we define a codeword $c_{\sigma}$ such that for each $y \in Y$, we have
\begin{align}
    c_{\sigma,y} := {\bf 1}[\sigma(y) \in R].\label{eq:build-code}
\end{align}
We then define $C_{R,\Psi} := \{c_{\sigma} \mid \sigma : X \to D\}$. We now show that our definitions of sparsity and non-redundancy of codes corresponds to that of CSPs.

\begin{proposition}\label{prop:csp-to-code}
Given $R \subseteq D^r$ and $n \in \N$, we have that
\begin{align}
\SPR(R, n, \eps) &= \max_{\substack{\Psi\emph{ instance of }\CSP(R)\\\emph{on $n$ variables}}}\SPR(C_{R,\Psi}, \eps)\label{eq:equiv-SPR}\\
\NRD(\overline{R}, n) &= \max_{\substack{\Psi\emph{ instance of }\CSP(R)\\\emph{on $n$ variables}}}\NRD(C_{R,\Psi})\label{eq:equiv-NRD}
\end{align}
\end{proposition}
\begin{proof}
We first prove (\ref{eq:equiv-SPR}). Observe that for any instance $\Psi := (X,Y)$ of $\CSP(R)$ and any weight function $w : Y \to \R_{\ge 0}$ we have that the $(R,w)$-weight of $\sigma$ is equal to the $w$-weight of $c_{\sigma}$ due to (\ref{eq:build-code}). Thus, $\SPR(C_{R,\Psi}, \eps)$ is precisely the $\eps$-sparsity of $\Psi$. Therefore, $\SPR(R,n,\eps)$ is the maximum such $\eps$-sparsity, we have proved (\ref{eq:equiv-SPR}).

Finally, we prove (\ref{eq:equiv-NRD}). Fix an instance $\Psi := (X,Y)$ of $\CSP(R)$. Let $Z \subseteq Y$ be a maximum-sized non-redundant index set of $C_{R,\Psi}$. In particular, for each $z \in Z$, there is an assignment $\sigma_z : X \to D$ such that $\sigma_z$ $R$-satisfies $z$ but no other $z' \in Z \setminus \{z\}$. Thus, $\sigma_z$ $\overline{R}$-satisfies $Z \setminus \{z\}$ but not $z$. Thus, $(X,Z)$ is a non-redundant instance of $\CSP(R)$, so $\NRD(\overline{R}, n) \ge \NRD(C_{R,\Psi})$ for all $\Psi$.

Likewise, if we let $(X,Z)$ be a maximum-sized non-redundant instance of $\CSP(\overline{R})$ with witnessing assignments $\sigma_z$ for $z \in Z$, we can view $\Psi := (X,Z)$ as an instance of $\CSP(\overline{R})$. Thus, $\{c_{\sigma_z} : z \in Z\}$ witness that $Z$ is a non-redundant index of $C_{R,\Psi}$. Therefore, $\NRD(C_{R,\Psi}) = |Z|$. This proves (\ref{eq:equiv-NRD}).
\end{proof}

Thus, we can now prove that Theorem~\ref{thm:main-code} implies Theorem~\ref{thm:main}.
\begin{proposition}
Theorem~\ref{thm:main-code} implies Theorem~\ref{thm:main}.
\end{proposition}
\begin{proof}
The lower bound follows from Proposition~\ref{prop:spr-lb}. For the upper bound, consider any instance $\Psi$ of $\CSP(R)$ on $n$ variables and $m$ clauses. Since $m \le n^r$, by Theorem~\ref{thm:main-code}, we have that
\[
    \SPR(C_{R,\Psi}, \eps) \le O(\NRD(C_{R,\Psi}(r\log n)^6/\eps^2).
\]
Take the maximum of both sides over all instances $\Psi$ on $n$ variables. Then, Proposition~\ref{prop:csp-to-code} implies Theorem~\ref{thm:main}.
\end{proof}

\subsection{Chernoff Bounds}

We now state a few Chernoff-type bounds that will be useful in our sparsification arguments. We let $\exp$ and $\log$ denote the natural exponential and logarithm, respectively. For any base $b > 1$, we let $\exp_b(x) := b^x$ and $\log_b(x) := \log(x)/\log(b)$.

As our sparsifiers utilize i.i.d. subsampling, we use the following Chernoff bound to analyze the probability that Hamming weight is preserved.
\begin{theorem}[e.g., \cite{Motwani_Raghavan_1995,mitzenmacher2005Probability,fung2011general}]\label{thm:chernoff}
Let $X_1, \hdots, X_n$ be i.i.d. samples of the Bernoulli distribution with probability $p \in [0,1]$. Then, for all $\delta > 0$,
\begin{align}
  \Pr[X_1 + \cdots + X_n < (1-\delta) p n] &\le \exp(-\delta^2 n p / 2)\label{eq:chernoff-lower}\\
  \Pr[X_1 + \cdots + X_n > (1+\delta) p n] &\le \exp(-\delta^2 n p / (2+\delta))\label{eq:chernoff-upper}
\end{align}
\end{theorem}

The key applications of the Chernoff bound in our sparsifier constructions are as follows.

\begin{lemma}\label{lem:chernoff-small}
Let $S \subseteq [m]$ be a random subset of $[m]$ such that each element is included independently with probability at most $1/3$. Then,
\[
  \Pr[|S| > m/2] \le \exp(-m / 30).
\]
\end{lemma}
\begin{proof}
By a standard coupling argument, the probability that $|S| > m/2$ is maximized when each element is included with probability exactly $1/3$. We then apply (\ref{eq:chernoff-upper}) with $p=1/3$, $n = m$, and $\delta = 1/2$. In that case, $\delta^2 n p / (2+\delta) = m/30$.
\end{proof}

\begin{lemma}\label{lem:chernoff-ham}
Let $S \cup T$ be a partition of $[m]$. For $p \in (0,1]$, let $S_p$ be a random subset of $S$ where each element of $S$ is included independently with probability $p$. For any codeword $c \in \{0,1\}^m$ and any $\eps \in (0,1)$, we have that
\[
\Pr_{S_p}\left[\frac{1}{p}\Ham(c|_{S_p}) + \Ham(c|_T) \not\in [1-\eps, 1+\eps] \cdot \Ham(c)\right] < 2\exp(-\eps^2 \Ham(c) p/3).
\]
\end{lemma}

\begin{proof}
Let $w_S = \Ham(c|_{S})$ and $w_T = \Ham(c|_{T})$, so $\Ham(c) = w_S + w_T$. If $w_S = 0$, then $\Ham(c) = w_T$. Thus, the probability of failure is $0$. Otherwise, apply Theorem~\ref{thm:chernoff} with $n := w_S$ and $\delta := \eps\cdot \frac{w_S+w_T}{w_S}$ (note that $\delta$ may be larger than $1$) to have that
\begin{align*}
  \Pr_{S_p}[\Ham(c|_{S_p}) \not\in [1-\delta,1+\delta]\cdot p n] 
                                                &<2\exp\left(-\frac{\delta^2 w_S p}{2+\delta}\right)\\
                                                &=2\exp\left(-\eps^2 \left(\frac{w_S+w_T}{w_S}\right)^2 \frac{w_S}{2 + \eps\cdot \frac{w_S+w_T}{w_S}} p\right)\\
                                                &=2\exp\left(-\eps^2 \frac{(w_S + w_T)^2}{2w_S+ \eps(w_S+w_T)}p\right)\\
                                                &\le 2\exp\left(-\eps^2 \frac{(w_S + w_T)^2}{3(w_S+w_T)}p\right)\\
                                                &= 2\exp\left(-\eps^2 \Ham(c) p/3\right),
\end{align*}
where the second-to-last line uses the fact that $\eps \in (0,1)$. Finally, to complete the proof, note that
\begin{align*}
  \Ham(c|_{S_p}) \not\in [1-\delta,1+\delta] p n
&\iff \frac{1}{p} \Ham(c|_{S_p}) + \Ham(c|_T) \notin [1-\delta,1+\delta]\cdot w_S + w_T\\
&\iff\frac{1}{p} \Ham(c|_{S_p}) + \Ham(c|_T) \notin [1 - \eps, 1+\eps]\Ham(c)  \ . \qedhere
\end{align*}
\end{proof}

\section{Warmup: A Simple Sparsifier}\label{sec:simple}

As a warm-up for our main result, we prove that every code $C \subseteq \{0,1\}^m$ has an $\eps$-sparsifier of size within $\widetilde{O}_{\eps}(\log |C| \cdot \NRD(C)).$ 

\begin{theorem}\label{thm:nrd-spr-warmup}
For all nonempty $C \subseteq \{0,1\}^m$ and $\eps \in (0,1)$, we have that
\[
  \SPR(C, \eps) = O(\NRD(C) \log |C| (\log^3 m)/\eps^2).
\]
More precisely,
\begin{align}
\SPR(C, \eps) \le 36 \NRD(C) \log_2(4|C|) \log_2^3 (2m) / \eps^2 \ . \label{eq:spr-simple-precise}
\end{align}
\end{theorem}

\begin{remark}
In terms of CSP sparsification, we have that for any $R \subseteq D^r$ that
\[
    \SPR(R, n, \eps) = O(n \cdot \NRD(\overline{R}, n) (\log |D|)(r\log n)^3/\eps^2).
\]
since any instance $\Psi$ of $\CSP(R)$ on $n$ variables has at most $n^r$ clauses and at most $|D|^n$ assignments (so $|C_{\Psi,R}|\le |D|^n$). Ignoring logarithmic factors, we are within a factor of $n$ of Proposition~\ref{prop:spr-lb}, which is highly nontrivial given the fact that apriori we only know that $\SPR(R, n, \eps) \in [\Omega(n), n^r].$ However, Theorem~\ref{thm:nrd-spr-warmup} has no nontrivial implications in the $r=2$ case.
\end{remark}

For any $d \in [m]$, let $C_{\le d}$ be the set of codewords of $C$ with Hamming weight at most $d$. The key observation of this simple sparsifier is that $\NRD(C)$ carefully controls the support size of $C_{\le d}$.

\begin{lemma}\label{lem:NRD-supp-bound}
For all $C \subseteq \{0,1\}^m$ and all $d \in \{0,1,\hdots, m\}$, we have that
\[
  \lvert\supp(C_{\le d})\rvert \le d \cdot \NRD(C).
\]
\end{lemma}

\begin{remark}
For linear codes, this lemma is trivial: the non-redundancy of a linear code is precisely its dimension and each time we expand the support of a linear code with a new codeword, the rank increases by $1$. A variant of this observation appears in \cite{khanna2024Code} as part of their divide-and-conquer framework for linear code sparsification. 
\end{remark}

\begin{proof}
We prove this result by induction on $d$. The base case of $d = 0$ is trivial. Now assume $d \ge 1$. Pick $I \subseteq [m]$ minimal such that\footnote{This is known as a \emph{hitting set} (e.g., \cite{even2005Hitting}).} for every non-zero $c \in C$, $\card{\supp(c) \cap I} \ge 1$. We claim that $I$ is non-redundant. If not, there exists $i \in I$ such that for every $c \in C$ with $c_i = 1$, we have that $\card{\supp(c) \cap I} \ge 2$. In that case, $I \setminus \{i\}$ is a smaller set with the prescribed property, contradicting $I$'s minimality. Therefore, $I$ is non-redundant and
\[
    (C_{\le d})|_{\overline{I}} \subseteq (C|_{\overline I})_{\le {d-1}}.
\]
Therefore, by the induction hypothesis,
\[
\card{\supp(C|_{\overline I})_{\le {d-1}}} \le (d-1) \NRD(C|_{\overline I}) \le (d-1) \NRD(C).
\]
Thus, since $I$ is non-redundant,
\[
    \card{\supp(C_{\le d})} \le |I| + \card{\supp(C|_{\overline I})_{\le {d-1}}} \le \NRD(C) + (d-1)\NRD(C) = d\NRD(C). \qedhere
\]
\end{proof}

We now proceed to construct the sparsifier. We use a fairly standard divide-and-conquer technique (e.g., \cite{khanna2024Code,khanna2024Characterizations}).

\begin{proof}[Proof of Theorem~\ref{thm:nrd-spr-warmup}]
We seek to prove (\ref{eq:spr-simple-precise}) by strong induction on $m$. The base cases of $m \le 36$ is trivial. Define $\eps_0 := \eps / (2\log_2(2m))$ and
\[
  d_0 := 9 \log_2(4|C|) /\eps_0^2 = 36 \log_2(4|C|) \log_2^2(2m)/\eps^2.
\]
Let $T = \supp(C_{\le d_0})$ and $S = [m] \setminus T$. By Lemma~\ref{lem:NRD-supp-bound}, we have that $|T| \le d_0 \cdot \NRD(C)$. Let $p := 1/3$ and let $S_p$ be a random subset of $S$ where each element is included independently with probability $p$. We seek to show that with positive probability, we have that $|S_p| \le m/2$ and for all $c \in C$,
\begin{align}
  \frac{1}{p}\Ham(c|_{S_p}) + \Ham(c|_T) \in [1-\eps_0, 1+\eps_0] \cdot \Ham(c).\label{eq:spar-1}
\end{align}
First, by Lemma~\ref{lem:chernoff-small} we have that $\Pr[|S_p| > m/2] < \exp(-m/30) < 1/3$ since $m > 36$. To prove (\ref{eq:spar-1}), first observe that if $\Ham(c) \le d_0$, then $\supp(c) \subseteq T$, so (\ref{eq:spar-1}) trivially follows. Otherwise, by Lemma~\ref{lem:chernoff-ham}, the probability that (\ref{eq:spar-1}) fails is at most 
\[
2\exp\left(-\frac{\eps_0^2 \Ham(c)}{9}\right) < 2\exp\left(-\frac{d_0\eps_0^2}{9}\right) = 2\exp(-\log_2(4|C|)) \le \frac{1}{2|C|}.
\]
In particular, by the union bound the probability that (\ref{eq:spar-1}) holds for all $c \in C$ is at least $1/2$. Therefore, we have proved the existence of $S' \subseteq S$ of size at most $m/2$ such that for all $c \in C$, the map
\[
  \forall i \in [m], w(i) := \begin{cases}
    1 & i \in T\\
    3 & i \in S'\\
    0 & i \in S \setminus S',
\end{cases}
\]
is an $\eps_0$-sparsifier of $C$. We now apply the induction hypothesis to $C' := C|_{S'}$ to find an $\eps' := \eps-2\eps_0$-sparsifier $w' : S'\to \R_{\ge 0}$ of size at most $36\NRD(C') \log_2(4|C'|) \log_2^3 (2|\supp C'|) / \eps'^2.$
Now define the sparsifier $\widetilde{w} : [m] \to \R_{\ge 0}$ by
\[
  \forall i \in [m], \widetilde{w}(i) := \begin{cases}
    1 & i \in T\\
    3w'(i) & i \in S'\\
    0 & i \in S \setminus S'.
\end{cases}
\]
Observe that for any $c \in C$, we have that
\begin{align*}
  \langle \widetilde{w}, c\rangle &= 3\langle w', c|_{S'}\rangle + \Ham(c|_{T})\\
                                  &\in 3[1-\eps',1+\eps']\cdot \Ham(c|_{S'}) + \Ham(c|_{T})\\
                                  &\subseteq [1-\eps',1+\eps']\cdot \langle w, c\rangle\\
                                  &\subseteq [1-\eps',1+\eps']\cdot[1-\eps_0,1+\eps_0] \cdot \Ham(c)\\
                                  &\subseteq [1-\eps, 1+\eps] \cdot \Ham(c),
\end{align*}
where the second-to-last inclusion follows from $w$ being an $\eps_0$-sparsifier. Thus, $\widetilde{w}$ is an $\eps$-sparsifier. To finish, we bound the size of $\widetilde{w}$ as follows:
\begin{align*}
\lvert \supp(\widetilde{w})\rvert &= \lvert \supp(w')\rvert + |T|\\
&\le \frac{36\NRD(C') \log_2(4|C'|) \log_2^3 (2|\supp C'|)}{\eps'^2} + \frac{36 \NRD(C)\log_2(4|C|) \log_2^2(2m)}{\eps^2}\\
&\le 36 \NRD(C) \log_2(4|C|) \left[\frac{\log_2^3 m}{(\frac{\log_2 m}{\log_2 (2m)}\eps)^2} + \frac{\log_2^2(2m)}{\eps^2}\right]\\
&= \frac{36\NRD(C)\log_2(4|C|)\log_2^2 (2m)}{\eps^2}[\log_2 m + 1]\\
& =\frac{36\NRD(C)\log_2(4|C|)\log_2^3 (2m)}{\eps^2}  \ . \qedhere
\end{align*}
\end{proof}

\section{The Entropy Sparsifier}\label{sec:entropy}

In this section, we prove Theorem~\ref{thm:main-code}. As mentioned in the introduction, this proof crucially needs bounds on the entropy of probability distributions over our non-linear code. We begin by defining the necessary notation for discussing such concepts and then proceed to analyze Gilmer's entropy method (and its subsequent refinements) and how these can assist with sparsification.

\subsection{OR-closure, VC Dimension, and Entropy}\label{subsec:vc-prelim}
We being by defining a few basic properties of codes.

\subsubsection{OR-closure and VC Dimension}

Given two strings $a,b \in \{0,1\}^m$, we define their bitwise-OR $a \vee b$ to be the string $c \in \{0,1\}^m$ such that $c_i = a_i \vee b_i$ for all $i \in [m]$. Likewise, given a list $a_1, \hdots, a_k \in \{0,1\}^m$, we let $\bigvee_{i=1}^k a_i$ denote their collective bitwise-OR.

\begin{definition}
Given $C \subseteq \{0,1\}^m$ we define its $\OR$-closure (or $\OR$-span) to be
\[
  \spor(C) := \bigcup_{k=0}^{\infty} \left\{\bigvee_{i=1}^k c_i : c_1, \hdots, c_k \in C\right\},
\]
where the empty $\OR$ is defined to be $0^m$. If $C = \spor(C)$, we say that $C$ is \emph{OR-closed}.
\end{definition}

We now discuss the VC dimension of $C$.

\begin{definition}
For $C \subseteq \{0,1\}^m$, we define the VC dimension of $C$, denoted by $\VC(C)$, to be the size of largest $I \subseteq [m]$ such that $C|_{I} = \{0,1\}^I$. In particular, any code with at most one codeword has VC dimension $0$. 
\end{definition}

We now restate an observation of \cite{bessiere2020Chain} that non-redundancy of a code is precisely the VC dimension of its OR-closure.

\begin{proposition}[Obs.~4~\cite{bessiere2020Chain}, restated]\label{prop:nrd-vc}
For all $C \subseteq \{0,1\}^m$, $\NRD(C) = \VC(\spor(C))$.
\end{proposition}

\begin{proof}
Observe that for all $I \subseteq [m]$, $\spor(C|_{I}) = \{0,1\}^I$ if and only if for all $i \in I$, there exists $c_i \in C_{I}$ such that for all $i' \in I$, $c_{i,i'} = 1$ iff $i = i'$. In other words, $\spor(C|_{I}) = \{0,1\}^I$ if and only if $I$ is non-redundant in $C|_{I}$ (and thus $C$). Since $\spor(C|_{I}) = \spor(C)|_{I}$ for all $I \subseteq [m]$, we have $\VC(\spor(C))$ is precisely the non-redundancy of $C$.
\end{proof}

Our main use of VC dimension is to bound the size of the corresponding code via the Sauer-Shelah-Peres lemma.

\begin{lemma}[Sauer-Shelah-Peres~\cite{sauer1972density,shelah1972combinatorial}]\label{lem:sauer-shelah}
For all $C \subseteq \{0,1\}^m$,
\[
|C| \le \sum_{i=0}^{\VC(C)} \binom{m}{i} \le (m+1)^{\VC(C)}.
\]
\end{lemma}

As an immediate corollary of Proposition~\ref{prop:nrd-vc} and Lemma~\ref{lem:sauer-shelah}, we can control the size of $\spor(C)$ in terms of $\NRD(C)$ via the Sauer-Shelah-Perles lemma.

\begin{corollary}\label{cor:spor-bound}
For any $C \subseteq \{0,1\}^m$, 
\[
    \card{\spor(C)} \le (m+1)^{\NRD(C)}.
\]
\end{corollary}

\subsubsection{Entropy} We define the binary entropy function to be 
\[
h(x) := x \log_2(1/x) + (1-x) \log_2(1/(1-x))
\]
for $x \in (0,1)$ with $h(0) = h(1) = 0$ with maximum value at $h(1/2) = 1$. Given a probability distribution $\mathcal D$ over a finite set, we define the entropy of the distribution as
\[
H(\mathcal D) := \sum_{u \in \supp(\cD)} \Pr_{\cD}[u] \log_2\left(\frac{1}{\Pr_{\cD}[u]}\right).
\]
Since $x \log_2(1/x)$ is a concave function of $x$, we have that every distribution $\cD$ over a finite set $S$ has $H(\cD) \le \log_2|S|$, with equality when $\cD$ is the uniform distribution over $S$. As such, the following is implied by Corollary~\ref{cor:spor-bound}.

\begin{proposition}\label{prop:entropy-bound}
For any $C \subseteq \{0,1\}^m$ and any probability distribution $\cD$ over $\spor(C)$, we have that
\begin{align}
  H(\cD) \le \NRD(C) \cdot \log_2(m+1).\label{eq:entropy-nrd-bound}
\end{align}
\end{proposition}

Given a distribution $\cD$ over $\{0,1\}^m$, we define $\E[\cD] \in [0,1]^m$ to be the coordinate-wise expected value of the distribution. Given $I \subseteq [m]$, we let the punctured distribution $\cD|_{I}$ to be
\[
  \Pr_{\cD|_{I}}[c] := \sum_{\substack{c' \in \{0,1\}^m\\c'|_{I} = c}} \Pr_{\cD}[c'].
\]

We next discuss the significance of entropy to the structure of codes.

\subsection{Gilmer's Entropy Method}\label{subsec:gilmer}

In 2022, Gilmer~\cite{gilmer2022constant} made the following breakthough on the structure of union-closed sets, which we restate in terms of the OR-closed codes.

\begin{theorem}[\cite{gilmer2022constant}, restated]\label{thm:gilmer}
Let $C \subseteq \{0,1\}^m$ be OR-closed, let $\cD$ be the uniform distribution over $C$. Then, there exists $i \in [m]$ for which $\E[\cD]_i \ge 1/100$.
\end{theorem}

As mentioned in the introduction, numerous technical improvements have now risen the constant $1/100$ to more than $0.382$~\cite{alweiss2022improved,chase2022approximate,pebody2022Extensiona,sawin2023improved,cambie2022better}, with $1/2$ still conjectured as being the best possible improvement. Of note, the previous best version of Theorem~\ref{thm:gilmer} was a lower bound of $\Omega(1/\log_2 |C|)$~\cite{knill1994graph,wojcik1999Unionclosed,gilmer2022constant}. We now give some intuition as to why the $\Theta(\log_2|C|)$ ``gain'' in Gilmer's theorem is essentially the same $\Theta(\log_2|C|)$ we seek to shave in Theorem~\ref{thm:nrd-spr-warmup}. 

From the perspective of code sparsification, Gilmer's theorem appears quite useful, as adding the coordinate $i$ with $\E[\cD]_i \ge 1/100$ to our sparsifier allowed us to make nontrivial progress on sparsifying a constant fraction of $\spor(C)$. If we make the (bold) assumption that this constant-fraction of codewords need not be revisited, we can recursively apply Gilmer's theorem to the remaining $99/100$ fraction of $\spor(C)$. Then, since $\log_2 \card{\spor(C)} \le \NRD(C) \cdot \log_2(m+1)$, we will have ``sparsified'' all of $\spor(C)$ (and thus $C$) using only $\widetilde{O}(\NRD(C))$ coordinates.

Although this informal sketch has serious technical issues (e.g., why can we ``forget'' a codeword after saving a single coordinate?), it turns out that we can make a rigorous argument affirming this intuition. For simplicity, assume all codewords in $C \subseteq \{0,1\}^m$ have the same Hamming weight $d$. What we eventually prove (see Theorem~\ref{thm:NRD-decomp}) is that for any $\lambda \ge 1$ there exist $\approx \lambda \NRD(C)$ coordinates $I \subseteq [m]$ such that the punctured code $C|_{\overline{I}}$ has very few codewords (i.e., at most $\approx \exp(d/\lambda)$). In otherwords, the coordinates of $I$ capture the most important distinguishing features among codewords of $C$. We can then subsample the remaining coordinates of $C$ to make progress toward the sparsifier, like in one step of the recursive argument proving Theorem~\ref{thm:nrd-spr-warmup}.

We note that the linear code sparsifiers of Khanna, Putterman, and Sudan~\cite{khanna2024Code,khanna2024Characterizations} also prove a similar bound (e.g., see Theorem 2.2 in \cite{khanna2024Code}). However, the proof methods are very different. In \cite{khanna2024Code}, they recursively remove ``dense'' subcodes of their linear code and then use a ``Bencz{\'u}r-Karger-style'' contraction algorithm to prove the result code is sparse. In our case, we instead simultaneously understand the dense and sparse structures of our code $C$ using Gilmer's entropy method. In particular, by applying the minimax theorem in Proposition~\ref{prop:cover-sparse}, we show that our code either has a probability distribution over it in which each coordinate is unlikely to equal $1$ (``$\theta$-sparse'') or there is probability distribution over the coordinates has nontrivial overlap with \emph{every} codeword (``$\theta$-cover''). If there is a``$\theta$-cover,'' we prove we can add a coordinate to our sparsifier which makes substantial progress toward properly sparsifying every codeword.  Conversely, if there is a $\theta$-sparse distribution, we use Gilmer's entropy method to prove that the ``important part'' of the distribution is quite sparse, and thus can be removed. We apply Proposition~\ref{prop:cover-sparse} many times to repeatedly refine our understanding of sparse and dense structures in $C$. When the process terminates, we obtain the decomposition theorem (Theorem~\ref{thm:NRD-decomp}), from which the remainder of the proof follows using standard techniques.

\paragraph{The Entropy Bound.} We now return to discussing concrete technical details. The key entropy bound we use is a refined version of Gilmer's method due to Sawin~\cite{sawin2023improved}. Like for Theorem~\ref{thm:gilmer}, we restate the bound in terms of non-linear codes.

\begin{theorem}[Theorem 2,~\cite{sawin2023improved}, adapted]\label{thm:sawin}
Let $\mathcal D$ be a distribution on $\{0,1\}^m$ such that $\max_{i \in [m]}\E[\cD]_i \le p \le \frac{3-\sqrt{5}}{2} \approx 0.3819$. Let $A,B \sim \mathcal D$ be i.i.d.\ samples from $\cD$. Then,
\begin{align}
  H(A \vee B)\ge H(A) \cdot \frac{h(1-(1-p)^2)}{h(p)}.\label{eq:sawin}
\end{align}
\end{theorem}
Note that (\ref{eq:sawin}) is tight when $\cD$ is the product of $m$ independent Bernoulli distributions with probability $p$. The key idea used by Gilmer and others to prove results like (\ref{eq:sawin}) is to inductively show for $i \in [n]$ that
\begin{align}
H(A_i \vee B_i | A_{[i-1]}, B_{[i-1]}) \ge H(A_i | A_{[i-1]})\cdot \frac{h(1-(1-p)^2)}{h(p)}.\label{eq:process}
\end{align}
The crucial property is that the conditional random variable $\E[A_i|A_{[i-1]}]$ is independently and identically distributed to $\E[B_i|B_{[i-1]}]$. Thus, (\ref{eq:process}) can be reduced to proving a simple (but ingenious) inequality about two i.i.d.\ $[0,1]$-valued random variables with mean at most $p$. See \cite{sawin2023improved} and related works  (such as the exposition by Boppana~\cite{boppana2023Useful}) for more details.

 Although the application of (\ref{eq:sawin}) to the union-closed sets conjecture is mostly interested in the regime in which $p$ is a constant, we shall focus on using this bound in the regime in which $p$ tends to $0$ as $p \to \infty$. We recursively apply Theorem~\ref{thm:sawin} to amplify the growth in entropy.

\begin{corollary}\label{cor:entropy-power}
Let $\mathcal D$ be a distribution over $\{0,1\}^m$ with $\max_{i \in [m]}\E[\cD]_i \le p$. Let $N$ be a power of two such that $1 - (1-p)^{N/2} \le \frac{3-\sqrt{5}}{2}$. Let $A_1, \hdots, A_N$ be i.i.d.\ samples from $\cD$, then
\[
H\left(\bigvee_{i=1}^N A_i\right) \ge H(\cD) \cdot \frac{h(1 - (1-p)^N)}{h(p)}.
\] 
\end{corollary}
\begin{remark}
We only prove our result when $N$ power of two as more precision would not significantly improve our main result. See Conjecture 1 and Lemma 26 in~\cite{wakhare2024two} for a comparable bound for all $N$.
\end{remark}

\begin{remark}
Note that if $p = \Theta(1/N)$, then $h(1-(1-p)^N) / h(p) \approx \frac{N}{\log N}$ (see Lemma~\ref{lem:or-closed-sparse-bound}), so we get a nearly factor of $N$ boost in the entropy. As we shall soon see, this ``boost'' is analogous to the exponential savings in the Bencz{\' u}r-Karger cut bound.
\end{remark}

\begin{proof}
We prove this by induction on $\log_2 N$. If $N=2$, we apply Theorem~\ref{thm:sawin} directly. Otherwise, we consider the random variable $\cD_{N/2}$ corresponding to the distribution $\bigvee_{i=1}^{N/2}A_i$, and note that by independence, for all $i \in [m]$,
\[
\Pr_{c \sim \cD_{N/2}}[c_i=1] = 1 - \Pr_{c \sim \cD}[c_i=0]^{N/2} \le 1 - (1-p)^{N/2} \le \frac{3-\sqrt{5}}{2}.
\]
Therefore,
\begin{align*}
  H\left(\bigvee_{i=1}^N A_i\right)  &\ge H(\cD_{N/2}) \cdot \frac{h(1 - (1-p)^N)}{h(1-(1-p)^{N/2})}\\
&\ge H(\cD) \cdot \frac{h(1 - (1-p)^{N/2}))}{h(p)} \cdot \frac{h(1 - (1-p)^N)}{h(1-(1-p)^{N/2})},\\
&= H(\cD) \cdot \frac{h(1 - (1-p)^N)}{h(p)},
\end{align*}
where the second line invokes the induction hypothesis because $1-(1-p)^{N/4} \le 1-(1-p)^{N/2} \le \frac{3-\sqrt{5}}{2}$.
\end{proof}

\subsection{Improved Entropy Bound for $\theta$-sparse Distributions}\label{subsec:theta-sparse}

We say that a distribution $\cD$ over $\{0,1\}^m$ is $\theta$-sparse if for all $\max_{i \in [m]}\E[\cD]_i \le 1/\theta$. By using Gilmer's entropy method, we can show that the entropy of any $\theta$-sparse distributions is approximately $\frac{\log \theta}{\theta}$ times smaller than the RHS of (\ref{eq:entropy-nrd-bound}). In some sense, this bound is similar to \cite{khanna2024Code}'s code-counting lemma for sparse linear codes, although their bound is proved using a Bencz{\'u}r-Karger-style contraction algorithm whereas ours is proved using Gilmer's entropy method.

We begin by proving the following bound for any distribution over any OR-closed code.

\begin{lemma}\label{lem:or-closed-sparse-bound}
Let $C \subseteq \{0,1\}^m$ be an OR-closed code and let $\theta \ge 1$. Let $\cD$ be a $\theta$-sparse distribution over $C$. Then,
\begin{align}
H(\cD) \le \frac{3\log_2(3\theta)}{\theta}\cdot \log_2 |C|.\label{eq:closed-bound}
\end{align}
\end{lemma}

\begin{proof}
If $\theta < 4$, we directly see that
\[
H(\cD) \le \log_2 |C| \le \frac{3\log_2(3\theta)}{\theta}\cdot \log_2|C|.
\]
Otherwise, let $p := 1/\theta$ and $N := 2^{\lfloor \log_2 \theta\rfloor-1} \in (\theta/4, \theta/2]$.  Since $\theta \ge 4$, we have that $N \ge 2$. By Bernoulli's inequality,
\begin{align*}
  1 - (1-p)^{N/2} &\le pN/2 \le 1/4 < \frac{3-\sqrt{5}}{2}, \text{ and}\\
  1 - (1-p)^N &\le pN \le 1/2.
\end{align*}
Furthermore, since $1-p \le e^{-p}$ and $pN > 1/4$, we have that
\[
    1 - (1-p)^N \ge 1-e^{-pN} > 1 - e^{-1/4} > 0.22.
\]
Thus, since the binary entropy function in increasing the range $[0,1/2]$, 
\begin{align*}
h(1-(1-p)^N) > h(0.22) > 1/3.
\end{align*}

Therefore, by Corollary~\ref{cor:entropy-power} and the fact that $C$ is OR-closed, we have that if $A_1, \hdots, A_N$ are i.i.d.\ samples of $\cD$, then
\begin{align*} 
  \log_2 |C| &\ge H\left(\bigvee_{i=1}^N A_i\right)\\
                            &\ge H(\cD) \cdot \frac{h(1 - (1-p)^N)}{h(p)}\\
                            &\ge H(\cD) \cdot \frac{1/3}{\log_2(3\theta)/\theta}.
\end{align*}
Thus, (\ref{eq:closed-bound}) holds.
\end{proof}

As an immediate corollary of Lemma~\ref{lem:or-closed-sparse-bound}, we can bound the entropy of any $\theta$-sparse distribution in terms of the support's non-redundancy.

\begin{corollary}\label{cor:entropy-Karger}
Let $C \subseteq \{0,1\}^m$ and $\theta \ge 1$. For any $\theta$-sparse distribution $\cD$ over $C$ we have that
\begin{align}
  H(\cD) \le H(C,\theta) := \frac{3\log_2(3\theta)}{\theta}\cdot \NRD(C)\cdot \log_2(m+1).\label{eq:theta-sparse-entropy}
  \end{align}
\end{corollary}

\begin{proof}
By Lemma~\ref{lem:or-closed-sparse-bound}, Lemma~\ref{lem:sauer-shelah}, and Proposition~\ref{prop:nrd-vc}, we have that
\begin{align*}
    H(\cD) &\le \frac{3\log_2(3\theta)}{\theta}\cdot \log_2\card{\spor{C}}\\
    &= \frac{3\log_2(3\theta)}{\theta} \cdot \VC(\spor{C}) \cdot \log_2(m+1)\\
    &= \frac{3\log_2(3\theta)}{\theta} \cdot \NRD(C) \cdot \log_2(m+1)
\end{align*}
can directly apply Proposition~\ref{prop:entropy-bound} to obtain that
\end{proof}

As an informal application, for $d \le \NRD(C)$, let $C_d$ be the set of codewords of weight exactly $d$ in $C$. Recall that the warm-up sparsifier (Theorem~\ref{thm:nrd-spr-warmup}) adds the entire support of $C_d$ to the sparsifier, so we are currently get no nontrivial savings. However, if the uniform distribution $\cD$ over $C_d$ is $\theta$-sparse, then we know that $|C_d| = H(\cD) \le \widetilde{O}(\NRD(C)/\theta)$. As such, by subsampling $\supp(C_d)$ at a rate of $\widetilde{O}_{\eps}(1/\theta)$, we can still get an $\eps$-sparsifier for $C_d$, beating the bound given by Lemma~\ref{lem:NRD-supp-bound}. In the next section, we discuss the scenario in which $\cD$ (or in fact any other distribution on $C_d$) is not $\theta$-sparse. In that case, we show using the minimax theorem that the codewords of $C_d$ have a ``$\theta$-cover.''

\subsection{Minimax}\label{subsec:minimax}

For any $\theta \ge 1$, we say that a code $C \subseteq \{0,1\}^m$ has a $\theta$-cover if there exists a probability distribution $\cQ$ over $[m]$ such that
\begin{align}
\forall c \in C,\ \underset{i \sim \cQ}{\E}[c_i] \ge 1/\theta.
\end{align}

By definition, code $\{\}$ has a $\theta$-cover for all $\theta \ge 1$, while the other trivial code $\{0^m\}$ has no $\theta$-cover for all $\theta \ge 1$. We now use the minimax theorem~\cite{neumann1928theorie} to show that $\theta$-covers and $\theta$-sparse distributions are dual to each other.

\begin{proposition}\label{prop:cover-sparse}
For every $C \subseteq \{0,1\}^m$ and every $\theta \ge 1$, at least one of the following is true.
\begin{itemize}
\item[(1)] $C$ has a $\theta$-cover.
\item[(2)] There exists a probability distribution over $C$ which is $\theta$-sparse.
\end{itemize}
\end{proposition}

\begin{proof}
Since the empty code has a $\theta$-cover for all $\theta \ge 1$, we may assume that $C$ is nonempty. Consider the following zero-sum game. Have Alice pick a codeword $c \in C$ and Bob pick an index $i \in [m]$. The payoff of the game is $c_i$, which Alice is seeking to minimize and Bob is seeking to maximize. By the minimax theorem~\cite{neumann1928theorie}, there exists $\eta \in [0,1]$, a probability distribution $\cP$ over $C$, and a probability distribution $\cQ$ over $[m]$ with the following properties.
\begin{align}
\forall i \in [n], \underset{c \sim \cP}{\E}[c_i] &\le \eta \label{eq:eta-sparse}\\
\forall c \in C, \underset{i \sim \cQ}{\E}[c_i] &\ge \eta\label{eq:eta-cover}
\end{align}
In particular, if $1/\theta \ge \eta$, then (\ref{eq:eta-sparse}) implies that $\cP$ is a $\theta$-sparse distribution supported on $C$. Otherwise, if $1/\theta \le \eta$, then (\ref{eq:eta-cover}) implies that $\cQ$ is a $\theta$-cover of $C$.
\end{proof}

We now use Proposition~\ref{prop:cover-sparse} to prove that if we remove a ``small'' number of codewords from $C$, then the remainder of $C$ has a $\theta$-cover.

\begin{lemma}\label{lem:sparse-removal}
For any $C \subseteq \{0,1\}^m$ and $\theta \ge 1$. Recall that $H(C,\theta)$ is the RHS of (\ref{eq:theta-sparse-entropy}).  There exists $S \subseteq C$ of size at most $\exp_2(H(C,\theta))$ such that $C \setminus S$ has a $\theta$-cover.
\end{lemma}
\begin{proof}
Intuitively, we seek to iteratively delete some codewords of $C$ until Proposition~\ref{prop:cover-sparse} guarantees the remainder of $C$ has a $\theta$-cover. Each codeword removed has a corresponding $\theta$-sparse distribution associated to it. To prove not too many codewords are deleted, we take a (careful) weighted average of the $\theta$-sparse distributions in such a way that every removed codeword has equal probability of being sampled. In that case, the entropy of this $\theta$-sparse distribution is at least the logarithm of the number of codewords removed, so we can bound the number of codewords by Corollary~\ref{cor:entropy-Karger}.

Given a probability distribution $\cD$ over $C$, we let $\mu_{\cD} : C \to R_{\ge 0}$ be the probability density function (PDF) of $\cD$. That is, $\mu_{\cD}(c)$ is the probability that $c$ is sampled from $\cD$. We say that $\mu_{\cD}$ is a $\theta$-sparse PDF if $\cD$ is $\theta$ sparse. Further, we define $\max \mu$ to be the highest probability attained by $\mu$, and $\argmax \mu := \{c \in C : \mu(c) = \max \mu\}$.

If $C$ has a $\theta$-cover, take $S = \emptyset{}$. Otherwise, by Proposition~\ref{prop:cover-sparse}, there is a $\theta$-sparse PDF $\mu_1$ supported on $C$. We now inductively define a sequence of probability distributions $\nu_1, \hdots, \nu_t$, where $t \le |C|$ is decided by a stopping criterion. As our base case, we define $\nu_1 := \mu_1$.

Now repeat the following process. Consider $i \in \{1,2,\hdots\}$. Let $S_i := \argmax \nu_i$. If $C \setminus S_i$ has a $\theta$-cover, we stop, and set $t := i$. Otherwise, if $C \setminus S_i$ has no $\theta$-cover, by Proposition~\ref{prop:cover-sparse}, there exist  a $\mu_{i+1}$ be a $\theta$-sparse PDF on $C \setminus S_i$. By convention, we assume that $\mu_i(c) = 0$ for all $c \in S_i$. Let $p_i := \max \nu_i$, and for every $c \in C \setminus S_i$ define
\[
    q_c := \frac{p_i - \nu_i(c)}{p_i - \nu_i(c) + \mu_{i+1}(c)}.
\]
Note that $q_c \in (0,1]$ because $p_i > \nu_i(c)$ by definition of $S_i$. Let $q := \min_{c \in C \setminus S_i} q_c$. Further define $\nu_{i+1} := (1-q)\nu_i + q\mu_{i+1}$ and $S_{i+1} := \argmax \nu_{i+1}$. We claim that
\begin{align}
    S_{i+1} = S_i \cup \argmin \{q_c : c \in C \setminus S_i\} \supsetneq S_i.\label{eq:grow}
\end{align}

To see why, first note that for all $c \in S_i$, 
\[
    \nu_{i+1}(c) = (1-q)\nu_i(c) + q \mu_{i+1}(c) = (1-q)p_i.
\]
Likewise, for all $c \in \argmin \{q_c : c \in C \setminus S_i\}$, we have that
\begin{align*}
    \nu_{i+1}(c) &= (1-q)\nu_i(c) + q \mu_{i+1}(c)\\
    &= \nu_i(c) + q(\mu_{i+1}(c)-\nu_i(c))\\
    &= \frac{(p_i - \nu_i(c) + \mu_{i+1}(c))\nu_i(c)+(p_i - \nu_i(c))(\mu_{i+1}(c) - \nu_i(c))}{p_i - \nu_i(c) + \mu_{i+1}(c)}\\
    &= \frac{p_i\mu_{i+1}(c)}{p_i - \nu_i(c) + \mu_{i+1}(c)}\\
    &= (1-q)p_i.
\end{align*}
However, for all other $c \in C \setminus S_i$, we have that
\begin{align*}
    \nu_{i+1}(c) &= (1-q)\nu_i(c) + q \mu_{i+1}(c)\\
    &= (1-q_c)\nu_i(c) + q_c\mu_{i+1}(c) + (q_c-q)(\nu_i(c) - \mu_{i+1}(c))\\
    &= \frac{(p_i - \nu_i(c) + \mu_{i+1}(c))\nu_i(c)+(p_i - \nu_i(c))(\mu_{i+1}(c) - \nu_i(c))}{p_i - \nu_i(c) + \mu_{i+1}(c)} + (q_c -q)(\nu_i(c) - \mu_{i+1}(c))\\
    &= (1-q_c)p_i + (q_c -q)(\nu_i(c) - \mu_{i+1}(c))\\
    &< (1-q_c)p_i + (q_c - q)(p_i - 0)\\
    &= (1-q)p_i,
\end{align*}
where the inequality uses the fact that $q_c - q > 0$, $p_i > \nu_i(c)$, and $\mu_{i+1}(c) \ge 0$. Therefore, we have proved (\ref{eq:grow}), so $S_i := \argmax \nu_i$ is strictly increasing as $i$ increases.

Since $C$ is finite in size and the empty code has a $\theta$-cover, this process must terminate in $t \le |C|$ steps. That is, $C \setminus S_t$ has a $\theta$-cover. It suffices to prove that $|S_t| \le 2^{H(C,\theta)}$, where $H(C,\theta)$ is the RHS of (\ref{eq:theta-sparse-entropy}). Since $S_t = \argmax \nu_t$, we have that $\max \nu_t \le \frac{1}{|S_t|}$, so 
\[
  H(\nu_t) = \sum_{c \in \supp(\nu_t)} \nu_t(c) \log_2\left(\frac{1}{\nu_t(c)}\right) \ge \sum_{c \in \supp(\nu_t)} \nu_t(c) \log_2\card{S_t} = \log_2\card{S_t}.
\]
Further, $\nu_t$ is a convex combination of $\theta$-sparse distributions $\mu_1, \hdots, \mu_t$, so $\nu_t$ is also a $\theta$-sparse over $C$. Thus, by Corollary~\ref{cor:entropy-Karger}, we have that
\[
\log_2|S_t| \le H(\nu_t) \le H(C,\theta),
\]
as desired.
\end{proof}

We next utilize Lemma~\ref{lem:sparse-removal} to prove our key decomposition theorem.

\subsection{Decomposition Theorem}

In Khanna, Putterman, and Sudan's linear code sparsification~\cite{khanna2024Code}, their key technical result (their Theorem~2.2) is that after puncturing a ``dense'' collection of coordinates, the remaining code has a sparsity property. Their proof proceeds by recursively removing dense subcodes until the remainder is sparse. However, as we shall see, the situation is a bit more nuanced for non-linear codes. Instead, we alternate between removing sparse and dense portions of our code using Lemma~\ref{lem:sparse-removal}. By keeping track of a suitable monovariant, we can carefully control the tradeoff between the dense and sparse portions of our decomposition. More precisely, by conditioning on $\approx \lambda \NRD(C)$ coordinates to puncture, we save a factor of approximately $\lambda$ in the exponent on the number of codewords remaining.

\begin{theorem}\label{thm:NRD-decomp}
For any $C \subseteq \{0,1\}^m$, integer $d \ge 1$, and real $\lambda \ge 1$, there exists $I \subseteq [m]$ of size at most $2\lambda \NRD(C) \log_2(4m)$ such that $C_{\le d}|_{\overline{I}}$ has at most $m \cdot \exp_2(3d\log_2^2(2m)/\lambda)$ codewords.
\end{theorem}

\begin{proof}
Since deleting codewords can only decrease the non-redundancy, assume without loss of generality that $C = C_{\le d}$. We first handle a few simple edge cases.

\textbf{Case 1, $d \ge \lambda \NRD(C)$.} In this case, we let $I := \emptyset$ and note that by Lemma~\ref{lem:sauer-shelah}, $C$ has at most \[\exp_2(\NRD(C) \cdot \log_2(m+1)) \le \exp_2(d\log_2(m+1)/\lambda) \le m \cdot \exp_2(3d\log_2^2(2m)/\lambda)\]
codewords, as desired.\footnote{This easy case is similar to code-counting bound for the proof of Theorem~\ref{thm:nrd-spr-warmup}.}

\textbf{Case 2, $m \le 2\lambda \NRD(C)\log_2(4m)$.} In this case, we set $I = [m]$ and observe that $C|_{\overline{I}}$ has at most 1 codeword (empty string).

We now move onto the main case.

\smallskip
\textbf{Case 3, $d < \lambda \NRD(C)$ and $m > 2\lambda \NRD(C)\log_2(4m)$.}

Let $\theta := \lambda \NRD(C)/ d > 1$. 
 We inductively build sequences $\emptyset = I_0 \subsetneq \cdots \subsetneq I_t \subseteq [m]$ and $\emptyset = T_0 \subseteq \cdots \subseteq T_t \subseteq C$ for a currently unspecified $t \in \N$ such that for all $j \in \{0,1,\hdots, t\}$,
\begin{align}
  \card{I_j} &= j \label{eq:I-size}\\
  \card{T_j|_{\overline{I_j}}} &\le j \cdot \exp(H(C,\theta)) \label{eq:T-size}
\end{align}

We build the sequence using the following (randomized) procedure.
\begin{itemize}
    \item \textbf{Base Case:} Set $I_0 := \emptyset$ and $T_0 := \emptyset$.
    \item \textbf{For} $j \in \{0, 1, \hdots, m-1\}$,
    \begin{itemize}
    \item \textbf{If} $C = T_j$, \textbf{halt} with $t := j$.
    \item Let $C_j := (C \setminus T_{j})|_{\overline{I_j}}$.
    \item Invoke Lemma~\ref{lem:sparse-removal} to find $S_{j+1} \subseteq C_j$ of size at most $\exp_2(H(C,\theta))$ such that $C_j \setminus S_{j+1}$ has a $\theta$-cover $\cQ_{j+1}$. 
    \item Randomly sample $i_{j+1}$ from $\cQ_{j+1}$.
        \item Set $I_{j+1} := I_j \cup \{i_{j+1}\}$.
    \item Set $T_{j+1} := T_j \cup \{c \in C : c|_{\overline{I_j}} \in S_{j+1}\}$.
    \end{itemize}
\end{itemize}

As written, it seems possible that we may exhaust the \textbf{For} loop without meeting the halt condition. However, we shall prove in Claim~\ref{claim:mono} that with positive probability, we meet the halt condition for some $j = t \le 2 \lambda \NRD(C) \log_2(4m) < m$. Assuming this is the case, it suffices to show that $C|_{\overline {I_t}}$ has at most $m \cdot \exp_2(3d\log_2^2(2m)/\lambda)$ codewords. By the halting condition, we have that $C = T_t$. Therefore, we have that
\[
C|_{\overline{I_t}} = \bigcup_{i=1}^{t} S_i|_{\overline{I_t}}.
\]
Hence, the size of $C|_{\overline{I_t}}$ is at most
\begin{align*}
  \sum_{i=1}^t |S_i| &\le \sum_{i=1}^t H(C, \theta)\\
                     &\le t \cdot \exp_2\left(\frac{3\log_2(3\theta)}{\theta} \cdot \NRD(C) \cdot \log_2(m+1)\right)\\
                     &\le m \cdot \exp_2\left(\frac{3d\log_2(3\lambda \NRD(C)/d))}{\lambda} \cdot \log_2(2m)\right) \ ,
\end{align*}
recalling that $\theta =\tfrac{\lambda \cdot \NRD(C)}{d}$.
By assumption on the size of $m$, we have that
\[
  \frac{3\lambda \NRD(C)}{d} \le \frac{3m}{2\log_2(4m)d} < 2m.
\]
Thus, $C|_{\overline{I_t}}$ has size at most $m \cdot \exp_2(3d\log_2^2(2m)/\lambda)$, as desired. To finish, we prove Claim~\ref{claim:mono}.

\begin{claim}\label{claim:mono}
    With nonzero probability, the procedure halts after at most $2\lambda \NRD(C)\log_2(4m) < m$ steps.
\end{claim}
\begin{proof}
Consider the following monovariant defined for any $I \subseteq [m]$ and $T \subseteq [m]$.
\begin{align}
  f(I, T) := \sum_{c \in C \setminus T} \exp_2(\Ham(c|_{\overline{I}})) \label{eq:mono}
\end{align}
In particular, if $I' \supseteq I$ and $T' \supseteq T$, then $f(I', T') \le f(I, T)$. Observe that since every codeword of $C$ has Hamming weight at most $d$, we have that $|C|$ has size at most $\binom{m}{0} + \cdots + \binom{m}{d} \le (m+1)^d$. Thus,
\[
    f(I_0, T_0) = f(\emptyset, \emptyset) = \sum_{c \in C} \exp_2(\Ham(c)) \le |C| \cdot 2^d \le (2m+2)^d \le \exp(d\log(2m+2)) \ .
\]
Assume that for some $t_0$, the probability that the procedure terminates within $t_0$ steps is zero. For any $j \in \{0,1,\hdots, t_0-1\}$, we have that
\begin{align*}
    \underset{i_{j+1} \sim \cQ_{j+1}}{\E}[f(I_{j+1}, T_{j+1})] &= \underset{i \sim \cQ_{j+1}}{\E}[f(I_j \cup \{i\}, T_{j+1})]\\
    &= \sum_{c \in C_{d} \setminus T_{j+1}} \underset{i \sim \cQ_{j+1}}{\E} [\exp_2(\Ham(c|_{\overline{I_j}})-c_i)]\\
    &=\sum_{c \in C \setminus T_{j+1}} \exp_2(\Ham(c|_{\overline{I_j}}))\underset{i \sim \cQ}{\E} [\exp_2(-c_i)]\\
    &\le \sum_{c \in C \setminus T_{j+1}} \exp(\Ham(c|_{\overline{I_j}}))\left(1-\frac{1}{\theta} + \frac{1}{\theta} \cdot \frac{1}{2}\right)\\
    &\le \left(1 - \frac{1}{2\theta}\right)f(I_j, T_{j+1}) \le \left(1 - \frac{1}{2\theta}\right)f(I_j, T_{j}),
\end{align*}
Thus, with nonzero probability, we have that for all $j \in [t_0-1]$ 
\[
    f(I_j, T_j) \le \left(1 - \frac{1}{2\theta}\right)^j f(I_0, T_0) \le \exp\left(d\log(2m+2) - \frac{j}{2\theta}\right).
\]
In particular, if $j \ge (2\log_2(4m) - 1)\lambda \NRD(C)$, then
\begin{align*}
  f(I_j, T_j) &\le \exp(d\log(2m+2) - (2\log_2(4m) - 1)\lambda \NRD(C)/(2\theta))\\
              &= \exp(d (\log(2m+2) - \log_2(4m) + 1/2))\\
              &< \exp(-d/9)  < 1.
\end{align*}
However, each term in the sum (\ref{eq:mono}) defining $f(I_j,T_j)$ is at least $1$. Therefore, $f(I_j, T_j) = 0$ and $C = T_j$. Thus, the halting condition will be met at step $j$. In other words, with nonzero probability the halting condition must be met in \[\lceil (2\log_2(4m) - 1)\lambda \NRD(C)\rceil \le 2\lambda \NRD(C) \log_2(4m) < m\] steps, as desired. 
\end{proof}
This completes the proof of Theorem~\ref{thm:NRD-decomp}.
\end{proof}

We now have all the ingredients we need to prove our main result.

\subsection{Proof of Theorem~\ref{thm:main-code}}

We now prove Theorem~\ref{thm:main-code}. We use a recursive argument similar to that of Theorem~\ref{thm:nrd-spr-warmup}, which itself was inspired by \cite{khanna2024Code}.

\begin{theorem}[Theorem~\ref{thm:main-code}, more precise version]\label{thm:main-code-precise}
For all $C \subseteq \{0,1\}^m$ and $\eps \in (0,1)$,
\begin{align}
  \SPR(C,\eps) \le 800 \NRD(C) \log^6(4m) / \eps^2.\label{eq:main-code-bound}
\end{align}
\end{theorem}

\begin{proof}
Like in the proof of Theorem~\ref{thm:nrd-spr-warmup}, we prove (\ref{eq:main-code-bound}) by strong induction on $m$. The base cases of $m \le 800$ are trivial, and if $\NRD(C) = 0$, then $C \subseteq \{0^m\}$ so the zero sparsifier of size $0$ works.

Now consider the following parameters:
\begin{align*}
\eps_0 &:= \frac{\eps}{2\log_2 (4m)}\\
\lambda &:= \frac{100\log_2^3 (4m)}{\eps_0^2} = \frac{400\log_2^5 (4m)}{\eps^2}\\
d_0 &:= 2\lambda \log_2(4m)\\
\ell &:= \lceil\log_2 (m/d_0)\rceil\\
p &:= \frac{1}{3}.
\end{align*}
By Lemma~\ref{lem:NRD-supp-bound}, we have that there exists $I_0 \subseteq [m]$ of size at most $d_0 \NRD(C) = 2\lambda \NRD(C) \log_2(4m)$ which contains the support of all codewords $c \in C$ of Hamming weight at most $d_0$. For all $j \in [\ell]$, let
\[
  C_j := \{c \in C : \Ham(c) \in (2^{j-1}d_0, 2^jd_0]\}.
\]

Observe that $C = C_{\le d_0} \cup C_1 \cup \cdots \cup C_\ell$. Now, for each $j \in [\ell]$, apply Theorem~\ref{thm:NRD-decomp} to $C_j$ with $d := 2^j d_0$ and our choice of $\lambda$. As a result, there exists $I_j \subseteq [m]$ such that $|I_j| \le 2\lambda \NRD(C) \log_2(4m)$ and
\begin{align}
\card{C_j|_{\overline{I_j}}} \le m \cdot \exp_2 \left(\frac{3 \cdot 2^j d_0 \log_2^2(2m)}{\lambda}\right) \le m \cdot \exp_2 \left(6 \cdot 2^j \log_2^3(4m)\right).\label{eq:cj-bound}
\end{align}

Let $I = I_0 \cup I_1 \cup \cdots \cup I_{\ell}$. Then, $|I| \le 2(\ell+1)\lambda \NRD(C) \log_2(4m) \le 2\lambda \NRD(C) \log_2^2(4m).$ Further, for all $j \in [\ell]$, we have that $\card{C_j|_{\overline{I}}} \le \card{C_j|_{\overline{I_j}}}$ and so is bounded from above by the RHS of (\ref{eq:cj-bound}).

Let $S_p \subseteq [m] \setminus I$ be a random subset where each element is included independently with probability $p := 1/3$. We claim that with positive probability, we have that $|S_p| \le m/2$ and for all $c \in C$,
\begin{align}
  \frac{1}{p}\Ham(c|_{S_p}) + \Ham(c|_{I}) \in [1-\eps_0, 1+\eps_0] \cdot \Ham(c).\label{eq:use-chernoff}
\end{align}
As in Theorem~\ref{thm:nrd-spr-warmup}, by Lemma~\ref{lem:chernoff-small}, we have that $\Pr[|S| > m/2] \le \exp(-m/30) < 1/100$ since $m > 800$. We now turn to proving (\ref{eq:use-chernoff}). If $c \in C_{\le d_0}$, then $\supp(c) \subseteq I_0 \subseteq I$, so (\ref{eq:use-chernoff}) always holds. Now pick $i \in [\ell]$ and consider $c \in C_j$. By Lemma~\ref{lem:chernoff-ham}, we have that the probability that (\ref{eq:use-chernoff}) does not hold is at most
\[
  2\exp(-\eps_0^2\Ham(c) p / 3) < 2\exp(-\eps_0^22^{j-1}d_0/9) = 2\exp(-(100/9)2^j\log_2^4(4m)).
\]
Naively, we need to now take a union bound of size $\card{C_j}$. However, we only need to take a union bound of size $\card{C_j|_{\overline I}}$. To see why, for a fixed $\widetilde{c} \in C_j|_{\overline I}$ consider any $c \in C_j$ with $c|_{\overline I} = \widetilde{c}$ and $\Ham(c)$ is minimal. Assume that (\ref{eq:use-chernoff}) holds for $c$, for any other $c' \in C_j$ with $c'|_{\overline I} = \widetilde{c}$ we have that
\begin{align*}
  \frac{1}{p}\Ham(c'|_{S_p}) + \Ham(c'|_{I}) &= \frac{1}{p}\Ham(c|_{S_p}) + \Ham(c|_{I}) + (\Ham(c'|_{I}) - \Ham(c|_I))\\
&\in [1-\eps_0,1+\eps_0]\cdot \Ham(c) + (\Ham(c') - \Ham(c))\\
&\subseteq [1-\eps_0,1+\eps_0]\cdot \Ham(c'),
\end{align*}
since $\Ham(c') \ge \Ham(c)$. Therefore by (\ref{eq:cj-bound}), the probability that (\ref{eq:use-chernoff}) fails for any $c \in C_j$ is at most
\begin{align*}
  2\exp(-(100/9)2^j\log_2^4(4m)) \cdot m \cdot \exp_2 \left(6 \cdot 2^j \log_2^3(4m)\right) &\le 2m\exp_2(2^j\log_2^3(4m)\left(6 - 11\log_2(4m)\right))\\
&\le 2m \exp_2(-5\log_2^3(4m))\\
&\le \frac{1}{2m}.
\end{align*}
In particular, the probability that (\ref{eq:cj-bound}) fails for any $c \in C$ is at most $\ell/(2m) \le 1/2$. Thus, there exists $S \subseteq [m] \setminus I$ with $|S| \le m/2$ for which $\frac{1}{p}\Ham(c|_{S}) + \Ham(c|_{I}) \in [1-\eps_0, 1+\eps_0] \cdot \Ham(c)$ for all $c \in C$. To finish, let $\eps' := \eps - 2\eps_0$ and apply the induction hypothesis to $\eps'$-sparsify $C|_S$. Let $w : S \to \R_{\ge 0}$ be the resulting sparsifier. We then build a sparsifier $\widetilde{w} : [m] \to \R_{\ge 0}$ such that
\[
  \forall i \in [m], \widetilde{w}(i) := \begin{cases}
    1 & i \in I\\
    w(i)/p = 3w(i) & i \in S\\
    0 & \text{otherwise}.
\end{cases}
\]
By the same logic as in the proof of Theorem~\ref{thm:nrd-spr-warmup},  for any $c \in C$, we have that
\begin{align*}
  \langle \widetilde{w}, c\rangle &= 3\langle w, c|_{S}\rangle + \Ham(c|_{I})\\
                                  &\in 3[1-\eps',1+\eps']\cdot \Ham(c|_{S}) + \Ham(c|_{I})\\
                                  &\subseteq [1-\eps',1+\eps']\cdot (3 \Ham(c|_{S}) + \Ham(c|_{I}))\\
                                  &\subseteq [1-\eps',1+\eps']\cdot[1-\eps_0,1+\eps_0] \cdot \Ham(c)\\
                                  &\subseteq [1-\eps, 1+\eps] \cdot \Ham(c),
\end{align*}
as desired. To finish, we bound the size of $\widetilde{w}$. In particular,
\begin{align*}
\card{\supp(\widetilde{w})} &= \card{\supp(w)} + \card{I}\\
                            &\le 800 \NRD(C|_{S}) \log_2^6(4\card{\supp(C|_{S})}) / \eps'^2 + 2\lambda \NRD(C)\log_2^2(4m)\\
&\le \NRD(C)\left[\frac{800 \log_2^6(2m)}{\eps^2(1 - 1/\log_2(4m))^2} + \frac{800\log^5_2(4m)}{\eps^2}\right]\\
&= \frac{800 \NRD(C)}{\eps^2} \left[\log_2^4(2m)\log_2^2(4m) + \log^5_2(4m)\right]\\
&\le \frac{800 \NRD(C) \log_2^5(4m)}{\eps^2} [\log_2(2m) + 1]\\
&= \frac{800 \NRD(C) \log_2^6(4m)}{\eps^2} \ . \qedhere
\end{align*}
\end{proof}

\section{Extension to Weighted Sparsification}\label{sec:spr-weighted}

We now prove our main result for sparsifying \emph{weighted} CSPs. More precisely, given a predicate $R \subseteq D^r$, a weighted instance of $\CSP(R)$ can be viewed as an ordinary instance $\Psi := (X, Y \subseteq X^r)$ of $\CSP(R)$ along with a weight function $w : Y \to \R_{> 0}$ for the clauses. Note that we assume all weights are nonnegative, else the set of clauses $Y$ can be reduced in size. Given an assignment $\sigma : X \to D$, recall that its $(R,w)$-weight is defined to be
\[
    \wt(R, \Psi, w, \sigma) := \sum_{y \in Y} w(y) \one[(\sigma(y_1), \hdots, \sigma(y_r)) \in R].
\]
Then, a $(w,\eps)$-sparsifier of $\Psi$ is a function $\widetilde{w} : Y \to \R_{\ge 0}$ such that
\[
(1-\eps)\wt(R, \Psi, w, \sigma) \le \wt(R, \Psi, \widetilde{w}, \sigma) \le (1+\eps)\wt(R, \Psi, w, \sigma).
\]
The $\eps$-sparsity of $(\Psi, w)$ is then the minimal support size of any $(w,\eps)$-sparsifier of $\Psi$. We can now define the weighted sparsity of a predicate.

\begin{definition}
For any $R \subseteq D^r$, $n \in \N$ and $\eps \in (0,1)$, we define $\wSPR(R, n, \eps)$ to be the maximal $\eps$-sparsity of an weighted instance $(\Psi, w)$ of $\CSP(R)$ on $n$ variables.
\end{definition}

Likewise, given a code $C \subseteq \{0,1\}^m$ and a weight function $w : [m] \to \R_{> 0}$, we say that $\widetilde{w} : [m] \to \R_{\ge 0}$ is $(w,\eps)$-sparsifier of $C$ is for all $c \in C$,
\begin{align}
    (1-\eps)\langle w, c\rangle \le \langle \widetilde{w}, c\rangle \le (1+\eps)\langle w, c\rangle,\label{eq:w-eps-spars}
\end{align}
with the $\eps$-sparsity of $(C,w)$ being the minimum support size of any $(w,\eps)$-sparsifier. This leads to an analogous definition of weighted code sparsity.
\begin{definition}
For any $C \subseteq \{0,1\}^m$ and $\eps \in (0,1)$, we define $\wSPR(C, \eps)$ to be the maximum $\eps$-sparsity of $(C,w)$ among all weight functions $w : [m] \to \R_{> 0}$.
\end{definition}

\begin{proposition}\label{prop:wspr-code-csp}
For all $R \subseteq D^r$, $n \in \N$, and $\eps \in (0,1)$, we have that
\[
    \wSPR(R, n, \eps) = \max_{\substack{\emph{$\Psi$ instance of $\CSP(R)$}\\\emph{on $n$ variables}}} \wSPR(C_{R,\Psi}, \eps),
\]
with $C_{R, \Psi}$ defined as in Section~\ref{subsec:implies}.
\end{proposition}
\begin{proof}
The proof is essentially identical to the proof of (\ref{eq:equiv-SPR}) in Proposition~\ref{prop:csp-to-code}.\end{proof}

It is easy to see that 
\begin{align*}
\wSPR(R, n, \eps) &\ge \SPR(R, n, \eps) \ge \NRD(\overline{R}, n).
\end{align*}
However, we can prove a stronger (tight!) lower bound for both in terms of \emph{chain length}. We denote the chain length of a code $C$ by $\CL(C)$ and the chain length of a CSP predicate $R$ by $\CL(R, n)$ for instances with $n$ variables. These terms are defined in Section~\ref{subsec:cl-code} and Section~\ref{subsec:cl-csp}, respectively. However, we can immediately state the main results of this section.

\begin{theorem}\label{thm:wspr-code}
For all $C \subseteq \{0,1\}^m$ and $\eps \in (0,1)$, we have that
\[
    \CL(C) \le \wSPR(C, \eps) = O(\CL(C)(\log m)^6/\eps^2).
\]
\end{theorem}

\begin{theorem}\label{thm:wspr-csp}
For all $R \subseteq D^r$, $n \in \N$ and $\eps \in (0,1)$, we have that
\[
    \CL(\overline{R}, n)-1 \le \wSPR(R, n, \eps) = O(\CL(\overline{R},n)(r\log n)^6/\eps^2)
\]
\end{theorem}

\subsection{Chain Length of Codes}\label{subsec:cl-code}

As it is more elementary to state, we begin by defining the chain length of a code.

\begin{definition}\label{def:chain-code}
Let $C \subseteq \{0,1\}^m$ be a code. A \emph{chain} of length $\ell$ is a pair of injective maps $a : [\ell] \to [m]$ and $c : [\ell] \to C$ such that the following conditions hold.
\begin{align*}
\forall i \in [\ell], c(i)_{a(i)} &= 1.\\
\forall 1 \le i < j \le \ell, c(j)_{a(i)} &= 0.
\end{align*}
The \emph{chain length} of $C$, denoted by $\CL(C)$ is the maximum-length chain.
\end{definition}

 Recall that if we line up the codewords of $C$ as rows of an $|C| \times m$ matrix and allow arbitrary column permutations, $\NRD(C)$ is the dimension of the largest identity submatrix within that matrix. In this setup, $\CL(C)$ is the dimension of the largest upper triangular square submatrix with $1$'s on the diagonal. When $C$ is a linear code, both $\NRD(C)$ and $\CL(C)$ equal the dimension of $C$. Further, for nontrivial $C$, we have that $1 \le \NRD(C) \le \CL(C) \le m$. 

If the set of indices $a : [\ell] \to [m]$ is clear from context, we refer to just the list of codewords $c : [\ell] \to C$ as the chain. We now prove that chain-length of a code is always at least its non-redundancy. 
\begin{proposition}\label{prop:cl-ge-nrd}
For all $C \subseteq \{0,1\}^m$, $\CL(C) \ge \NRD(C)$.
\end{proposition}
\begin{proof}
Let $I \subseteq [m]$ be a maximum-sized non-redundant index set of $C$. Let $c : I \to C$ witness that $I$ is non-redundant. Let $\ell := |I|$. For any map bijection $a : [\ell] \to I$, we have that $(a, c \circ a)$ is a chain of length $\ell$. 
\end{proof}

We now consider an example where $\CL(C)$ and $\NRD(C)$ can be as different as possible.
\begin{example}\label{example:chain} Consider the code $C := \{1^n0^{m-n} : n \in [m]\}.$
It is straightforward to check that the defining codewords of $C$ form a chain, so $\CL(C) = m$. However, $\NRD(C) = 1$ because for any $c, c' \in C$ either $\supp(c) \subseteq \supp(c')$ or $\supp(c') \subseteq \supp(c)$. Therefore, by Theorem~\ref{thm:main-code}\footnote{This can be done in a more elementary manner by subsampling the $i$th coordinate with probability $p_i := \min(1, O(\tfrac{\log m}{i\eps^2}))$ and giving it a weight of $1/p_i$ if kept.} $\SPR(C, \eps) \le (\log^{O(1)} m)/\eps^2$. However, if we consider the weighting $w(i) = 2^i$, then any pair of codewords has weight different by a factor of $2$. Hence, for $\eps < 1/2$, every $\eps$-sparsifier must keep every coordinate of the code or else two codewords will be given identical weight by the sparsifier. Therefore, $\wSPR(C, \eps) = m$.
\end{example}

We now formalize this example into a lower bound of $\wSPR(C, \eps)$ for every code $C$.

\begin{lemma}\label{lem:wspr-cl-lb}
For every $C \in \{0,1\}^m$ and $\eps \in (0,1)$, we have that $\wSPR(C, \eps) \ge \CL(C)$.
\end{lemma}

\begin{proof}
Let $\ell := \CL(C)$. By permuting $[m]$, we may assume without loss of generality that there exist $c(1), \hdots, c(\ell) \in C$ such that $c(i)_j = 1$ if $i = j$ and $c(i)_j = 0$ if $j < i$. Pick $\lambda := \tfrac{4}{1-\eps}$ and consider the following weight function $w : [m] \to \R_{> 0}$.
\[
    w(i) := \begin{cases}
    \lambda^{\ell+1-i}& i \in [\ell]\\
    \frac{1}{m} & \text{otherwise}
    \end{cases}
\]
Observe then that for each $i \in [\ell]$, that
\[
\langle w, c(i)\rangle \ge w(i)c(i)_i = \lambda^{\ell+1-i}, \text{ and }
\]
\[
\langle w, c(i)\rangle \le \sum_{j=i}^{m} w(i)
= \sum_{j=i}^{\ell} \lambda^{\ell+1-j} + \frac{m-\ell}{m}
<  1 + \lambda + \cdots + \lambda^{\ell+1-i}
< \frac{1}{1 - 1/\lambda} \cdot \lambda^{\ell+1-i} \ .
\]
Now consider any $\eps$-sparsifier $\widetilde{w} : [m] \to \R_{\ge 0}$. Let $S := \supp(\widetilde{w})$ and for each $i \in [\ell]$, let $S_i := S \cap \supp c(i)$. We seek to prove that $|S| \ge \ell$. To prove this, it suffices to show for all $i \in [\ell]$ that
\begin{align}
    S_i \setminus \bigcup_{j=i+1}^{\ell} S_i \neq \emptyset \label{eq:S-i-big}
\end{align}
First, observe that since $\langle w, c(\ell)\rangle$ has nonzero weight, so $S_\ell \neq \emptyset$. Now assume for sake of contradiction that for some $i \in [\ell-1]$, has $S_i \subseteq S_{i+1} \cup \cdots \cup S_\ell$. Then, since $\widetilde{w}$ is an $\eps$-sparsifier and $\lambda := 4/(1-\eps)$,
\begin{align*}
4\lambda^{\ell-i} &= (1-\eps)\lambda^{\ell+1-i}
\le (1-\eps)\langle w, c(i)\rangle
\le \langle \widetilde{w}, c(i)\rangle\\
&= \sum_{a \in S_i} \widetilde{w}(a) &\text{ (definition of $S_i$)}\\
&\le \sum_{j=i+1}^{\ell} \sum_{a \in S_j} \widetilde{w}(a)
= \sum_{j=i+1}^{\ell} \langle \widetilde{w}, c(j)\rangle\\
&\le (1+\eps) \sum_{j=i+1}^{\ell} \langle w, c(j)\rangle\\
&\le (1+\eps) \sum_{j=i+1}^{\ell} \frac{1}{1 - 1/\lambda} \cdot \lambda^{\ell+1-j}\\
&\le \frac{1+\eps}{(1 - 1/\lambda)^2} \lambda^{\ell-i}
= \frac{4^2(1+\eps)}{(3+\eps)^2} \lambda^{\ell-i}
\le \frac{32}{9}\lambda^{\ell-i}\\
&< 4\lambda^{\ell-i},
\end{align*}
a contradiction. Therefore (\ref{eq:S-i-big}) holds for all $i \in [\ell]$, so $|S| \ge \ell$. Therefore, $\wSPR(C) \ge \CL(C)$.
\end{proof}

\subsection{Chain Length of CSPs}\label{subsec:cl-csp}

We define the chain length of a CSP predicate in terms of sets of satisfying assignments. Recall for an instance $\Psi$ of $\CSP(R)$ on $n$ variables, $\sat(R, \Psi)$ is the set of satisfying assignments to the $\Psi$. Lagerkvist and Wahlstr{\"o}m~\cite{lagerkvist2020Sparsification} first defined the concept, although our definition is closer to that of Bessiere, Carbonnel, and Katsirelos~\cite{bessiere2020Chain}.

\begin{definition}[\cite{lagerkvist2020Sparsification,bessiere2020Chain}]\label{def:chain-csp}
Given $R \subseteq D^r$ and $n \in \N$, we define a \emph{chain} to be a sequence of instances $\Psi_1, \hdots, \Psi_\ell$ of $\CSP(R)$ on $n$ variables such that
\begin{align}
    \sat(R, \Psi_1) \subsetneq \sat(R, \Psi_2) \subsetneq \cdots \subsetneq \sat(R, \Psi_\ell).\label{eq:chain-csp}
\end{align}
The \emph{chain length} of $\CSP(R)$ on $n$ variables, denoted by $\CL(R, n)$, is the maximum length $\ell$ of such a chain.
\end{definition}

It is not immediately obvious how Definition~\ref{def:chain-csp} relates to Definition~\ref{def:chain-code}. We prove this as follows
\begin{proposition}\label{prop:cl-code-csp}
For all $R \subseteq D^r$ and $n \in \N$, we have that
\begin{align}
    \CL(\overline{R}, n) - 1= \max_{\substack{\emph{$\Psi$ instance of $\CSP(R)$}\\\emph{on $n$ variables}}} \CL(C_{R,\Psi}),\label{eq:cl-code-csp}
\end{align}
with $C_{R, \Psi}$ defined as in Section~\ref{subsec:implies}.
\end{proposition}
\begin{remark}
Note that chain lengths for CSPs and codes are off by one for the same reason that a path graph with $n$ vertices has only $n-1$ edges. In particular, the code's chain length is effectively counting the number of $\subsetneq$'s in (\ref{eq:chain-csp}).
\end{remark}
\begin{proof}
We first prove that the LHS of (\ref{eq:cl-code-csp}) is at least the RHS. Fix an instance $\Psi := (X, Y \subseteq X^r)$ with $|X| = n$ and let $\ell := \CL(C_{R,\Psi})$. Thus, there exists $y : [\ell] \to Y$ and $c : [\ell] \to C_{R,\Psi}$ such that for all $i \in [\ell]$, $c(i)_{y(i)} = 1$ and for all $1 \le i < j \le \ell$, $c(j)_{y(i)} = 0$.

Let $\sigma_1, \hdots, \sigma_\ell : X \to D$ correspond to $c(1), \hdots, c(j)$. In particular, for all $i \in [\ell]$, $\sigma_i(y(i)) \in R$ but for all $1 \le i < j \le \ell$, $\sigma_j(y(i)) \not\in R$. For all $i \in \{0,1,\hdots, \ell\}$ define the instance
\[
    \Psi_i := (X, \{y(j) : j \in [i]\}).
\]
as an instance of $\CSP(\overline{R})$. Since adding clauses can only decrease the number of satisfying assignments, we have that
\begin{align}
    \sat(\overline{R}, \Psi_\ell) \subseteq \sat(\overline{R}, \Psi_{\ell-1}) \subseteq \cdots \subseteq \sat(\overline{R}, \Psi_1) \subseteq \sat(\overline{R}, \Psi_0)\label{eq:csp-chain}
\end{align}
To prove these inclusions are strict, for each $i \in [\ell]$,
Observe that
\[
    \sigma_{i} \in \sat(\overline{R}, \Psi_{i-1}) \setminus \sat(\overline{R}, \Psi_{i}),
\]
as $\sigma_{i}(y(1)), \hdots, \sigma_{i+1}(y(i-1)) \not\in R$ but $\sigma_{i+1}(y(i)) \in R$. Therefore, all inclusions in (\ref{eq:csp-chain}) are strict, so $\CL(\overline{R}, n) \ge \CL(C_{R,\Psi})+1$.

We now prove that the LHS of (\ref{eq:cl-code-csp}) is at most the RHS. Let $\ell := \CL(\overline{R}, n) - 1$. Fix a set $X$ of size $n$ and consider instances $\Psi_0 := (X, Y_0), \hdots, \Psi_\ell := (X, Y_\ell)$ of $\CSP(\overline{R})$ on the variable set $X$ such that
\[
\sat(\overline{R}, \Psi_\ell) \subsetneq \sat(\overline{R}, \Psi_{\ell-1}) \subsetneq \cdots \subsetneq \sat(\overline{R}, \Psi_1) \subsetneq \sat(\overline{R}, \Psi_0)
\]
For each $i \in [\ell]$, pick $\sigma_i \in \sat(\overline{R}, \Psi_{i-1}) \setminus \sat(\overline{R}, \Psi_{i})$. Since $\sigma_i \not\in \sat(\overline{R}, \Psi_i)$, there exists $y(i) \in Y_i$ which $\sigma_i$ does not satisfy.

Now consider $Y := \{y(1), \hdots, y(\ell)\}$ and let $\Psi := (X,Y)$. It suffices to prove that $\CL(C_{R,\Psi}) \ge \ell$. For each $i \in [\ell]$, there exists $c(i) \in C_{R,\Psi}$ corresponding to the assignment $\sigma_i$. Since $\sigma_i \in \sat(\overline{R}, \Psi_{i-1}) \setminus \sat(\overline{R}, \Psi_{i})$, we have that $c(i)_{y(1)} = \cdots c(i)_{y(i-1)} = 0$ but $c(i)_{y(i)} = 1$. Therefore, $(y : [\ell] \to Y, c : [\ell] \to C_{R,\Psi})$ is a chain of length $\ell$ in $C_{R,\Psi}$.
\end{proof}

As an immediate corollary of the proof, we have the following more combinatorial interpretation of chain length. This is somewhat closer to Lagerkvist and Wahlstr{\"o}m's~\cite{lagerkvist2020Sparsification} definition of chain length.

\begin{corollary}\label{cor:chain-csp}
Let $R \subseteq D^r$ be a relation and $X$ a variable set of size $n$. We have that $\CL(R, n) \ge \ell + 1$ if and only if there exist assignments $\sigma_1, \hdots, \sigma_\ell : X \to D$ and clauses $y_1, \hdots, y_\ell \in X^r$ such that for all $i \in [\ell]$,
\[
  \sigma_i(y_1), \hdots, \sigma_i(y_{i-1}) \in R \text{ and } \sigma_i(y_i) \not\in R.
\]
\end{corollary}

We can now show that Theorem~\ref{thm:wspr-code} implies Theorem~\ref{thm:wspr-csp}.

\begin{proposition}
Theorem~\ref{thm:wspr-code} implies Theorem~\ref{thm:wspr-csp}.
\end{proposition}
\begin{proof}
Fix $R \subseteq D^r, n \in \N, \eps > 0$. Apply Theorem~\ref{thm:wspr-code} to $C_{R,\Psi}$ for all instances $\Psi$ of $\CSP(R)$ on $n$ variables. Then, take the maximum of these resulting inequalities for all such $\Psi$. By applying Proposition~\ref{prop:wspr-code-csp} and Proposition~\ref{prop:cl-code-csp}, we obtain Theorem~\ref{thm:wspr-csp}.
\end{proof}

\subsection{Proof of Theorem~\ref{thm:wspr-code}}

We now prove Theorem~\ref{thm:wspr-code}. We use a bucketing approach by \cite{khanna2024Characterizations} to split the weights into groups of roughly equal weight (within a factor of $m^{O(1)}$). We then use a standard repetition trick and Theorem~\ref{thm:main-code} to sparsify these groups. Finally, we argue that the total size of the sparsifier cannot be significantly more than the length of the maximal chain in our code.

As a warmup, we first prove that non-redundancy is an upper bound if the ratio between the maximum and minimum weight is at most $m^{O(1)}$. (this observation is also in \cite{khanna2024Characterizations} for linear codes).

\begin{lemma}\label{lem:rep-wspr}
Let $C \subseteq \{0,1\}^m$ and $w : [m] \to R_{> 0}$ be a positive weight function such that $\max(w) / \min(w) \le m^3$. Then, for all $\eps \in (0,1)$, $C$ has a $(w,\eps)$-sparsifier of size at most
\[
    10^7 \NRD(C)\log^6(4m)/\eps^2.
\]
\end{lemma}
\begin{proof}
If $\NRD(C) = 0$, return the empty sparsifier. If $\eps \le 1/m$, return $w$ itself as $m \le 1/\eps^2$. Thus, assume $\NRD(C) \ge 1$ and $\eps > 1/m$. For each $i \in [m]$ define
\[
    b_i := \left\lfloor \frac{2 w(i)}{\eps \min w}\right\rfloor \le 2m^5.
\]
Let $\widetilde{m} = \sum_{i=1}^m b_i \le 2m^6.$ Let $f : \{0,1\}^m \to \{0,1\}^{\widetilde{m}}$ be such that
\[
    f(x) := \underbrace{x_1,\hdots,x_1}_{b_1}, \underbrace{x_2,\hdots,x_2}_{b_2}, \hdots, \underbrace{x_m,\hdots,x_m}_{b_m}.
\]
We also let $S_i \subseteq [\widetilde{m}]$ be the set of $b_i$ coordinates which are equal to $x_i$. Observe that for all $x \in \{0,1\}^m$, we have that
\begin{align*}
\Ham(f(x)) &= \sum_{i=1}^n \left\lfloor \frac{2w(i)}{\eps \min w}\right\rfloor x_i\\
&\in\sum_{i=1}^n \left[\frac{2w(i)}{\eps \min w} - 1, \frac{2 w(i)}{\eps \min w}\right]\cdot x_i\\
&\subseteq \sum_{i=1}^n [1-\eps/2, 1] \cdot \frac{w(i)x_i}{(\eps/2)\min(w)}&\text{(since $w(i) \ge \min(w)$)}\\
&\subseteq [1-\eps/2, 1] \cdot \frac{\langle w, x\rangle}{(\eps/2)\min(w)}.
\end{align*}

Now define $\widetilde{C} := \{f(c) : c \in C\}.$ Note that $\NRD(\widetilde{C}) = \NRD(C)$ as the repetition of coordinates cannot increase non-redundancy. Let $\widetilde{w} : [\widetilde{m}] \to \R_{\ge 0}$ be an $(\eps/2)$-sparsifier of $\widetilde{C}$. By Theorem~\ref{thm:main-code}, we have that
\[
    \card{\supp(\widetilde{w})} \le 800\NRD(\widetilde{C}) \log_2^6(4\widetilde{m})/(\eps/2)^2 \le 10^7 \NRD(C) \log_2^6(4m) / \eps^2.
\]
Now consider the map $w' : [m] \to R_{\ge 0}$ defined by
\[
   \forall i \in [m], w'(i) := \frac{\eps \min(w)}{2} \sum_{j \in S_i} \widetilde{w}(j).
\]
Clearly, $\card{\supp(w')} \le \card{\supp(\widetilde{w})}$ so it suffices to prove that $w'$ is a $(w,\eps)$-sparsifier of $C$. To see why, for all $c \in C$,
\begin{align*}
\langle w', c\rangle &= \frac{\eps \min(w)}{2} \langle \widetilde{w}, f(c)\rangle\\
&\in [1-\eps/2,1+\eps/2]\cdot \frac{\eps \min(w)}{2} \Ham(f(c))\\
&\subseteq [1-\eps/2, 1+\eps/2] \cdot [1-\eps/2, 1] \cdot \langle w, x\rangle\\
&\subseteq [1-\eps, 1+\eps] \cdot \langle w, x\rangle,
\end{align*}
as desired.
\end{proof}

\begin{proof}[Proof of Theorem~\ref{thm:wspr-code}.]
Fix our code $C \subseteq \{0,1\}^m$, our positive weight function $w : [m] \to \R_{> 0}$, and $\eps \in (0,1)$. Similar to the proof of Lemma~\ref{lem:rep-wspr}, we may assume that $\NRD(C) \ge 1$, $\eps > 6/m$, and $m \ge 2$.

Define a function $t : [m] \to \Z$ be
\[
    t(i) := \left\lfloor \frac{\log w(i)}{3\log m}\right\rfloor.
\]
Let $T := \{t(i) : i \in \Z\}$ be a finite set. For each $c \in C$, define its \emph{type} to be
\[
    \type(c) := \max_{i \in \supp(c)} t(i) \in T.
\]
For each $t \in \Z$, define
\begin{align*}
    I_t &:= \{i \in [m] : t(i) = t\},\text{ and}\\
    C_t &:= \{c \in C : \type(c) = t\}.
\end{align*}
Note that these sets are only nonempty for $t \in T$. For each $t \in T$, observe that
\[
    \frac{\max w|_{I_t}}{\min w|_{I_t}} \le \frac{m^{3t+3}}{m^{3t}} = m^3.
\]

For each $t \in T$, let $\widetilde{w}_t : [m] \to \R_{\ge 0}$ be an $(\eps/2)$-sparsifier of $(C_{t} \cup C_{t+1} \cup \{1_{I_t}\})|_{I_t}$, where $1_{I_t}\in \{0,1\}^m$ is the indicator of the set $I_t$. If $t \not\in T$, define $\widetilde{w}_t : [m] \to \R_{\ge 0}$ to be the zero map.

For all $t \in T$, note that $(C_{t} \cup C_{t+1})|_{I_t}$ has at least one nonzero codeword. Therefore,
\[
    \NRD((C_{t} \cup C_{t+1} \cup \{1_{I_t}\})|_{I_t})\le \NRD((C_{t} \cup C_{t+1})|_{I_t}) + 1 \le 2 \NRD((C_{t} \cup C_{t+1})|_{I_t}).
\]
Thus, by Lemma~\ref{lem:rep-wspr}, we can ensure that
\begin{align}
    \card{\supp(\widetilde{w}_t)} &\le 8\cdot 10^7 \NRD((C_{t} \cup C_{t+1})|_{I_t})\log^6(4m)/\eps^2.\label{eq:w-t}
\end{align}
Now consider $\widetilde{w} : [m] \to \R_{\ge 0}$ defined by
\[
    \forall i \in [m], \widetilde{w}(i) := \widetilde{w}_{t(i)}(i).
\]
We claim that $\widetilde{w}$ is an $(w,\eps)$-sparsifier of $C$. To verify this, for any $c \in C$ with $t := \type(c)$, we have that
\begin{align*}
\langle \widetilde{w}, c\rangle &= \langle \widetilde{w}_{t-1}, c|_{I_{t-1}}\rangle + \langle \widetilde{w}_t, c|_{I_t}\rangle + \sum_{\substack{t' \in T\\t' \le t-2}} \langle \widetilde{w}_{t'}, c|_{I_{t'}}\rangle\\
&\in [1-\eps/2,1+\eps/2] \cdot \sum_{i \in I_{t-1} \cup I_t} w(i)c_i + \left[0, \sum_{\substack{t' \in T\\t' \le t-2}} \langle \widetilde{w}_{t'}, 1_{I_{t'}}\rangle\right]\\
&= [1-\eps/2, 1+\eps/2] \cdot \left[\langle w, c\rangle - \sum_{i \in \supp(c) \setminus (I_{t-1} \cup I_t)} w(i)\right]\\
&\ \ \ \ \ \ + (1 + \eps/2)\left[0, \sum_{\substack{t' \in T\\t' \le t-2}} \langle w, 1_{I_{t'}}\rangle\right]
\end{align*}
Since $t = \type(c)$, we have that $\langle w, c\rangle \ge m^{3t}$. Further, observe that
\[
\sum_{i \in \supp(c) \setminus (I_{t-1} \cup I_t)} w(i) \le \sum_{\substack{t' \in T\\t' \le t-2}} \sum_{i \in I_{t'}} w(i) \le m \cdot m^{3t-3} \le  m^{3t} \cdot \frac{\eps}{6} \le \frac{\eps}{6}\langle w, c\rangle,
\]
where we used the fact that $\eps > 6/m$. Therefore, 
\begin{align*}
\langle \widetilde{w}, c\rangle &\in ([1-\eps/2, 1+\eps/2] \cdot [1-\eps/6, 1] + (1+\eps/2)\cdot [0, \eps/6]) \cdot \langle w, c\rangle \subseteq [1-\eps, 1+\eps] \langle w, c\rangle,
\end{align*}
as desired. It thus suffices to bound $\card{\supp(\widetilde{w})}$. By (\ref{eq:w-t}), we have that
\begin{align}
\card{\supp(\widetilde{w})} &\le \sum_{t \in T} \card{\supp(\widetilde{w}_t)} \le 8\cdot 10^7 \left[\sum_{t \in T} \NRD((C_t\cup C_{t+1})|I_t)\right]\log^6(4m)/\eps^2.\label{eq:bound-wt}
\end{align}
To finish, we prove the following claim.
\begin{claim}\label{claim:NRD-CL-bound}
$\sum_{t \in T} \NRD((C_t\cup C_{t+1})|_{I_t}) \le 2\CL(C).$
\end{claim}
\begin{proof}
For each $t \in T$, let $j_t := \NRD((C_t \cup C_{t+1})|_{I_t})$ and let $i^{(t)}_1, \hdots, i^{(t)}_{j_t} \in I_t$ be distinct indices such that there exist $c^{(t)}(1), \hdots, c^{(t)}(j_t) \in C_t \cup C_{t+1}$ such that
\[
    \forall a,b \in [j_t], c^{(t)}(a)_{i^{(t)}_{b}} = \one[a = b].
\]
We can view this non-redundant set as a chain in $C$. Call this chain $\cC_t$. The key observation is that if $t, t' \in T$ are such that $t \ge t'+2$, then the concatenation of $\cC_t$ and $\cC_{t'}$ is a chain. Since $\cC_t$ and $\cC_{t'}$ are each chains in $C$, it suffices to prove for any $a \in [j_t]$ an $b \in [j_{t'}]$ that
\begin{align}
    c^{(t')}(b)_{i_a^{(t)}} = 0.\label{eq:c-zero}
\end{align}
To see why, for any $c \in C_{t'} \cup C_{t'+1}$, we have that $\type(c) \le t'+1 < t$, so $\supp(c) \cap I_t = \emptyset$. Therefore, (\ref{eq:c-zero}) holds.

Thus, for any decreasing subsequence $t_1, \hdots, t_\ell$ of $T$ with consecutive terms having difference at least $2$, we have that the concatenation of $\cC_{t_1}, \hdots, \cC_{t_\ell}$ is a chain in $C$, so $j_{t_1} + \cdots + j_{t_\ell} \le \CL(C)$. Since $T$ can partitioned into two such subsequences (e.g., even/odd), we have that $j_1 + \cdots + j_\ell \le 2\CL(C)$, as desired.
\end{proof}
Applying Claim~\ref{claim:NRD-CL-bound} to (\ref{eq:bound-wt}) proves the theorem.
\end{proof}

\begin{remark}
A mistake in the proof of Theorem~\ref{thm:wspr-code} was caught while preparing the follow-up work \cite{BGP26b}. In particular, the initial (incorrect) version of the proof claimed that $\langle \widetilde{w}, c\rangle = \langle \widetilde{w}_{t-1}, c|_{I_{t-1}}\rangle + \langle \widetilde{w}_t, c|_{I_t}\rangle$, which is only approximately true. The precise error terms are now tracked and bounded in the corrected proof.
\end{remark}

\begin{remark}
Alternatively, we could bound the sum in (\ref{eq:bound-wt}) by $|T|\NRD(C)$. Since $|T| \le 2 + (\log \frac{\max w}{\min w})/(3\log m),$ we get the following immediate corollary which gives a savings if $\max w/\min w$ is small. It can also be seen an interpolation between Theorem~\ref{thm:main-code} and Theorem~\ref{thm:wspr-code}.
\end{remark}

\begin{corollary}\label{cor:wspr-improved}
Let $C \subseteq \{0,1\}^m$ and let $w : [m] \to \R_{\ge 0}$ be a weight function. For all $\eps \in (0,1)$, there exist a $(w,\eps)$-sparsifier of $C$ of size at most
\[
    O\left(\log \frac{m\max w}{\min w} \cdot \NRD(C)\log^5 m/\eps^2\right).
\]
\end{corollary}

\begin{remark}
As observed in \cite{BGP26b}, one can recursively apply Theorem~\ref{thm:wspr-code} to itself to get a better sparsifier whose size is independent of the length $m$ of the codewords. We defer to \cite{BGP26b} for details on how this can be done.
\end{remark}

\section{Immediate Applications of Non-redundancy and Chain Length to Sparsification}\label{sec:spr-app}

In this section, we prove many corollaries of Theorem~\ref{thm:main-code}, Theorem~\ref{thm:main}, and Theorem~\ref{thm:wspr-code-intro} based on established facts in the non-redundancy literature.

\subsection{Code Sparsification for Non-Abelian Groups}\label{subsec:non-abelian}

Given a finite group $G$ with identity element $e$, we define a $G$-code to be a subgroup $H \le G^m$ for some $m \in \N$. Define the Hamming weight of $h \in H$ to be the number of coordinates $i \in [m]$ not equal to $e$. A \emph{$G$-code $\eps$-sparsifier} is a map $w : [m] \to \R_{\ge 0}$ such that for every $h \in H$, the $w$-weight of $h$ (i.e., the sum of the weights of the non-identity coordinates) is within a factor of $1\pm \eps$ of its Hamming weight.

The works of Khanna, Putterman, and Sudan~\cite{khanna2024Code,khanna2024Characterizations} proved the existence (and in fact efficient construction) of $G$-code $\eps$-sparsifiers of size $\log |H| \log^{O(1)}m/\eps^2$ for any Abelian group $G$. In particular $G = \Z/2\Z$ captures cut sparsification. We show as an immediate corollary of Theorem~\ref{thm:main-code} that comparably-sized sparsifiers exist for any \emph{non-Abelian} group $G$. The proof adapts techniques from \cite{lagerkvist2020Sparsification}.

\begin{theorem}\label{thm:group-sparsifier}
For any finite group $G$, any $H \le G^m$, and any $\eps \in (0,1)$, there is an $\eps$-sparsifier of $H$ of size $O(\log |H| (\log^6 m)/\eps^2)$.
\end{theorem}

\begin{proof}
If $|G| = 1$, there is nothing to prove, so assume $|G| \ge 2$. Let $C \subseteq \{0,1\}^m$ be the code such that for every $h \in H$, the binary string $(\one[h_i \neq e] : i \in [m])$ is an element of $C$. Observe that any $\eps$-sparsifier of $C$ is an $\eps$-sparsifier of $H$ (and vice-versa). Therefore, by Theorem~\ref{thm:main-code}, it suffices to prove that $\NRD(C) \le \log_2 |H|$.

Let $I := \{i_1, \hdots, i_\ell\} \subseteq [m]$ be non-redundant. Thus, there exists $h_1, \hdots, h_\ell$ such that $h_{j,i_{j'}} \neq e$ if and only if $j = j'$. Now consider the set
\[
  S := \left\{h_1^{b_1} \cdots h_\ell^{b_\ell} \in H: b \in \{0,1\}^\ell\right\}.
\]
Note that for each $b \in \{0,1\}^\ell$, the coordinates within $I$ of $h_1^{b_1} \cdots h_\ell^{b_\ell}$ which do not equal $e$ are precisely the support of $b$. Therefore, $|S| = 2^\ell$. Thus, $\ell \le \log_2 |H|$, so $\NRD(C) \le \log_2|H|$, as desired.
\end{proof}

For non-Abelian groups $G$, Section~7 of \cite{khanna2024Characterizations} observed that for any two non-commuting elements $a,b \in G$ that the predicate $P := \{x \in \{0,1\}^4 : a^{x_1}b^{x_2}a^{-x_3}b^{-x_4} = e\}$ has $\NRD(P, n) \ge \Omega(n^2)$, suggesting that sparsification for (Abelian) affine CSPs does not extend to non-Abelian groups. However, this is not exactly the case---there are predicates corresponding to non-Abelian code sparsification, falling under the more general framework of ``Mal'tsev predicates.'' See Section~\ref{subsec:maltsev} for more details.

\subsection{Maximum-size Dichotomy}\label{subsec:size}

Observe that for any $R \subseteq D^r$, we have that $\NRD(R, n), \SPR(R, n, \eps) \le n^r$ as there are at most $n^r$ possible clauses on $n$ variables. A natural question is then which predicates have non-redundancy/sparsification of size $\Omega(n^r)$. When $D = \{0,1\}$, it is known that $\NRD(R, n) = \Omega(n^r)$ if and only if $R = \{0,1\}^r \setminus \{b\}$ for some $b \in \{0,1\}^r$, otherwise $\NRD(R, n) = O(n^{r-1})$~\cite{chen2020BestCase}. Using linear code sparsification, the analogous result is also known for sparsification~\cite{khanna2024Characterizations}. In particular, $\SPR(R, n, \eps) = \Omega(n^r)$ if and only if $R = \{b\}$ for some $b \in \{0,1\}^r$, otherwise $\SPR(R, n, \eps) = \widetilde{O}_{\eps}(n^{r-1}).$ 

For non-Boolean $D$, an analogous result is known for non-redundancy due to Carbonnel~\cite{carbonnel2022Redundancy}). The main proof technique is applying a hypergraph Tur{\'a}n result due to Erd{\H o}s~\cite{erdos1964extremal}.

\begin{theorem}[Carbonnel~\cite{carbonnel2022Redundancy}]\label{thm:carbonnel}
For $r \in \N$, let $\delta_r := 2^{1-r}$. For all $R \subseteq D^r$, we have that $\NRD(R, n) = \Omega(n^r)$ if and only if there exists subsets $D_1, \hdots, D_r \subseteq D$ of size exactly $2$ such that $|R \cap (D_1 \times \cdots \times D_r)| = 2^r-1$. Otherwise, $\NRD(R, n) = O(n^{r-\delta_r})$.
\end{theorem}

As an immediate corollary of our main result Theorem~\ref{thm:main-code} tying sparsification to non-redundancy, we can now extend this dichotomy to sparsification.

\begin{theorem}
For $r \in \N$, let $\delta_r := 2^{1-r}$. For all $R \subseteq D^r$ and $\eps \in (0,1)$, we have that $\SPR(R, n, \eps) = \Omega(n^r)$ if and only if there exists subsets $D_1, \hdots, D_r \subseteq D$ of size exactly $2$ such that $|R \cap (D_1 \times \cdots \times D_r)| = 1$. Otherwise, $\SPR(R, n, \eps) = O(n^{r-\delta_r}\log^6 n / \eps^2)$.
\end{theorem}

\begin{proof}
Apply Theorem~\ref{thm:main} to Theorem~\ref{thm:carbonnel}.
\end{proof}

\subsection{Generalized Group Equations: Mal'tsev Sparsification}\label{subsec:maltsev}

For over a half-century (e.g., \cite{bergman2011Universal}), the field of universal algebra has sought to generalize various algebraic structures, such as Abelian and non-Abelian groups. One particular notion that has been vital to the study of CSPs is that of \emph{Mal'tsev\footnote{Spelling variants include Maltsev, Malcev, and Mal'cev.} operators}~\cite{MR65533}. Succinctly, for a domain $D$, a Mal'tsev operator is a ternary function $\varphi : D^3 \to D$ such that for any $x, y \in D$, we have that $\varphi(x,y,y) = \varphi(y,y,x) = x$. As a concrete example, for any group $G$, consider $\varphi_G(x,y,z) := x \cdot y^{-1} \cdot z$. It is clear that $\varphi_G(x,y,y) = \varphi_G(y,y,x) = x$, so $\varphi_G$ is a Mal'tsev operator.

Every finite domain $D$ and Mal'tsev operator $\varphi$ has a corresponding family of CSP predicates (e.g., \cite{bergman2011Universal,barto2017Polymorphisms}). In particular, we say that $R \subseteq D^r$ is a \emph{Mal'tsev predicate} if for any three tuples $a,b,c \in R$, we have that\footnote{In CSP parlance, we say that $\varphi$ is a \emph{Mal'tsev polymorphism} of $R$.}
\[
  (\varphi(a_1,b_1,c_1), \hdots, \varphi(a_r,b_r,c_r)) \in R.
\]
In particular, for an Abelian group $G$, the Mal'tsev predicates corresponding to $(G, \varphi_G)$ are precisely the affine predicates, i.e., cosets of $G^r$. For non-Abelian groups and more general Mal'tsev operators, the corresponding predicates are much more subtle (see e.g., \cite{lagerkvist2020Sparsification}). However, enough is known about Mal'tsev predicates to prove that their non-redundancy is quite small. In fact, \cite{lagerkvist2020Sparsification,bessiere2020Chain} prove that the non-redundancy of every Mal'tsev predicate is linear by utilizing a characterization of Mal'tsev predicates by~\cite{bulatov2006Simple}.

\begin{theorem}[\cite{lagerkvist2020Sparsification,bessiere2020Chain}, simplified]\label{thm:maltsev-nrd}
For every Mal'tsev predicate $R \subseteq D^r$, we have that $\NRD(R, n) = O_D(n)$.
\end{theorem}
It turns out that Theorem~\ref{thm:maltsev-nrd} can be strengthened by considering gadget reductions which preserve non-redundancy up to a constant factor. See Section~\ref{subsec:gadget} and Section~\ref{sec:conclusion} for more details. As an immediate corollary of Theorem~\ref{thm:main}, we get a nearly identical result for sparsification.
\begin{theorem}[``Mal'tsev sparsification'']\label{thm:maltsev-spr}
For every Mal'tsev predicate $R \subseteq D^r$, we have that $\SPR(\overline{R}, n, \eps) = O_D(n\log^6 n/\eps^2)$.
\end{theorem}

\subsection{Gadget Reductions}\label{subsec:gadget}

An excellent property of non-redundancy is that there is a rich gadget reduction framework for describing the relationships between the non-redundancy of predicates.\footnote{Formally these reductions are called \emph{functionally guarded primitive positive (fgpp) definitions}~\cite{lagerkvist2017Kernelization,carbonnel2022Redundancy}.} In particular, the framework allows for three types of operations.

\begin{theorem}[NRD gadget reductions, adapted from Proposition~2~\cite{carbonnel2022Redundancy}]\label{thm:nrd-gadget}
Let $P, Q \subseteq D^r$ be predicates. The following gadget reductions are (approximately) monotonic with respect to non-redundancy.
\begin{itemize}
\item[1.]\textbf{Projection (or minor).} For any $s \in \N$ and $\pi : [s] \to [r]$, let $R \subseteq D^s$ be defined by $R := \{(a_{\pi(1)}, \hdots, a_{\pi(s)}) \in D^s : a \in P\}$. Then, $\NRD(R, n) = O(\NRD(P, n))$.
\item[2.]\textbf{Domain restriction.} For any maps $f_1, \hdots, f_r : E \to D$, consider the predicate $R := \{a \in E^r: (f_1(a_1), \hdots, f_r(a_r)) \in P\}$. Then, $\NRD(R, n) = O(\NRD(P, n))$.
\item[3.]\textbf{Conjunction.} Let $R := P \wedge Q$. Then, $\NRD(R, n) = O(\NRD(P,n) + \NRD(Q,n))$.
\end{itemize}
\end{theorem}

As an immediate consequence of Theorem~\ref{thm:main}, we obtain a comparable gadget framework for sparsification, although somewhat less transparent due to the appearance of complemented relations.
\begin{theorem}
Let $P, Q \subseteq D^r$ be predicates and $\eps > 0$. The following gadget reductions are (approximately) monotonic with respect to sparsification.
\begin{itemize}
\item[1.]\textbf{Projection (or minor).} For any $s \in \N$ and $\pi : [s] \to [r]$, let $R \subseteq D^s$ be defined by $R := \{(a_{\pi(1)}, \hdots, a_{\pi(s)}) \in D^s : a \in P\}$. Then, $\SPR(\overline{R}, n, \eps) = \widetilde{O}_{\eps}(\SPR(\overline{P}, n, \eps))$.
\item[2.]\textbf{Domain restriction.} For any maps $f_1, \hdots, f_r : E \to D$, consider the predicate $R := \{a \in E^r: (f_1(a_1), \hdots, f_r(a_r)) \in P\}$. Then, $\SPR(\overline{R}, n) = \widetilde{O}_{\eps}(\SPR(\overline{P}, n, \eps))$.
\item[3.]\textbf{Disjunction.} Let $R := P \vee Q$. Then, $\SPR(R, n, \eps) = \widetilde{O}_{\eps}(\SPR(P,n,\eps) + \SPR(Q,n,\eps))$.
\end{itemize}
\end{theorem}

Of note, the gadget operations are sufficiently rich that the relative non-redundancy/sparsification of a predicate can be described in terms of a set of (universal) algebraic operations it is closed under called \emph{pattern partial polymorphisms}. See \cite{jonsson2017Time,lagerkvist2018Which,lagerkvist2020Sparsification,carbonnel2022Redundancy} for more details.

\subsection{Applications of Chain Length to Weighted Code Sparsifiers}

We now discuss how the previously-mentioned results for unweighted sparsification can be extended to weighted sparsification.

\subsubsection{Weighted Code Sparsification for Non-Abelian Groups}

Recall for any finite group $G$ and any subgroup $H \le G^m$ that an $\eps$-sparsifier of $H$ is a (sparse) reweighting of $[m]$ which preserves the Hamming weight over every $h \in H$ up to a factor of $1 \pm \eps$. Likewise, for any weight function $w : [m] \to \R_{\ge 0}$, we can define an $(w,\eps)$-sparsifier analogous to (\ref{eq:w-eps-spars}). Using Theorem~\ref{thm:wspr-code}, we can prove an analogue of Theorem~\ref{thm:group-sparsifier} for weighted sparsification.

\begin{theorem}
For any finite group $G$, any $H \le G^m$, any weight function $w : [m] \to \R_{\ge 0}$, and any $\eps \in (0,1)$, there is an $(w,\eps)$-sparsifier of $H$ of size $O(\log |H| (\log^6 m)/\eps^2)$.
\end{theorem}

\begin{proof}
If $|G| = 1$, there is nothing to prove, so assume $|G| \ge 2$. Let $C \subseteq \{0,1\}^m$ be the code such that for every $h \in H$, the binary string $(\one[h_i \neq e] : i \in [m])$ is an element of $C$. Observe that any $(w,\eps)$-sparsifier of $C$ is a $(w,\eps)$-sparsifier of $H$ (and vice-versa). Therefore, by Theorem~\ref{thm:wspr-code}, it suffices to prove that $\CL(C) \le \log_2 |H|$.

Consider a sequence $i_1, \hdots, i_\ell \subseteq [m]$ along with $h_1, \hdots, h_\ell \in H$ such that for all $1 \le j \le j' \le \ell$, we have that $h_{j,i_{j'}} \neq e$ if and only if $j = j'$. Now consider the set
\[
  S := \left\{h_1^{b_1} \cdots h_\ell^{b_\ell} \in H: b \in \{0,1\}^\ell\right\}.
\]
We claim that all the elements defining $S$ are distinct. To see why, assume there exist distinct $b, b' \in \{0,1\}^{\ell}$ with $h_1^{b_1} \cdots h_\ell^{b_\ell} = h_1^{b'_1} \cdots h_\ell^{b'_\ell}.$
Let $j \in [\ell]$ be the first coordinate for which $b_j \neq b'_j$. Then, we have that $h_j^{b_j} \cdots h_\ell^{b_\ell} = h_j^{b'_j} \cdots h_\ell^{b'_\ell}.$ Note that the $i_j$th coordinate of the LHS is equal to $h_{j,i_j}^{b_j}$. Likewise, the $i_j$th coordiante of the RHS is equal to $h_{j,i_j}^{b'_j}$. Since $h_{j,i_j} \neq e$ and $b_j \neq b'_j$, we have a contradiction. Thus, the elements defining $S$ are distinct, so $|S| = 2^\ell$. Thus, $\ell \le \log_2 |H|$, so $\CL(C) \le \log_2|H|$, as desired.
\end{proof}

\subsubsection{Maximum-size Dichotomy}\label{subsec:msd-cl}

Although Theorem~\ref{thm:carbonnel} due to Carbonnel~\cite{carbonnel2022Redundancy} is only stated in terms of non-redundancy, it turns out that it also applies to chain length. We give a proof for completeness.

\begin{theorem}\label{thm:carbonnel-cl}
For $r \in \N$, let $\delta_r := 2^{1-r}$. For all $R \subseteq D^r$, we have that $\CL(R, n) = \Omega(n^r)$ if and only if there exists subsets $D_1, \hdots, D_r \subseteq D$ of size exactly $2$ such that $|R \cap (D_1 \times \cdots \times D_r)| = 2^r-1$. Otherwise, $\CL(R, n) = O(n^{r-\delta_r})$. 
\end{theorem}
\begin{proof}
If there exists subsets $D_1, \hdots, D_r \subseteq D$ of size exactly $2$ such that $|R \cap (D_1 \times \cdots \times D_r)| = 2^r-1$, then by Proposition~\ref{prop:cl-ge-nrd} and Theorem~\ref{thm:carbonnel}, we have that $\CL(R, n) \ge \NRD(R, n) = \Omega(n^r)$, as desired.

Now assume such subsets do not exist. Let $X$ be a set of variables of size $n$, and let $\ell := \CL(R, n) - 1$. By Corollary~\ref{cor:chain-csp}, there exist assignments $\sigma_1, \hdots, \sigma_\ell : X \to D$ and clauses $y_1, \hdots, y_\ell \in X^r$ such that for all $i \in [\ell]$, $\sigma_i(y_1), \hdots, \sigma_i(y_{i-1}) \in R$ and $\sigma_i(y_i) \not\in R.$

Let $Y := \{y_1, \hdots, y_\ell\}$. We claim there do not exist $X_1, \hdots, X_r \subset X$ of size $2$ such that $X_1 \times \cdots \times X_r \subseteq Y.$ Otherwise, there exists $y_0 \in X_1 \times \cdots \times X_r$ such that there is an assignment $\sigma : X \to D$ for which $\sigma(y_0) \not\in R$ but $\sigma(y) \in R$ for all $y \in X_1 \times \cdots \times X_r$ not equal to $y_0$. For that to be the case, the sets $D_1 := \sigma(X_1), \hdots, D_r := \sigma(X_r)$ must each have size exactly $2$ (since something must differ from $\sigma(y_{0,i})$ for each $i \in [m]$), and
\[
D_1 \times \cdots \times D_r = \{\sigma(y) : y \in X_1 \times \cdots \times X_r\}.
\]
Therefore, $(D_1 \times \cdots D_r) \cap R$ has size exactly $2^r-1$ (since $\sigma(y_0)$ is excluded), a contradiction. Hence, the $r$-uniform hypergraph on vertex set $X$ and edge set $Y$ lacks any complete $r$-partite subgraph with each part having size exactly $2$. Thus, by a theorem of Erd{\H o}s~\cite{erdos1964extremal}, we have that $\ell= O(n^{r-\delta_r})$, so $\CL(R, n) = O(n^{r-\delta_r})$.
\end{proof}
By Theorem~\ref{thm:wspr-code}, this immediately extends to weighted sparsification.
\begin{theorem}
For $r \in \N$, let $\delta_r := 2^{1-r}$. For all $R \subseteq D^r$ and $\eps \in (0,1)$, we have that $\wSPR(R, n, \eps) = \Omega(n^r)$ if and only if there exists subsets $D_1, \hdots, D_r \subseteq D$ of size exactly $2$ such that $|R \cap (D_1 \times \cdots \times D_r)| = 1$. Otherwise, $\wSPR(R, n, \eps) = O(n^{r-\delta_r}(r\log n)^6 / \eps^2)$.
\end{theorem}

\subsubsection{Generalized Group Equations: Mal'tsev Sparsification}

Both Lagerkvist and Wahlstr{\"o}m~\cite{lagerkvist2020Sparsification} and Bessiere,  Carbonnel, and Katsirelos~\cite{bessiere2020Chain} proved a linear upper bound on the chain length of any Mal'tsev predicate.

\begin{theorem}[\cite{lagerkvist2020Sparsification,bessiere2020Chain}]\label{thm:maltsev-cl}
For every Mal'tsev predicate $R \subseteq D^r$, we have that $\CL(R, n) = O_D(n)$.
\end{theorem}

Thus, as an immediate corollary of Theorem~\ref{thm:wspr-code}, we have that

\begin{theorem}
For every Mal'tsev predicate $R \subseteq D^r$, we have that $\wSPR(\overline{R}, n, \eps) = O_D(n\log^6 n/\eps^2)$.
\end{theorem}

\subsubsection{Gadget Reductions}

Unlike for non-redundancy, a formal gadget framework is not explicitly stated in the literature for chain length. Bessiere,  Carbonnel, and Katsirelos~\cite{bessiere2020Chain} did observe (Observation 2) that a restricted family of gadget reductions called ``c-definitions'' is monotone with respect to chain length. In particular, using the notation of Theorem~\ref{thm:nrd-gadget}, they allow arbitrary conjunctions as well as projection reductions from maps $\pi : [s] \to [r]$ which are injective. 

However, in practice (i.e., in \cite{lagerkvist2020Sparsification}), chain length is closely used in association with pattern partial polymorphisms, suggesting that the full suite of gadget reductions in Theorem~\ref{thm:nrd-gadget} also apply to chain length. We leave filling this gap in the literature as future work.

\section{New results on Non-redundancy via Matching Vector Families}\label{sec:nrd-mv}

In this section, we consider the following families of arity-3 predicates. Let $G$ be a finite Abelian group of order at least $3$ with operator $+$ and identity element $0$.
\begin{align*}
\tLIN_{G} &:= \{(x,y,z) \in G^3 : x+y+z = 0\}\\
\tLIN^{*}_{G} &:= \tLIN_G \setminus \{(0,0,0)\}.
\end{align*}
\begin{remark}
The inspiration for studying $\tLIN^{*}_G$ came from investigating a related predicate $\BCK := \{111,222,012,120,201\}$ from \cite{bessiere2020Chain}. See Section~\ref{sec:open-nrd} for more details.
\end{remark}

Since $\tLIN_{G}$ is defined by a linear equation over an Abelian group, it is already known (e.g., \cite{lagerkvist2020Sparsification,bessiere2020Chain}) that $\NRD(\tLIN_G, n) = \Theta_G(n)$ and $\SPR(\overline{\tLIN_G}, n) = \widetilde{\Theta}_G(n/\eps^2)$ \cite{khanna2024Characterizations}.

However, much less is known about $\tLIN^{*}_G$. In particular, existing methods in the literature can only deduce $\NRD(\tLIN^{*}_G, n) \in [\Omega_G(n), O_G(n^2)]$. In this section, we introduce a new method for bounding the non-redundancy of predicates like $\tLIN^{*}_G$ by connecting them to the theory of \emph{matching vector families} \cite{yekhanin20083query,raghavendra2007note,gopalan2009note,efremenko20093query,dvir2011Matching} used primarily in the construction of locally decodable codes (LDCs). In particular, this is the first example in the non-redundancy/sparsification literature of a predicate whose optimal exponent as a function of $n$ is non-integral.

\subsection{Conditional Non-redundancy}

Recall (Definition~\ref{def:NRD-CSP}) that for a predicate $R \subseteq D^r$, an instance $\Psi := (X, Y \subseteq X^r)$ is non-redundant if for all $y \in Y$, there exists an assignment $\sigma_y : X \to D$ such that $\sigma_y(y) \not\in R$ but $\sigma_y(y') \in R$ for all $y' \in Y \setminus \{y\}$. We now modify the definition of non-redundancy slightly to capture conditional information about these assignments $\sigma_y$. 

\begin{definition}[Conditional Non-redundancy]\label{def:cond-nrd}
Given $P \subseteq Q \subseteq D^r$, we say that an instance $\Psi := (X,Y \subseteq X^r)$ is a non-redundant instance of $\CSP(P \mid Q)$ if for all $y \in Y$, there exists an assignment $\sigma_y : X \to D$ such that $\sigma_y(y) \in Q \setminus P$ and $\sigma_y(y') \in P$ for all $y' \in Y \setminus \{y\}$. We let $\NRD(P\mid Q, n)$ be the maximum number of clauses in a non-redundant instance of $\CSP(P \mid Q)$ with $n$ variables.
\end{definition}

Observe that each $\sigma_y$ is a satisfying assignment to $\CSP(Q)$. Also note that if $Q = D^r$, then $\NRD(P\mid Q, n) = \NRD(P, n)$. The key motivation for conditional non-redundancy is that it satisfies a form of the triangle inequality.

\begin{proposition}\label{prop:nrd-tri}
Let $P \subseteq Q \subseteq R \subseteq D^r$ be predicates, then we have that
\begin{align}
\NRD(P\mid Q, n) \le \NRD(P\mid R, n) \le \NRD(P\mid Q, n) + \NRD(Q\mid R, n).\label{eq:nrd-tri}
\end{align}
In particular, if $R = D^r$, we have that
\begin{align}
\NRD(P\mid Q,n) \le \NRD(P, n) \le \NRD(P\mid Q, n) + \NRD(Q, n).\label{eq:nrd-tri-simple}
\end{align}
\end{proposition}

In particular, if we seek to understand $\NRD(P,n)$, it suffices to identify a predicate $Q \supsetneq P$ with $\NRD(Q,n) \ll \NRD(P,n)$ and to bound $\NRD(P|Q,n)$. Then, we can use the structure of satisfying assignments to $\CSP(Q)$ to better constrain the non-redundancy of $\CSP(P)$. As we shall see, this is particularly useful when $Q$ is an affine predicate.

\begin{proof}
Note that $Q \setminus P \subseteq R \setminus P$. Thus, any non-redundant instance of $\CSP(P \mid Q)$ is a non-redundant instance of $\CSP(P \mid R)$. This proves the left inequality of (\ref{eq:nrd-tri}).

Let $(X,Y \subseteq X^r)$ be a non-redundant instance of $\CSP(P \mid R)$. Let $\{\sigma_y : y \in Y\}$ be a set of assignments satisfying the conditions of Definition~\ref{def:cond-nrd}. Let $Y_R \subseteq Y$ be the set of $y \in Y$ such that there exists an assignment $\sigma'_y : X \to D$ for which $\sigma'_y(y) \in R \setminus Q$ and $\sigma'_y(y') \in P$ for all $y' \in Y \setminus \{y\}$.  Since $P \subseteq Q$, we have that $(X, Y_{R})$ is a non-redundant instance of $\CSP(Q \mid R)$, so $|Y_R| \le \NRD(Q\mid R, n)$.

By definition of $Y_R$, for all $y \in Y \setminus Y_{R}$, we have that $\sigma'_y(y) \in Q \setminus P$. Therefore, $(X, Y \setminus Y_R)$ is a non-redundant instance of $\CSP(P \mid Q)$. Thus, $|Y \setminus Y_R| \le \NRD(P\mid Q, n)$. Since $|Y| \le |Y \setminus Y_R| + |Y_R|$, we have proved (\ref{eq:nrd-tri}).
\end{proof}

As a consequence, we can prove a simple $O_G(n^2)$ upper bound on $\NRD(\tLIN^*_G, n)$. A similar upper bound for a similar predicate was given by \cite{bessiere2020Chain} via a somewhat different technique, see Section~\ref{sec:open-nrd} for more details.

\begin{proposition}
\label{prop:simple-quadratic-ub}
$\NRD(\tLIN^*_G, n) = O_G(n^2)$.
\end{proposition}
\begin{proof}
As previously discussed, $\NRD(\tLIN_G,n) = O_G(n)$. Hence, by (\ref{eq:nrd-tri-simple}), it suffices to prove that $\NRD(\tLIN^*_G|\tLIN_G, n) \le n^2.$ Let $(X,Y)$ be a non-redundant instance of $\CSP(\tLIN^*_G \mid \tLIN_G)$. 

They key observation is that any two non-redundant clauses must share at most one variable in common. In particular, if both $(x_1,x_2,x_3) \in Y$ and $(x_1,x_2,x'_3) \in Y$ with $x_3 \neq x'_3$, then any satisfying assignment $\sigma$ to $(X,Y)$ as in instance of $\CSP(\tLIN_G)$ must impose that $\sigma(x_3) = \sigma(x'_3)$. Therefore, either both clauses satisfy $\tLIN^*_G$ or neither do. This contradicts that $(X,Y)$ is a non-redundant instance. Hence for any $(x_1,x_2,x_3) \in Y$, $(x_1,x_2)$ uniquely determines the clause. Thus, $|Y| \le n^2$, as desired.
\end{proof}

In Section~\ref{subsec:nrd-ub}, we shall substantially improve on this upper bound for many finite Abelian groups $G$.

\subsection{Connection to Matching Vector Families}\label{subsec:nrd-mv}

Recall from the introduction that a matching vector (MV) family is two lists of vectors of the same length such that the dot product of a pair of vectors for the two lists equals 0 if and only if they have the same index in their respective lists. We now define a generalization of MV families where the homomorphisms are applied to triples of vectors at a time instead a single vector.

\begin{definition}\label{def:G-ensemble}
For an Abelian group $G$, a \emph{$G$-ensemble} consists of a dimension $d \in \N$, a collection of vectors $V \subseteq G^d$, a collection of edges $E \subseteq V^3$, and a collection of homomorphisms $\phi_e \in \Hom(G^d, G)$ for all $e \in E$. We require that each $e := (v_1, v_2, v_3) \in E$, we have that $v_1 + v_2 + v_3 = 0$. Finally, for any $e,e' \in E$, we have that $\phi_e$ maps all the vectors of $e'$ to $0$ if and only if $e = e'$. We define the \emph{size} of a $G$-ensemble to be $|E|$, the number of edges.
\end{definition}

We now connect the existence of these $G$-ensembles to the non-redundancy problem.

\begin{proposition}\label{prop:ensemble-nrd}
The maximum size of a $G$-ensemble with $n$ vectors (in any dimension) equals $\NRD(\tLIN^*_G | \tLIN_G, n)$.
\end{proposition}

\begin{proof}
First, we show for any $G$-ensemble $(d,V, E, \{\phi_e : e \in E\})$ that $(V, E)$ is a non-redundant instance of $\CSP(\tLIN^*_G \mid \tLIN_G)$. To see why, note that for any $e, e' = (v'_1, v'_2, v'_3) \in E$, we have that $\phi_e(e') = (\phi_e(v'_1), \phi_e(v'_2), \phi_e(v'_3))$. These values must sum to $0$ because $v'_1+v'_2+v'_3=0$ and $\phi_e$ is a homomorphism. Further, observe that $\phi_e(e') = (0,0,0) \in \tLIN_G \setminus \tLIN^*_G$ if and only if $e = e'$. Therefore, $(V,E)$ is indeed non-redundant.

Second, we show that any non-redundant instance $(X,Y)$ of $\CSP(\tLIN^*_G \mid \tLIN_G)$ can be turned into a $G$-ensemble $(d, V, E, \{\phi_e : e \in E\})$ with $|V|= |X|$ and $|E|=|Y|$. Let $\{\sigma_y : X \to G \mid y \in Y\}$ be the assignments guaranteed by Definition~\ref{def:cond-nrd}. Let $H := \sat(\tLIN_G, Y),$ the set of satisfying assignments to $(X,Y)$ when viewed as an instance of $\CSP(\tLIN_G)$.

We now highlight the key step in the proof. Observe that $H$ can be viewed as a subgroup of $G^X$.  Let $h_1, \hdots, h_d \in H$ be any generating set of $H$. For each $x \in X$, we can construct $v_x \in G^d$ by letting $v_{x,i} = h_{i,x}$ for all $i \in [d]$ (i.e., the transpose). For each $y := (x_1,x_2,x_3) \in Y$, we let $e_y := (v_{x_1}, v_{x_2}, v_{x_3})$ with $E := \{e_y : y \in Y\}$. Since for every $h \in H$, we have $h_{x_1} + h_{x_2} + h_{x_3} = 0$, and in particular this holds for each $h_i$, we can conclude that $v_{x_1} + v_{x_2} + v_{x_3} = 0$.

Observe that for any $y \in Y$ each assignment $\sigma_y$ can be viewed as an element of $H$. Thus, there exists $z_{y,1}, \hdots, z_{y,d} \in \mathbb Z$ (not necessarily unique) such that $\sigma_y = z_{y,1}h_1 + \cdots + z_{y,d}h_d$. We can thus define the map $\phi_{e_y} \in \Hom(G^d, G)$ by
\[
  \phi_{e_y}(g_1, \hdots, g_d) := \sum_{i=1}^d z_{y,i}g_i.
\]
Thus, by definition of $\sigma_y$, each corresponding $\phi_{e_y}$ must map $e_{y'}$ to $(0,0,0)$ if and only if $y = y'$, as desired. Note this further implies that each $e_y$ is distinct, so $|E| = |Y|$, as desired.

As a last detail, note that $|V| \le |X|$ by definition of $X$, but $|V| < |X|$ if two $v_{x}$'s are identical. We can get around this issue by padding.  Increase $d$ to $d+|X|-|V|$ and pad each $v \in V$ by appending $|X|-|V|$ copies of $0$. We extend each homomorphism $\phi_e$ by ignoring these last $|X|-|V|$ coordinates. Pad $V$ with any $|X|-|V|$ additional vectors. It is straightforward to check we still have a $G$-ensemble of the same size, but now with $|X|$ vectors.
\end{proof}

We now heavily use Proposition~\ref{prop:ensemble-nrd} to construct novel upper and lower bounds on $\NRD(\tLIN^*_G, n)$.

\subsection{Constructing a Non-redundant Instance}\label{subsec:nrd-lb}

We now show a lower bound of size $\Omega(n^{1.5})$ for $\NRD(\tLIN^*_G, n)$ for all Abelian groups $G$ of order at least $3$.  The cardinality requirement is necessary as $\tLIN^*_{\Z/1\Z}$ is the empty relation and $\tLIN^*_{\Z/2\Z}$ is equivalent to $2$-in-$3$ SAT, which has linear non-redundancy (e.g., \cite{lagerkvist2020Sparsification,bessiere2020Chain,khanna2024Characterizations}). The key property of $G$ we need is quite simple: 

\begin{lemma}\label{lem:g-3}
If $G$ is an Abelian group with $|G| \ge 3$, then there exists $g_1, g_2, g_3 \in G \setminus \{0\}$ such that $g_1 + g_2 + g_3 = 0$.
\end{lemma}
\begin{proof}
If there exists some element $g \in G \setminus \{0\}$ of order $t \ge 3$, then we can take $g_1 = g_2 = g$ and $g_3 = (t-2)g$. Otherwise, every element of $g \in G \setminus \{0\}$ has order $2$. Since $|G| \ge 3$, there exists distinct $g_1, g_2 \in G \setminus \{0\}$. Set $g_3 = g_1 + g_2 \neq 0$.
\end{proof}

We can now prove our lower bound on non-redundancy. The construction presented is inspired by a method in the matching vector code literature (cf. Lemma~16 of \cite{dvir2011Matching}).

\begin{theorem}\label{thm:nrd-3-lin-lb}
For every Abelian group $G$ with $|G| \ge 3$, we have that \[\NRD(\tLIN^*_G, n) \ge \NRD(\tLIN^*_G|\tLIN_G, n)= \Omega(n^{1.5}).\]
\end{theorem}

\begin{proof}
By Proposition~\ref{prop:nrd-tri}, we have that $\NRD(\tLIN^*_G, n) \ge \NRD(\tLIN^*_G | \tLIN_G, n).$ By Proposition~\ref{prop:ensemble-nrd}, to lower bound $\NRD(\tLIN^*_G | \tLIN_G, n)$, it suffices to construct a $G$-ensemble on $n$-vectors of size $\Omega(n^{1.5})$. We describe the construction as follows.

Let $t = \lfloor \sqrt{n/3}\rfloor$ and let $d = 3t+2$. Let $g_1,g_2,g_3 \in G \setminus \{0\}$ be chosen as in Lemma~\ref{lem:g-3}. Let $I_1 := \{1, \hdots, t\}, I_2 := \{t+1, \hdots, 2t\},$ and $I_3 := \{2t+1, \hdots, 3t\}$. For all $i \in [d]$ and $g \in G$, let $b_i(g) \in G^d$ be the vector which has all coordinates equal to zero except for the $i$th coordinate which has value $g$. Now consider the following sets of vectors
\begin{align*}
V_{12} &:= \{u_{i_1,i_2} := b_{i_1}(g_1) + b_{i_2}(g_2) + b_{3t+1}(g_3) : i_1 \in I_1, i_2 \in I_2\},\\
V_{23} &:= \{v_{i_2,i_3} := -b_{i_2}(g_2) - b_{i_3}(g_3) - b_{3t+2}(g_1) : i_2 \in I_2, i_3 \in I_3\},\\
V_{13} &:= \{w_{i_1,i_3} := -b_{i_1}(g_1) + b_{i_3}(g_3) - b_{3t+1}(g_3) + b_{3t+2}(g_1) : i_1 \in I_1, i_3 \in I_3\}.
\end{align*}
Let $V = V_{12} \cup V_{23} \cup V_{13}$ which his of size $3t^2 \le n$, which can be made of size exactly $n$ via padding. For each triple $(i_1, i_2, i_3) \in I_1 \times I_2 \times I_3$ observe that $u_{i_1,i_2} + v_{i_2,i_3} + w_{i_1,i_3} = 0$. Therefore, we can define $E := \{e_{i_1,i_2,i_3} := (u_{i_1,i_2},v_{i_2,i_3}, w_{i_1,i_3}) \mid (i_1,i_2,i_3) \in I_1 \times I_2 \times I_3\} \subset V^3.$ The size of $|E|$ is $t^3 = \Omega(n^{3/2}).$ 

To finish, it suffices to construct maps $\phi_{i_1,i_2,i_3} \in \Hom(G^d, G)$ which map only $e_{i_1,i_2,i_3}$ to $(0,0,0)$. To do this, we define
\[
  \phi_{i_1,i_2,i_3}(z) := z_{i_1} + z_{i_2} + z_{i_3} + z_{3t+1} + z_{3t+2}.
\]
Now observe that for any $(i'_1, i'_2, i'_3) \in I_1 \times I_2 \times I_3$, we have that
\begin{align}
\phi_{i_1,i_2,i_3}(u_{i_1,i_2}) &= g_1 \one[i'_1 = i_1] + g_2 \one[i'_2 = i_2] + g_3,\label{eq:v12}\\
\phi_{i_1,i_2,i_3}(v_{i_2,i_3}) &= -g_2 \one[i'_2 = i_2] - g_3 \one[i'_3 = i_3] - g_1.\label{eq:v23}
\end{align}
Since each $g_i$ is nonzero and their sum equals zero, the only way (\ref{eq:v12}) equals $0$ is for $i'_1 = i_1$ and $i'_2 = i_2$. Likewise, the only way (\ref{eq:v23}) equals $0$ is for $i'_2 = i_2$ and $i'_3 = i_3$. Thus, the only way that (\ref{eq:v12}) and (\ref{eq:v23}) are both zero is for $(i'_1, i'_2, i'_3) = (i_1, i_2, i_3)$. Since $w_{i_1,i_3} = -u_{i_1,i_2}-v_{i_2,i_3}$ and $\phi_{i_1,i_2,i_3}$ is a homomorphism, we have that $\phi_{i_1,i_2,i_3}(e_{i'_1,i'_2,i'_3}) = (0,0,0)$ if and only if $(i_1, i_2,i_3) = (i'_1,i'_2,i'_3)$. Therefore, we have indeed constructed a $G$-ensemble of size $\Omega(n^{1.5})$.
\end{proof}

\begin{remark}\label{rem:extension}
Note that the only assignments any $\phi_{i_1,i_2,i_3}$ assigns to a clause are from the list
\begin{align*}
R := &\{(0,0,0),(0,g_3,-g_3),(-g_1,0,g_1),(-g_1,g_3,g_1-g_3),\\&\ \ (-g_2,-g_1,-g_3),(-g_2,g_2,0),(g_3,-g_1,g_1-g_3),(g_3,g_2,g_1)\}
\end{align*}
In particular, this implies our non-redundancy lower bound also applies to $R \setminus \{(0,0,0)\}$.  We explore this in more detail in Section~\ref{sec:open-nrd}. This observation also suggests that for large groups $G$, a construction which uses ``more'' of $\tLIN^*_G$ could yield a stronger lower bound.
\end{remark}

\subsection{Upper Bounding Non-redundancy}\label{subsec:nrd-ub}

In this section, we show how the theory of matching vector codes can also lead to upper bounds on $\NRD(\tLIN^*_G, n)$. Due to the existence of very large MV families for Abelian groups $G$ divisible by at least two primes~\cite{dvir2011Matching}, we assume that $G = \mathbb Z/p\mathbb Z$, where $p$ is a prime.

In particular, we can identify vectors in $G^d$ with vectors in the vector space $\F_p^d$. As a consequence, for our collection of vectors $V \subseteq G^d$, we can quantify the ``spread'' of $V$ by $\dim(\Span(V))$.\footnote{In \cite{khanna2024Characterizations}, the authors use $\log_2 \card{\Span(V)}$ as a surrogate for dimension for arbitrary finite Abelian groups. However, to extend Theorem~\ref{thm:ub-1.6} to all finite Abelian groups we would need a suitable analogue of Lemma~\ref{lem:dimU-upper-bound} for arbitrary groups, which in many cases is false (see~\cite{dvir2011Matching}).} 

As a first step, we give a bound on the size of any $\F_p$-ensemble as a function of $\dim(\Span(V))$. 

\begin{lemma}\label{lem:dimU-upper-bound}
Let $p$ be a prime. There exists a constant $c_p > 0$ with the following property. Let $(d, V, E, \{\phi_e : e \in E\})$ be an $\F_p$-ensemble. For any subspace $U \le \F^d$, we have that
\begin{align*}
    |E \cap U^3| \le \binom{\dim(U) + 2p-2}{2p-2}.
\end{align*}
In particular, for $U = \Span(V)$,
\begin{align}
  |E| \le \binom{\dim(\Span(V)) + 2p-2}{2p-2}\label{eq:E-ub}.
\end{align}
\end{lemma}

\begin{proof}
We mimic the proof of Theorem~23 of \cite{dvir2011Matching} (themselves inspired by \cite{babai1988LINEAR}).

Let $m = |E \cap U^3|$. Let $e_i = (u_i, v_i, w_i) \in E \cap U^3$ be an enumeration of these edges. For each $i \in [m]$, define the degree $2p-2$ polynomial $P_i : U \to \F$ as follows:
\[
    P_i(x) := (1 - \langle x, u_i\rangle^{p-1})(1 - \langle x, v_i\rangle^{p-1}).
\]
For each $e \in E \cap U^3$, there exists a unique $z_e \in U$ such that $\phi_{e}(x) = \langle z_e, x\rangle$ for all $x \in U$. In particular, for each $i,j \in [m]$, we have that
\begin{align*}
    P_i(z_{e_j}) &= (1 - \langle z_{e_j}, u_i\rangle^{p-1})(1 - \langle z_{e_j}, v_i\rangle^{p-1})\\
    &= (1 - \phi_{e_j}(u_i)^{p-1})(1 - \phi_{e_j}(v_i)^{p-1})\\
    &= \one[(\phi_{e_j}(u_i),\phi_{e_j}(v_i)) = (0,0)]\\
    &= \one[i = j].
\end{align*}
Thus, $P_1, \hdots, P_m$ are linearly independent in the vector space of degree $2p-2$ polynomials on $\dim(U)$ variables. Therefore, by a standard counting argument, $m \le \binom{\dim(U) + 2p-2}{2p-2}$, as desired.
\end{proof}

The key takeaway of Lemma~\ref{lem:dimU-upper-bound}, is that if $\dim(\Span(V)) \le n^{1/(p-1)-\eps}$ for some $\eps > 0$, then we have proved a subquadratic upper bound on $|E|$. To handle the case in which $V$ is more ``spread out,'' we diverge from the standard techniques in the matching vector codes literature and prove a more combinatorial upper bound. In particular, we utilize the fact that in the spread-out case, $V$ has large subsets of linearly independent vectors.

\begin{lemma}\label{lem:ind-ub}
Let $\F$ be a field and let $V \subseteq \F^d$ be a collection of vectors. Let $I \subseteq V$ be a linearly independent subset of vectors. Then, there are at most $6|V| \log_2 |V|$ triples of vectors $v_1,v_2,v_3 \in V$ summing to $0$ with at least one vector in $I$.
\end{lemma}

\begin{remark}
This bound is tight up to constant factors for the Hadamard code. In particular, let $V = \F^s$ and let $I$ be the standard basis. There are $\Theta(|I||V|) = \Theta(|V| \log |V|)$ triples of vectors summing to $0$ with at least one vector in $I$. However, due to the limitations imposed by Lemma~\ref{lem:dimU-upper-bound}, this construction does not correspond to an $\F$-ensemble.
\end{remark}

\begin{proof}
Let $e_1, \hdots, e_t$ be an enumeration of the elements of $I$. For $i \in t$, let $E_i \subseteq V^2$ be the pairs of vectors which sum to $-e_i$. Observe that each $E_i$ is a partial matching and $E_i \cap E_j = \emptyset$ for all $i \neq j$. Let $E_{\le i} = E_1 \cup \cdots \cup E_i$. Note that $3|E_{\le t}| = 3 \sum_{i=1}^t |E_i|$ is an upper bound on the number of edges from $E$ incident with exactly one edge of $I$. Thus, it suffices to prove that $|E_{\le t}| \le 2|V|\log_2 |V|$. The key observation is as follows.

\textbf{Observation.} View $E_{\le i}$ as the edges of a bipartite graph on $V \sqcup V$. For every $j > i$, and for every connected component $C$ of $E_{\le i}$, no edge of $E_j$ has both ends in $C$.

To see why this is true, consider any potential edge $(v_1, v_2) \not\in E_{\le i}$ lying in $C$. Since $E_{\le t}$ is bipartite, there exists a path $v_1 = u_1, u_2, \hdots, u_{2k} = v_2$ with $(u_a,u_{a+1})\in E_{\le i}$ for all $a \in [2k-1]$, where we use the fact that $(p,q) \in E_{\le i}$ if and only if $(q,p) \in E_{\le i}$. Now observe that
\[
  v_2 + v_1 = (u_1+u_2) - (u_2+u_3) + (u_3+u_4) - \cdots + (u_{2k-1} + u_{2k}) \in \Span\{e_1, \hdots, e_i\},
\]
where the inclusion is by the definition of $E_{\le i}$. Since $e_1, \hdots, e_t$ are linearly independent, we must have that $(v_1, v_2) \not\in E_j$ for all $j > i$, as desired.

\textbf{Structural Induction.} We now prove by induction the following claim: for every $i \in [t]$, and every connected component $C$ of $E_{\le i}$ (viewed as a bipartite graph on $V \sqcup V$), we have that the subgraph of $E_{\le i}$ induced by $C$, $E_{\le i}\!\upharpoonright_{C}$, has at most $|C|\log_2 |C|$ edges.

For the base case of $i=1$, note that every connected component has either exactly one vertex and zero edges or two vertices and one edge. In both cases, the bound of $|C| \log_2 |C|$ holds.

Now assume the induction hypothesis holds for $E_{\le i}$. Let $C$ be a new (compared to $E_{\le i}$) connected component of $E_{\le i+1}$. Note that $C$ must be the union of components $C'_1, \hdots, C'_a$ of $E_{\le i}$. Further, due to the observation, the edges of $E_{i+1}$ must be between different components. Assume WLOG that $|C'_1| \ge \cdots \ge |C'_a|$. By the observation, at least one vertex of every edge of $E_{i+1}\!\upharpoonright_{C}$ must lie in $C'_2 \cup \cdots \cup C'_a$. Thus, we have a bound
\begin{align*}
  |E_{\le i+1}|_{C}| &\le \sum_{j=2}^a |C'_j| + \sum_{j=1}^a |E_{\le i} \upharpoonright_{C'_j}|\\
                 &\le\sum_{j=2}^a |C'_j| + \sum_{j=1}^a |C'_j|\log_2 |C'_j| &\text{ (induction hypothesis)}\\
                 &= |C'_1|\log_2|C'_1| + \sum_{j=2}^a |C'_j|\log_2 (2|C'_j|)\\
                 &\le |C'_1|\log_2 |C| + \sum_{j=2}^a |C'_j| \log_2 |C| &\text{ (since $|C| \ge |C_1|+|C_j| \ge 2|C_j|$)}\\
                 &= |C| \log_2|C|,
\end{align*}
as desired.

To finish, each component $C$ of $E_{\le t}$ has at most $|C| \log_2 |C| \le |C| \log_2 |V|$ edges. Since the sum of the sizes of the components is $2|V|$, we have that $|E_{\le t}| \le 2|V|\log_2 |V|$, as desired.
\end{proof}

We now conclude with the main result of this section.

\begin{theorem}\label{thm:ub-1.6}
For any prime $p$, we have $\NRD(\tLIN^*_{\Z/p\Z}, n) = O_p(n^{2 - \eps_p}\log n)$, for $\eps_p = \frac{2}{2p-1}$.
\end{theorem}

\begin{proof}
By Proposition~\ref{prop:nrd-tri} and the fact that $\NRD(\tLIN_{\Z/p\Z}, n) = O(n)$, it suffices to prove that $\NRD(\tLIN^*_{\Z/p\Z}|\tLIN_{\Z/p\Z}, n) = O_p(n^{2 - \eps_p}\log^{1-\eps_p}n)$. By Proposition~\ref{prop:ensemble-nrd}, it suffices to prove this upper bound on the size of any $\F_p$-ensemble with $n$ vectors. To do this, we shall show that every $\F_p$-ensemble on $n'$ vectors for $n' \le n$ has a vector involved in at most $\Delta := O_p((n\log n)^{1-\eps_p})$ edges; i.e., the hypergraph is $\Delta$-degenerate. In particular, this implies that any $\F_p$-ensemble on $n$ vectors has at most $n\Delta = O_p(n(n\log n)^{1-\eps_p})$ edges, as desired.

To do this, we perform casework on the size of $\dim(\Span(V))$.

\textbf{Case 1, $\dim(\Span(V)) \le n^{\eps_p}$}. By Lemma~\ref{lem:dimU-upper-bound}, we have our graph has at most 
\[
\binom{\dim(\Span(V))+2p-2}{2p-2} \le \binom{n^{\eps_p}+2p-2}{2p-2} = O_p(n^{2-\eps_p})
\]
edges. Thus, some vertex has degree at most $O_p(n^{1-\eps_p}) \le \Delta.$

\textbf{Case 2, $\dim(\Span(V)) > n^{\eps_p}$}. Let $I \subseteq V$ be a basis of $\Span(V)$. By Lemma~\ref{lem:ind-ub} there are at most $O(n \log n)$ edges incident with $I$. Thus, some vertex of $I$ has degree at most $O((n/|I|) \log n) = O_p(n^{1-\eps_p}\log n) \le \Delta,$ as desired.
\end{proof}

\begin{remark}[Efficient Kernelization]
We observe that this upper bound proof can be turned into an efficient kernelization algorithm (see Section~\ref{subsec:related-work}). More precisely, given any instance $(X,Y)$ of $\CSP(\tLIN^*_{\Z/p\Z})$ on $n$ variables, we can efficiently identify $Y' \subseteq Y$ of size at most $\widetilde{O}(n^{2-2/(2p-1)})$ such that $Y$ and $Y'$ have the same solution set.
\end{remark}

\begin{remark}
As an immediate corollary of Theorem~\ref{thm:nrd-3-lin-lb} and Theorem~\ref{thm:ub-1.6}, we have identified the first predicate $R$ for which there exist $a \in \N$ for which
\[
  a < \liminf_{n \to \infty} \frac{\log \NRD(R,n)}{\log n} \le \limsup_{n \to \infty} \frac{\log \NRD(R,n)}{\log n} < a + 1.
\]
In particular, we now know that
\begin{align*}
\NRD(\tLIN^*_{\Z/3\Z}, n) &\in [\Omega(n^{1.5}), \widetilde{O}(n^{1.6})], \text{ and }\\
\SPR(\overline{\tLIN^*_{\Z/3\Z}}, n, \eps) &\in [\Omega(n^{1.5}), \widetilde{O}(n^{1.6}/\eps^2)].
\end{align*}
\end{remark}

\section{Predicates with Unresolved Non-Redundancy}\label{sec:open-nrd}

In this section, we give an overview of various predicates discussed in the literature whose non-redundancy is unknown. Essentially all of these predicates are some modification of a system of linear equations, perhaps indicating some modification of the techniques in Section~\ref{sec:nrd-mv} could be useful in their resolution.

\subsection{The Simplest Unresolved Predicate}\label{subsec:simplest}

A central open question is many works on non-redundancy and related questions (e.g.,~\cite{lagerkvist2020Sparsification,bessiere2020Chain,carbonnel2022Redundancy} is characterizing which predicates $R \subseteq D^r$ have linear/near-linear non-redundancy (or related quantity). The best characterization to date (\cite{bessiere2020Chain,carbonnel2022Redundancy}) are the Mal'tsev predicates (Theorem~\ref{thm:maltsev-nrd}) and derived predicates via gadget reductions (see Section~\ref{subsec:gadget}). This characterization is tight for binary relations ($r=2$) \cite{bessiere2020Chain} as well as for Boolean predicates with $r \le 3$ \cite{chen2020BestCase,khanna2024Characterizations}.

However, multiple works~\cite{lagerkvist2020Sparsification,bessiere2020Chain,khanna2024Characterizations} have identified predicates which fall outside of this characterization for which no nontrivial lower bound is known. Surprisingly, assuming Mal'tsev predicates do characterize near-linear redundancy (and thus up to complement also sparsification), it suffices to resolve all of these examples in the literature by looking at a single ternary predicate (i.e., $r=3$) over a domain $D = \{0,1,2\}$ of size $3$:
\[
  \BCK := \{111,222,012,120,201\},
\]
where tuples are written as concatenated integers for succinctness.
In particular, we know\footnote{This predicate appears in Example 3 of \cite{bessiere2020Chain} as $R := \{001,020,122,202,210\}$, but one can show $\NRD(\BCK, n) = \Theta(\NRD(R, n))$ by applying to each coordinate a suitable bijection from $D$ to itself (see Section~\ref{subsec:gadget}). Explicitly, let $g_1, g_3$ be permutations which swap $0$ and $1$ and let $g_2$ be the cycle $0 \to 1 \to 2 \to 0$. Then, $\BCK = \{(x,y,z) : (g_1(x),g_2(y),g_3(z)) \in R\}$.}  from Example 3 and Observation 12 of \cite{bessiere2020Chain} that one cannot establish that $\BCK$ has near-linear non-redundancy via Theorem~\ref{thm:maltsev-nrd}. In particular, Bessiere,  Carbonnel, and Katsirelos~\cite{bessiere2020Chain}, adapting techniques of \cite{lagerkvist2020Sparsification},  show that every predicate $P \supseteq \BCK$ defineable from Mal'tsev predicates via gadget reductions must also contain the tuple $000$.  That is, the relation
\[
  \BCKp := \{000,111,222,012,120,201\}
\]
is in a strong sense the ``Mal'tsev closure'' of $\BCK$.  It is easy to verify that $\BCKp$ is Mal'tsev as it turns out that
\begin{align}
 \widetilde{\BCKp} :=  \{0,1,3\}^3 \cap \{(x,y,z) \in \F_7^3 : x+2y+4z = 0\} = \{000,111,333,013,130,301\},\label{eq:R-7}
\end{align}
which is identical to $\BCKp$ up to swapping $2$ and $3$.

\paragraph{Booleanization.} Although $\BCK$ lives over domain size $3$, its non-redundancy also has nontrivial implications in the Boolean setting. In particular, consider the following ``Booleanization'' of $\BCK$:
\begin{align*}
  \BCK_{\B} := \{(&\one[x=0], \one[x=1], \one[x=2],\\&\one[y=2], \one[y=0], \one[y=1],\\&\one[z=1], \one[z=2], \one[z=0]) : (x,y,z) \in R\} \subset \{0,1\}^9,\\
    =\ \ \{&010001100,001100010,100001010,010100001,001010100\}
\end{align*}
and one can define $\BCKp_{\B}$ analogously by adding $100010001$. If we view $\{0,1\}^9$ as $3\times 3$ Boolean matrices, then $\BCKp_{\B}$ is precisely the six permutation matrices, and $\BCK_{\B}$ is $\BCKp_{\B}$ with the identity matrix removed.

It is straightforward to observe that $\NRD(\BCK_{\B}, 3n) \ge \NRD(\BCK, n)$. For any non-redundant instance $(X, Y \subset X^3)$ of $\CSP(\BCK)$, let $X' = X \times \{0,1,2\}$ and map any $(x_1,x_2,x_3) \in Y$ to
\[
  ((x_1,0), (x_1,1), (x_1,2), (x_2,2), (x_2,0),(x_2,1), (x_3,1), (x_3,2), (x_3, 0)) \in Y' \subset (X')^3.
\]
Then, $(X', Y')$ must be a non-redundant instance of $\BCK_{\B}$. Thus, proving that $\NRD(\BCK, n) = \omega(n)$, would also show that $\NRD(\BCK_{\B}, n) = \omega(n)$.

The reason this observation is significant is that $\BCK_{\B}$ has appeared three times in the literature. In \cite{khanna2024Characterizations}, the predicate
\[
  P_{\textsc{KPS}} :=  \{000000000, 111111000, 111000111, 110001001, 101010010\}
\]
is constructed as an example for which $\SPR(\overline{P}_{\textsc{KPS}}, n, \eps)$ cannot be near-linear due to their linear code sparsification framework. In particular, any affine predicate containing $P_{\textsc{KPS}}$ must also contain $011100100$. Up to a permutation of the coordinates and flipping bits,\footnote{Explicitly, flip the first three bits of $P_{\textsc{KPS}}$, then permute the coordinates of $P$ according to the cycles $2\to 6\to 3\to 8\to 2$
and $4\to 5\to 7\to 9\to 4$.} it turns out that $P_{\textsc{KPS}}$ and $\BCK_{\B}$ are the same, with the same automorphism sending $011100100$ to the identity matrix in $\BCK^{+}_{\B}$.

Even earlier, \cite{chen2020BestCase} considered the predicate
\[
    P_{\textsc{CJP}} :=  \{100111001,000010011,010110110,000101100,001001111\}.
\]
Here, any affine predicate containing $P_{\textsc{CJP}}$ also has $111111111$. If we negate the last 6 bits and permute the coordinates suitably, we again arrive at $\BCK_{\B}$.

As one last example, \cite{lagerkvist2020Sparsification} considers the following predicate of arity $r = 10$,
\[
  P_{\textsc{LW}} = \{0000000001, 1000100010, 0100011000, 0011000100, 1000010100, 0010101000\}
\]
If one applies the gadget reduction where the last coordinate of $P_{\textsc{LW}}$ is set to $0$, then this reduced version of $P_{\textsc{LW}}$ is equal to $\BCK_{\B}$ up to a permutation of the coordinates.\footnote{In fact, \cite{lagerkvist2020Sparsification}'s combinatorial interpretation of $Q$ inspired our definition of $\BCK_{\B}$.} Thus, $\NRD(P_{\textsc{LW}}, n) \ge \NRD(\BCK_{\B}, n) \ge \NRD(\BCK, n)$. Thus, again, it suffices\footnote{However, to answer the open question in \cite{lagerkvist2020Sparsification}, we would need to show that $\BCK$ has no near-linear kernelization algorithm, however every instance of $\BCK$ is satisfiable by assigning all $1$'s. Such trivialities can be circumvented by considering a variant such as $\widehat{\BCK} := \{120,201,021,210,102\}$.} to prove that $\BCK$ has super-linear non-redundancy to prove that $P_{\textsc{LW}}$ has super-linear non-redundancy.

\begin{remark}
  The reason why $P_{\textsc{LW}}$ has an extra tuple compared to $\BCK_{\B}$ is due to how \cite{lagerkvist2020Sparsification} constructed the predicate. Like \cite{bessiere2020Chain}, they also showed that $P_{\textsc{LW}}$ cannot be a gadget construction from Mal'tsev predicates (in fact, \cite{lagerkvist2017Kernelization} appears to be the first work to develop this technique). However, the style of argument they used (in some sense a ``depth 2'' circuit refutation) was only strong enough to work with $P_{\textsc{LW}}$ but not $\BCK_{\B}$. In contrast, the argument used by \cite{bessiere2020Chain} for $\BCK$ has depth 3. Such circuits are studied in much more detail in \cite{brakensiek2025richness}.
 \end{remark}

\paragraph{Known Bounds for $\NRD(\BCK, n)$.} We now shift our focus to discussing what is currently known about $\NRD(\BCK, n)$. Since $\BCK$ is nontrivial, we know that $\NRD(\BCK, n) = \Omega(n)$, which is still the best known lower bound to date.

For upper bounds, a baseline upper bound is $\NRD(\BCK, n) = O(n^2)$. This can be proved in a few ways (e.g., write out as a quadratic polynomial, or use the techniques of \cite{bessiere2020Chain}), but perhaps the conceptually simplest way uses the conditional non-redundancy framework established in Section~\ref{sec:nrd-mv}. In particular, since $\NRD(\BCKp, n) = O(n)$ by (\ref{eq:R-7}) and Theorem~\ref{thm:maltsev-nrd}, we have that $\NRD(\BCK, n) = \Theta(\NRD(\BCK\mid \BCKp, n))$.

It is straightforward to see that $\NRD(\BCK|\BCKp, n) = O(n^2)$ via the techniques of Proposition~\ref{prop:simple-quadratic-ub}, as any such non-redundant instance cannot have two clauses sharing two variables (or else one is assigned $000$ exactly when the other is). Thus, every pair of variables appears in at most one clause, proving the quadratic upper bound. More strongly, \cite{bessiere2020Chain} essentially showed that any tripartite non-redundant instance of $\CSP(\BCK \mid \BCKp)$ cannot have a $3$-cycle. As such, by invoking a hypergraph Tur{\'a}n bound of S{\'o}s, Erd{\H o}s, and Brown~\cite{sos1973existence}, they could then prove that $\NRD(\BCK|\BCKp, n) = o(n^2)$, however the improvement over quadratic is known to be at best subpolynomial~\cite{MR519318}.

In Section~\ref{sec:nrd-mv}, we indirectly gave a new upper bound for $\NRD(\BCK, n)$. In particular, observe that $\BCK \subsetneq \tLIN_{\Z/3\Z}^*$, but $\BCKp \setminus \BCK = \tLIN_{\Z/3\Z} \setminus \tLIN_{\Z/3\Z}^*$. Thus, any non-redundant instance of $\CSP(\BCK \mid \BCKp)$ is also a non-redundant instance of $\CSP(\tLIN_{\Z/3\Z}^* \mid \tLIN_{\Z/3\Z})$. Therefore, as an immediate corollary of Theorem~\ref{thm:ub-1.6}, we have that $\NRD(\BCK, n) = O(n^{1.6} \log n)$. In fact, one can do slightly better by observing in Lemma~\ref{lem:dimU-upper-bound} that the task of checking whether $(\langle x, u_i\rangle, \langle x, v_i\rangle, \langle x, w_i\rangle) \in \BCK$ for $u_i + v_i + w_i = 0$ can be achieved by the \emph{cubic} polynomial
\[
  P_i(x) := (1- \langle x, u_i\rangle^2)(1 - \langle x, v_i\rangle). 
\]
Leaving the other steps in the proof of Theorem~\ref{thm:ub-1.6} unchanged, we can deduce that $\NRD(\BCK, n) = O(n^{1.5} \log n)$. The authors doubt the exponent of $1.5$ is tight.

\begin{remark}
Subsequent to the initial posting of this work, \cite{brakensiek2025richness} proved a slightly stronger bound $O(n^{1.5})$ by connecting non-redunedant instances of $\CSP(\BCK)$ to the structure of graphs with girth at least $6$. We leave combining these methods to get an even sharper upper bound as a question for future work.
\end{remark}

One last observation is that Theorem~\ref{thm:nrd-3-lin-lb} cannot be applied to $\NRD(\BCK, n)$. Although the instance constructed lacks $3$-cycles, one can check it contains other configurations which make it fail to be non-redundant. However, by Remark~\ref{rem:extension} and using $g_1=g_2=g_3=1$, we deduce that Theorem~\ref{thm:nrd-3-lin-lb} gives an $\Omega(n^{1.5})$ lower bound on the non-redundancy of
\begin{align*}
R := \{111,222,012,120,201,210\} = \BCK \cup \{210\},
\end{align*}
so in some sense we are ``only'' one tuple away from a superlinear lower bound for $\NRD(\BCK, n)$.

\begin{question}
Does there exist $\delta > 0$ such that $\NRD(\BCK, n) = \Omega(n^{1+\delta})$?
\end{question}

\subsection{Modulo Predicates.}

A class of predicates frequently examined in the literature \cite{jansen2019Optimal,lagerkvist2020Sparsification,khanna2024Characterizations} are of the following form
\[
  \LIN_{k,m,S} := \{x \in \{0,1\}^k : \Ham(x) \in S\!\!\!\!\mod m\},
\]
where $k$ and $m$ are positive integers, and nonempty $S \subsetneq \{0,1,\hdots, m-1\}$. Some special cases are understood. For example, if $|S|=1$, then $\NRD(\LIN_{k,m,S},n) = \Theta(n)$ because the predicate is affine. For general $S$, when $m$ is a prime, we know that $\NRD(\LIN_{k,m,S},n) = O(n^{|S|})$ because $\LIN_{k,m,S}$ has a degree-$k$ representation of the form
\[
  \LIN_{k,m,S} = \{0,1\}^k \cap \left\{x \in \F_m^k : \prod_{s \in S}\left(\sum_{i =1}^k x_i - s\right) = 0\right\}.
\]
However, even in the prime case, this upper bound does not appear tight in general, as finding a copy of $\OR_{|S|}$ in $\LIN_{k,m,S}$ seems to depend on the additive combinatorics properties of $S \mod m$ (e.g., whether $S$ is an arithmetic progression).

Even more confounding, the upper bound of $O(n^{|S|})$ does not apply in general when $m$ is non-prime. For instance, if $m = 6$ and $S = \{0,1\}$ \cite{khanna2024Characterizations,lagerkvist2020Sparsification,jansen2019Optimal}, the polynomial $\Ham(x)(\Ham(x)-1)$ has zeros whenever $\Ham(x) \in \{0,1,3,4\}$ because $2\cdot 3 \equiv 0 \mod 6.$ However, by setting $k-2$ bits of $\LIN_{k,6,\{0,1\}}$ to zero and flipping the other two bits, we get a copy of $\OR_2$, so $\NRD(\LIN_{k,6,\{0,1\}}) = \Omega(n^2)$.  Furthermore, $\LIN_{k,6,\{0,1\}}$ has a cubic representation
\[
  \sum_{1 \le i_1 < i_2 < i_3 \le k} x_{i_1}x_{i_2}x_{i_3} - 2\sum_{1\le i_1 < i_2 \le k} x_{i_1}x_{i_2}\equiv 0 \mod 6,
\]
so $\NRD(\LIN_{k,6,\{0,1\}}) = O(n^3)$.
That said, \cite{khanna2024Characterizations} proved $\LIN_{k,6,\{0,1\}}$ cannot similarly define $\OR_3$ and conjectured that $\LIN_{k,6,\{0,1\}}$ fails to have a quadratic representation for all sufficiently large $k$. However, they only ruled out a \emph{symmetric} polynomial representation (see Claim D.2~\cite{khanna2024Characterizations}). We complete the proof here.

\begin{proposition}
For $k \ge 16$, $\LIN_{k,6,\{0,1\}}$ has no quadratic representation over any Abelian group $G$.
\end{proposition}

\begin{proof}
Assume for sake of contradiction that there exists a map $g : \{0,1\}^k \to G$ such that
\[
  g(x) := g_0 + \sum_{i=1}^k g_i \cdot x_i + \sum_{1 \le i < j \le k} g_{ij} \cdot  x_i x_j = \one[x \not\in \LIN_{k,6,\{0,1\}}],
\]
where for any $h \in G$, $h \cdot 0 = 0$ and $h \cdot 1 = h$. In particular, since $g(x) = 0$ whenever $\Ham(x) \in \{0,1\}$, we know that the constant $(g_0)$ and linear ($g_i, i \in [k]$) coefficients of $g(x)$ all equal $0$.

For $S \subseteq [k]$, Let $1_S \in \{0,1\}^k$ be the indicator vector of $S$. Pick $|S| = 16$ and consider all subsets $T \subset S$ of size exactly $6$. Note that $g(1|_{S}) \neq 0$ but $g(1|_{T}) = 0$.  In particular,
\begin{align*}
0 = \sum_{T \in \binom{S}{6}} g(1_T)
    &= \sum_{T \in \binom{S}{6}} \sum_{\{i,j\} \subset T} g_{ij}\\
    &= \sum_{\{i,j\} \subset S} \sum_{T : \{i,j\} \subset T \in \binom{S}{6}} g_{ij}\\
    &= \sum_{\{i,j\}\subset S} \binom{14}{4}g_{ij}\\
    &= 1001 \cdot g(1|_{S}).
\end{align*}
Now, from the theory of combinatorial designs (see Table 1.32 in \cite{colbourn2006Handbook}), for $|S| = 16$, there exists $T_1, \hdots, T_{16} \subseteq S$ of size exactly $6$ such that each pair $\{i,j\} \subseteq S$ is contained in exactly two of the $T_i$'s. By similar logic as before, we have that
\[
0 = \sum_{i=1}^{16} g(1_{T_i}) = \sum_{\{i,j\} \subset S} 2g_{ij} = 2 \cdot g(1|_{S}).
\]
Thus,
\[
  g(1|_{S}) = 1001 \cdot g(1|_{S}) - 500 \cdot 2\cdot g(1|_{S}) = 0,
\]
a contradiction.
\end{proof}

\subsection{Near-polynomial Predicates}

We also briefly mention that predicates like $\tLIN^*_G$ in Section~\ref{sec:nrd-mv} can also be vastly generalized. For simplicity, we consider the case where $G$ is a cyclic group of order a prime $p$, so $G$ can be identified with the field $\F_p$. 

Pick integers $r \ge k \ge 1$ and let $f : \F_p^r \to \F_p$ be a polynomial in which every monomial has at most $k$ distinct variables. Assume also that $f(0^r) = 0$. Now consider the following pair of predicates
\begin{align*}
\POLY_f &:= \{x \in \F_p^r : f(x) = 0\}, \text{ and}\\
\POLY_f^* &:= \POLY_f \setminus \{0^r\}.
\end{align*}
Since computing $f$ only depends on the behavior of tuples of at most $k$ variables, we have that $\NRD(\POLY_f, n) = O(n^k)$ (e.g., \cite{lagerkvist2020Sparsification,khanna2024Characterizations}). However, the behavior of $\NRD(\POLY_f^*, n)$  appears to be much more difficult to control.

As a concrete example, consider $p = 3$ and $f(x) = (x_1+x_2+x_3)(x_1+x_2+x_3-1)$. Observe that for $z_1, z_2 \in \{0,1\}$, we have that $R(z_1,z_2,2)$ is equivalent to $\OR_2(z_1,z_2)$. Thus, $\NRD(\POLY^*_f, n) = \Omega(n^2)$. However, one can show that $\POLY^*_f$ cannot gadget-construct $\OR_3$, so by Theorem~\ref{thm:carbonnel}, we have that $\NRD(\POLY_f^*, n) = O(n^{2.75})$. We conjecture that $\NRD(\POLY^*_f, n) = \Omega(n^{2+\delta})$ for some $\delta > 0$.

\section{Conclusion and Open Questions}\label{sec:conclusion}

In this paper, we closely linked the sparsification of CSPs to combinatorial properties of the corresponding predicate: its non-redundancy (for unweighted sparsification) and its chain length (for weighted sparsification). In addition to the new sparsification bounds we obtained by directly applying known bounds for non-redundancy and chain length, we also developed brand new tools for studying non-redundancy by linking a family of predicates to the structure of matching vector families.

Beyond classifying the non-redundancy and chain length for specific predicates, such as those suggested in Section~\ref{sec:open-nrd}, we now discuss a few directions in which this work could be extended.

\smallskip
\paragraph{Efficient Sparsification for CSPs.} From an algorithmic perspective, the primary open question is making the sparsification procedures in this paper efficient in terms of the number of variables in the CSP instance. Any such procedure cannot be purely from the perspective of non-linear code sparsification, as our proofs (if turned into algorithms) already run in polynomial time with respect to the size of the code.

In Khanna, Putterman, and Sudan's recent line of work on linear code sparsification~\cite{khanna2024Code,khanna2024Characterizations}, their first sparsification method was inefficient (as it required solving the minimum-weight codeword problem), but was later made efficient. However, their algorithms heavily used techniques specialized to the structure of linear codes, such as efficiently finding a spanning set for a vector space and conditioning on certain linear conditions being satisfied (i.e., ``contraction''). Since Mal'tsev predicates behave similarly to systems of linear equations~\cite{bulatov2006Simple}, making a result like Theorem~\ref{thm:maltsev-spr} efficient seems well within reach, but an analogous result for general \textsf{NP}-hard CSPs\footnote{Or more precisely, CSPs which do not have a tight gadget reduction from any efficiently-solvable CSP.} seems much more daunting.

\begin{question}\label{ques:algo}
Can the upper bounds of Theorem~\ref{thm:main} and Theorem~\ref{thm:wspr-csp}  be made efficient?
\end{question}

In fact, an efficient algorithm for Theorem~\ref{thm:main} would imply that every CSP can be kernelized to (approximately) its non-redundancy which would resolve an important open question in the CSP kernelization community (e.g.,~\cite{carbonnel2022Redundancy}).

\smallskip
\paragraph{Non-redundancy Versus Chain Length for CSPs.} Another interesting direction question is to understand the relationship between non-redundancy and chain length for CSPs (see the open problems of \cite{carbonnel2022Redundancy}). Recall by Proposition~\ref{prop:cl-ge-nrd} that chain length is always at least non-redundancy.  However, as discussed in Example~\ref{example:chain}, there are non-linear codes for which non-redundancy and chain length are as distant from each other as possible. Conversely, for linear codes, chain length and non-redundancy are very closely\footnote{For codes over fields, chain length and non-redundancy are equal, but for codes over an Abelian group $G$ they can differ by a factor of $\log_2|G|$. For instance, consider for $G := \Z/2^k\Z$ and let $C \subset G^k$ be the span of $\{2^{k-1}e_1, 2^{k-2}e_1+2^{k-1}e_2, \hdots, e_1+2e_2+\cdots+2^{k-1}e_{k}\}.$ The non-redundancy of $C$ is $1$ but the chain length is $k = \log_2 |G|$.} linked. Since the particular codes which are to be sparsified for a CSP predicate $R$ are highly structured, it remains to be seen what the general relationship between $\CL(R, n)$ and $\NRD(R, n)$ is.

\begin{question}\label{ques:CL-NRD}
For every nontrivial $R \subseteq D^r$, is it the case that $\CL(R, n) \le \NRD(R, n) \cdot \log^{O(1)}(n)$? In fact, can we even show that there exists a universal constant $\alpha \ge 1$ such that $\CL(R, n) \le O(\NRD(R,n)^{\alpha})$?
\end{question}

As discussed in the introduction, the concepts of non-redundancy and chain length appear as the currently best upper and lower bounds on a query complexity problem for CSPs~\cite{bessiere2020Chain}. As such, a positive resolution of Question~\ref{ques:CL-NRD} could lead to a resolution of that question as well.

\smallskip 
\paragraph{Extending to Valued CSPs.}
A natural variant of CSP sparsification that has been studied in the literature (e.g., \cite{kenneth2023cut}) is to assign different weights to each assignment to a CSP constraint. In the CSP literature, this is referred to as a \emph{valued CSP} (VCSP) (e.g., \cite{thapper2015Necessary,thapper2017Power}). One method of abstracting a question like this would be to consider \emph{non-Boolean} code sparsification: we have a code $C \subseteq \Sigma^m$ along with a weight function $w : \Sigma \to \R_{\ge 0}$  so that the weight of a codeword $c \in C$ is $\sum_{i=1}^m w(c_i)$ (or more generally, a separate weight function from some family for each coordinate). To represent an ``unsatisfied'' assignment, we also assume there is $0 \in \Sigma$ with $w(0) = 0$.

As a baseline upper bound, split $C$ into $|\Sigma| - 1$ codes $\{C_{\alpha} : \alpha \in \Sigma \setminus \{0\}\}$ such that each $c \in C$ corresponds to $(\one[c_i = \alpha] : i \in [m])$ in $C_{\alpha}$. Then, we can combine the sparsifiers for each $C_{\alpha}$ to get a sparsifier for $C$. In general, this is a poor bound, as correctly sparsifying the count of each symbol is a much stronger property than estimating their weighted sum. 

Likewise, replacing each nonzero symbol with $1$ to reduce to the Boolean case does not in general yield a tight lower bound. For example, consider $\Sigma = \{0,1,2\}$ and $w(\alpha) = \alpha$. The Booleanization of the code $C := \{1,2\}^m$ gives every codeword a weight of $m$, so it has a trivial $O(1)$ sparsifier. However, one can work out that any $\eps$-sparsifier of $C$ requires $\Omega(m)$ support size. 

\begin{question}
What is the optimal sparsifiability of valued CSPs / non-Boolean codes?
\end{question}

\smallskip
\paragraph{Improved Lower Bounds.} For both unweighted and weighted sparsification, our current lower bounds on $\eps$-sparsification are purely a function of $n$. Can we improve the upper bound as a function of $\eps$? For instance, a lower bound of $n/\eps^2$ is known for cut sparsifiers, which is tight~\cite{andoni2016Sketching,MR3909627}. However, the lower bound techniques seem to crucially use that the non-redundancy/chain length of CUT is linear, so many nearly-isolated clauses can be used in a lower bound. For general CSPs, the optimal relationship between $n$ and $\eps$ is unclear. For instance, if $\NRD(R, n) = \Omega(n^r)$ (i.e., $R$ is $\OR_r$ or $\overline{R}$ is $\AND_r$), then no improvement as a function of $\eps$ is possible as every sparsifier has size at most $n^r$.

\begin{question}
For every $R \subseteq D^r$ and $\eps \in (0,1)$, is it the case that 
\begin{align*}
\Omega(\min(\NRD(R, n)/\eps^2, n^r)) &\le \SPR(R, n, \eps)?, \text{ and}\\
\Omega(\min(\CL(R, n)/\eps^2, n^r)) &\le \wSPR(R, n, \eps)?
\end{align*}
\end{question} 

\smallskip
\paragraph{Average-case Behavior.} As our final area of investigation, it would be interesting to understand the behavior of the non-redundancy (and thus sparsity\footnote{Note that via the reduction from CSP sparsification to code sparsification, there exists a sparsifier of an instance $(X,Y)$ of $\CSP(\overline{R})$ of size approximately the non-redundancy of $(X,Y)$ when viewed as an instance of $\CSP(R)$.}) of random instances of $\CSP(R)$. In particular, the size of the sparsifier may be significantly better than the guarantees of Theorem~\ref{thm:main}.
More formally for a variable set $X$ of size $n$, pick $m \in [n, n^r]$ and sample uniformly at random $m$ clauses $Y$ from $X^r$. What is the non-redundancy? If $R$ is an affine CSP like $\CUT$ or $\tLIN_G$, then it is not hard to see that if $m = \widetilde{\Omega}(n)$, then the space of assignments to the clauses will have rank $\Omega(n)$, so average-case behavior is essentially the same as worst-case behavior. One can make a similar argument for certain ``higher degree'' clauses like $\AND_r$.

However, for a predicate like $\tLIN^*_G$, the non-redundant instance constructed in Theorem~\ref{thm:nrd-3-lin-lb} seems rather ``fragile.'' In particular, a random sample of say $m := n^{1.5}$ clauses seems to not include $\widetilde{\Omega}(n)$ clauses from any copy of the construction in Theorem~\ref{thm:nrd-3-lin-lb}. As such, it could be possible that the non-redundancy of $\tLIN^*_G$ is near-linear in this regime. 

\begin{question}\label{ques:random}
For random instances of $\CSP(R)$, what is the relationship between clause count and non-redundancy?
\end{question}

Subsequent to the initial posting of our work, the intent behind Question~\ref{ques:random} was largely solved by \cite{BGP25} who computed an optimal sparsifier size for random instances of any \emph{valued} CSP. Two important features of this result are (1) the answer is always an integral polynomial (i.e., $\widetilde{O}_{\eps}(n^{k})$ for some integer $k$) and (2) depending on the model the optimal sparsifier size can be \emph{non-monotone} in the number of edges! However, since the sparsification criterion use by \cite{BGP25} was not exactly non-redundancy, Question~\ref{ques:random} is technically still open.

\bibliographystyle{alphaurl}
\bibliography{ref.bib}

\appendix

\end{document}